\pdfoutput=1
\documentclass[aps,pra,twocolumn,superscriptaddress,floatfix,
longbibliography,nofootinbib]{revtex4-2}

\usepackage{amsmath,amssymb,amsthm,bm,mathtools}
\usepackage{graphicx}
\usepackage{xcolor}
\usepackage{microtype}
\usepackage{multirow}
\usepackage{booktabs}
\usepackage{pifont}
\usepackage{tikz}
\usepackage{enumitem}
\usepackage[colorlinks=true,allcolors=blue!55!black]{hyperref}
\hypersetup{linkcolor=red,citecolor=red,urlcolor=red}

\newtheorem{definition}{Definition}
\newtheorem{theorem}{Theorem}
\newtheorem{proposition}{Proposition}
\newtheorem{corollary}{Corollary}
\newtheorem{lemma}{Lemma}
\newtheorem{example}{Example}

\newcommand{\Gcut}{\Gamma_{\mathrm{cut}}}
\newcommand{\Gsep}{\Gamma_{\mathrm{sep}}}
\newcommand{\Tr}{\operatorname{Tr}}
\newcommand{\poly}{\operatorname{poly}}
\newcommand{\supp}{\operatorname{supp}}
\newcommand{\Ehalf}{S_{1/2}}

\definecolor{colCut}{HTML}{0F62FE}          %
\definecolor{colSim}{HTML}{8A3FFC}          %
\definecolor{colTrain}{HTML}{198038}        %
\colorlet{colYes}{colTrain}
\definecolor{colNo}{HTML}{DA1E28}           %
\colorlet{colOpen}{colCut}
\definecolor{colTargetFill}{HTML}{FFE8F2}   %
\definecolor{colTargetLine}{HTML}{EE5396}   %
\colorlet{colClassA}{colSim!12}             %
\colorlet{colClassALine}{colSim!80!black}
\colorlet{colClassB}{colCut!12}             %
\colorlet{colClassBLine}{colCut!80!black}
\colorlet{colClassC}{colCut!12}             %
\colorlet{colClassCLine}{colCut!80!black}
\colorlet{colClassDLeft}{colCut!12}         %
\colorlet{colClassDMid}{colSim!12}          %
\colorlet{colClassDRight}{colTrain!12}      %
\colorlet{colClassDLine}{black!40}
\colorlet{colClassLine}{black!40}           %
\definecolor{colTierOne}{HTML}{FFD9BE}      %
\definecolor{colTierTwo}{HTML}{FF832B}      %
\definecolor{colTierThree}{HTML}{BA4E00}    %
\colorlet{colTierOneText}{black}
\colorlet{colTierTwoText}{black}
\colorlet{colTierThreeText}{white}
\colorlet{colCrossingFill}{colCut!18}
\colorlet{colCrossingLine}{colCut!80!black}
\colorlet{colLocalOpFill}{colCrossingFill}
\colorlet{colLocalOpLine}{colCrossingLine}
\colorlet{colUntouchedFill}{black!15}
\colorlet{colUntouchedLine}{black!45}
\definecolor{colSubA}{HTML}{F4F4F4}         %
\definecolor{colSubB}{HTML}{F4F4F4}         %
\colorlet{colBlockFill}{colTargetFill}
\colorlet{colBlockLine}{colTargetLine}
\colorlet{colInterface}{black}
\colorlet{colCutLimit}{colCut}
\colorlet{colWidthLimit}{black!60}
\colorlet{colBudget}{colCut!12}             %
\definecolor{colCrossFill}{HTML}{FCF4D6}    %
\colorlet{colLocal}{colTrain!15}           %
\definecolor{colBand}{HTML}{F4F4F4}         %
\definecolor{colSimTint}{HTML}{F6F2FF}      %
\definecolor{colCutTint}{HTML}{EDF5FF}      %
\colorlet{FLDCOpen}{colCut}
\colorlet{FLDCTheory}{colSim}
\colorlet{FLDCGood}{colTrain}
\colorlet{FLDCBad}{colNo}
\colorlet{FLDCPassHead}{colLocal}
\colorlet{FLDCFailHead}{colCrossFill}
\definecolor{colNeutral}{HTML}{E0E0E0}
\colorlet{FLDCOutsideHead}{colNeutral}
\colorlet{fldcBand}{colBand}
\colorlet{fldcEasyTint}{colNeutral}
\colorlet{fldcTheoryTint}{colSimTint}
\colorlet{fldcWorkTint}{colCutTint}

\providecommand{\fldcyes}{\textcolor{colYes}{\ding{51}}}
\providecommand{\fldcno}{\textcolor{colNo}{\ding{55}}}
\providecommand{\fldcopen}{\textcolor{colOpen}{\textbf{?}}}
\providecommand{\fldcna}{\textcolor{black!55}{--}}
\providecommand{\tA}{\colorbox{colTierOne}{\small\textbf{\textcolor{colTierOneText}{1}}}}
\providecommand{\tB}{\colorbox{colTierTwo}{\small\textbf{\textcolor{colTierTwoText}{2}}}}
\providecommand{\tC}{\colorbox{colTierThree}{\small\textbf{\textcolor{colTierThreeText}{3}}}}
\providecommand{\kmax}{\kappa_{\max}}
\providecommand{\kstar}{\kappa^{*}}
\providecommand{\auditblock}[1]{%
  \addlinespace[2pt]%
  \multicolumn{11}{@{}l@{}}{%
    \begingroup\setlength{\fboxsep}{2.5pt}%
    \colorbox{fldcBand}{\parbox{\dimexpr\linewidth-2\fboxsep\relax}{\strut\textbf{#1}\strut}}%
    \endgroup}\\\addlinespace[\tabIIblockgap]}
\providecommand{\auditfamily}[1]{%
  & \multicolumn{10}{l@{}}{\textbf{#1}}\tabularnewline}
\providecommand{\auditfocus}[2]{%
  \begingroup\setlength{\fboxsep}{2pt}%
  \colorbox{#1}{\parbox[t]{\dimexpr\linewidth-2\fboxsep\relax}{\raggedright\strut#2\strut}}%
  \endgroup}
\providecommand{\auditlegend}[2]{%
  \begingroup\setlength{\fboxsep}{2pt}\colorbox{#1}{\strut #2}\endgroup}
\providecommand{\auditfocusframed}[3]{%
  \begingroup\setlength{\fboxsep}{1.6pt}\setlength{\fboxrule}{0.4pt}%
  \fcolorbox{#1}{#2}{\parbox[t]{\dimexpr\linewidth-2\fboxsep-2\fboxrule\relax}{\raggedright\strut#3\strut}}%
  \endgroup}
\providecommand{\auditlegendframed}[3]{%
  \begingroup\setlength{\fboxsep}{1.6pt}\setlength{\fboxrule}{0.4pt}\fcolorbox{#1}{#2}{\strut #3}\endgroup}
\providecommand{\classbox}[2]{%
  \begingroup\setlength{\fboxsep}{2pt}%
  \colorbox{#1}{\parbox[t]{\dimexpr\linewidth-2\fboxsep\relax}{\raggedright\strut\textbf{#2}\strut}}%
  \endgroup}

\newcommand{\figMap}{%
\begin{figure*}[t]
\centering
\resizebox{\textwidth}{!}{%
\begin{tikzpicture}[x=1cm,y=1cm,>=stealth,font=\footnotesize,
  hd/.style={font=\footnotesize\bfseries,anchor=south west,inner sep=0pt,align=flush left},
  req/.style={draw=##1!80!black,fill=##1!12,rounded corners=2pt,text width=2.5cm,minimum height=1.3cm,inner sep=3.5pt,align=flush left,anchor=north west},
  crit/.style={draw=##1!80!black,fill=white,rounded corners=2pt,text width=4.2cm,minimum height=1.3cm,inner sep=3.5pt,align=flush left,anchor=north west},
  cls/.style={draw=colClassLine,fill=##1,rounded corners=2pt,text width=4.2cm,inner sep=3.5pt,align=flush left,anchor=north west},
  arr/.style={->,thick,black!55}]
  \coordinate (bus) at (7.9,0);
  \node[hd] at (0,0.12) {Requirements\\[-1pt]\mdseries Sec.~\ref{sec:preliminaries}};
  \node[hd] at (3.15,0.12) {Screening criteria\\[-1pt]\mdseries Table~\ref{tab:properties}};
  \node[hd] at (8.15,0.12) {eighteen circuit families\\[-1pt]\mdseries Tables~\ref{tab:families} and~\ref{tab:classes}};
  \node[hd] at (13.0,0.12) {Single-bridge application\\[-1pt]\mdseries Sec.~\ref{sec:application}};
  \node[req=FLDCOpen] (cut) at (0,0) {\textsc{Cut}\\ affordable knitting,\\ $\Gcut\le\poly(n)$};
  \node[req=FLDCTheory] (sim) at (0,-1.45) {\textsc{Sim}\\ hardness survives\\ the cut};
  \node[req=FLDCGood] (train) at (0,-2.9) {\textsc{Train}\\ resolvable\\ gradients};
  \node[crit=FLDCOpen] (ccut) at (3.15,0) {C1\ \ interface $\kappa$ and norm $\Gcut$\\ C2\ \ seam entropy $\Ehalf$};
  \node[crit=FLDCTheory] (csim) at (3.15,-1.45) {S1\ \ hard branch tasks\\ S2\ \ light cones\\ S3\ \ task type};
  \node[crit=FLDCGood] (ctrain) at (3.15,-2.9) {T1\ \ local depth\\ T2\ \ guarantee $g$\\ T3\ \ gradient variance};
  \draw[arr] (cut.east) -- (ccut.west);
  \draw[arr] (sim.east) -- (csim.west);
  \draw[arr] (train.east) -- (ctrain.west);
  \node[draw=black!40,fill=black!4,rounded corners=2pt,text width=7.35cm,inner sep=3.5pt,align=flush left,anchor=north west] (thm)
    at ([yshift=-0.25cm]train.south west)
    {\textbf{General results}\\
     Prop.~\ref{prop:cost-entropy}\ \ seam entropy sets a floor on $\Gcut$\\
     Thm.~\ref{prop:stitching}\ \ knitting easy branch tasks is easy\\
     Thm.~\ref{prop:regimes}\ \ growing subcircuits, bounded interface};
  \node[cls=colTargetFill,draw=colTargetLine,very thick] (target) at (8.15,0)
    {\textbf{Cheaply cuttable, still hard}\\ row 11c, and 12b if S1 holds};
  \node[cls=colClassA,draw=colClassALine] (easy) at ([yshift=-0.2cm]target.south west) {\textbf{Cuttable, no advantage}\\ rows 7, 8a, and 11b};
  \node[cls=colClassB,draw=colClassBLine] (nocut) at ([yshift=-0.2cm]easy.south west) {\textbf{Advantage, not cuttable}\\ row 11a};
  \node[cls=colClassC,draw=colClassCLine] (costly) at ([yshift=-0.2cm]nocut.south west) {\textbf{Optimally cut, unaffordable}\\ rows 8b and 12a};
  \node[cls=white,draw=colClassDLine,left color=colClassDLeft,right color=colClassDRight,middle color=colClassDMid] (struct) at ([yshift=-0.2cm]costly.south west) {\textbf{Structural failure}\\ rows 1--6};
  \draw[thick,black!55] (target.north west -| bus) -- (struct.south west -| bus);
  \draw[arr] (csim.east) -- (csim.east -| bus);
  \node[draw=colTargetLine,very thick,fill=colTargetFill,rounded corners=2pt,text width=4.25cm,inner sep=3.5pt,align=flush left,anchor=north west] (app) at (13.0,0)
    {\textbf{Two toric-code patches}\\[2pt]
     \setlength{\tabcolsep}{0pt}%
     \begin{tabular}{@{}l@{\hspace{0.8em}}l@{\hspace{0.8em}}l@{}}
     C1 & \fldcyes & \parbox[t]{2.95cm}{\raggedright cost flat in $N$}\\
     C2 & \fldcyes & \parbox[t]{2.95cm}{\raggedright $\Ehalf\le1$ bit}\\
     T1 & \fldcyes & \parbox[t]{2.95cm}{\raggedright bounded local depth}\\
     T3 & \fldcyes & \parbox[t]{2.95cm}{\raggedright no exponential decay at tested sizes}\\[2pt]
     S1 & \fldcopen & \parbox[t]{2.95cm}{\raggedright hardness open}
     \end{tabular}};
  \node[draw=colTargetLine,dashed,thick,fill=white,rounded corners=2pt,text width=4.25cm,inner sep=3.5pt,align=flush left,anchor=north west] (hw)
    at ([yshift=-0.35cm]app.south west)
    {\textbf{Hardware}, IBM Nighthawk\\ contact correlation up to $N=98$\\ knitted energies up to $N=142$};
  \draw[arr] (target.east) -- (target.east -| app.west);
  \draw[arr] (app.south) -- (hw.north);
\end{tikzpicture}}
\caption{How three requirements for distributed quantum advantage sort eighteen circuit families.
The three requirements for distributed quantum advantage (left) become eight screening criteria, and three general results underpin the knitting and hardness criteria (bottom).
Applied to eighteen circuit families, the criteria sort them into classes.
Every class except the target is colored by the requirement it fails, with the structural failures shaded from \textsc{Train} to \textsc{Cut}, the two they fail, where a structural failure means a missing guarantee rather than a proven absence of advantage.
Only the highlighted class combines a cheap cut with hardness inside a subcircuit, and it is the class that distributed finite local-depth circuits (FLDCs) aim for.
The single-bridge application meets the knitting and trainability criteria at the tested sizes, where its branch tasks are still classically tractable, so the hardness criterion S1 remains open beyond them.
The dashed box summarizes the device runs.}
\label{fig:map}
\end{figure*}%
}

\newcommand{\figPartition}{%
\begin{figure*}[t]
\centering
\resizebox{0.98\textwidth}{!}{%
\begin{tikzpicture}[x=0.52cm,y=0.5cm,font=\small,
  lg/.style={draw=colUntouchedLine,fill=colUntouchedFill,rounded corners=1pt,minimum width=0.42cm,minimum height=0.36cm,inner sep=0pt},
  cg/.style={draw=colCrossingLine,fill=colCrossingFill,rounded corners=1pt,inner sep=0pt},
  fb/.style={draw=colLocalOpLine,fill=colLocalOpFill,rounded corners=1pt,minimum width=0.62cm,minimum height=0.40cm,inner sep=1pt},
  blk/.style={draw=colBlockLine,fill=colBlockFill,rounded corners=1pt,inner sep=0pt},
  wire/.style={black!70},
  cut/.style={dashed,thick,FLDCBad}]
\begin{scope}
  \node[font=\bfseries] at (-0.9,1.0) {(a)};
  \foreach \y in {0,-1,-2.6,-3.6} \draw[wire] (0,\y) -- (7.2,\y);
  \fill[colSubA] (-0.1,0.45) rectangle (7.3,-1.45);
  \fill[colSubB] (-0.1,-2.15) rectangle (7.3,-4.05);
  \foreach \y in {0,-1,-2.6,-3.6} \draw[wire] (0,\y) -- (7.2,\y);
  \node[left] at (-0.1,-0.5) {$A$};
  \node[left] at (-0.1,-3.1) {$B$};
  \draw[cut] (-0.3,-1.8) -- (7.5,-1.8) node[right] {$\Pi$};
  \node[lg,minimum height=0.95cm] at (0.8,-0.5) {};
  \node[lg,minimum height=0.95cm] at (0.8,-3.1) {};
  \node[lg] at (3.4,0) {};
  \node[lg] at (3.4,-3.6) {};
  \node[lg,minimum height=0.95cm] at (5.6,-0.5) {};
  \node[lg,minimum height=0.95cm] at (5.6,-3.1) {};
  \draw[cg] (1.9,-0.7) rectangle (2.7,-2.9); \node[fill=colCrossingFill,inner sep=1pt] at (2.3,-1.8) {$\mathcal V_1$};
  \draw[cg] (4.1,-0.7) rectangle (4.9,-2.9); \node[fill=colCrossingFill,inner sep=1pt] at (4.5,-1.8) {$\mathcal V_2$};
\end{scope}
\begin{scope}[shift={(9.9,0)}]
  \node[font=\bfseries] at (-0.9,1.0) {(b)};
  \foreach \y in {-1,-2.6} \draw[wire] (0,\y) -- (1.6,\y);
  \draw[cg] (0.4,-0.7) rectangle (1.2,-2.9); \node[fill=colCrossingFill,inner sep=1pt] at (0.8,-1.8) {$\mathcal V_s$};
  \node at (2.95,-1.8) {$=\ \displaystyle\sum_i a_{s,i}$};
  \begin{scope}[shift={(4.5,0)}]
    \foreach \y in {-1,-2.6} \draw[wire] (0,\y) -- (1.8,\y);
    \draw[cut] (-0.2,-1.8) -- (2.0,-1.8);
    \node[fb] at (0.9,-1) {$\mathcal F^{A}_{s,i}$};
    \node[fb] at (0.9,-2.6) {$\mathcal F^{B}_{s,i}$};
  \end{scope}
\end{scope}
\begin{scope}[shift={(18.7,0)}]
  \node[font=\bfseries] at (-0.9,1.0) {(c)};
  \node at (-0.3,-1.8) {$c_{\bm\alpha}$};
  \begin{scope}[shift={(0.9,0)}]
  \fill[colSubA] (-0.1,0.45) rectangle (7.3,-1.45);
  \fill[colSubB] (-0.1,-2.15) rectangle (7.3,-4.05);
  \foreach \y in {0,-1,-2.6,-3.6} \draw[wire] (0,\y) -- (7.2,\y);
  \node[lg,minimum height=0.95cm] at (0.8,-0.5) {};
  \node[lg,minimum height=0.95cm] at (0.8,-3.1) {};
  \node[lg] at (3.4,0) {};
  \node[lg] at (3.4,-3.6) {};
  \node[lg,minimum height=0.95cm] at (5.6,-0.5) {};
  \node[lg,minimum height=0.95cm] at (5.6,-3.1) {};
  \node[fb,minimum width=0.5cm] at (2.3,-1) {$i_1$};
  \node[fb,minimum width=0.5cm] at (2.3,-2.6) {$i_1$};
  \node[fb,minimum width=0.5cm] at (4.5,-1) {$i_2$};
  \node[fb,minimum width=0.5cm] at (4.5,-2.6) {$i_2$};
  \draw[cut,draw=FLDCBad!50] (-0.3,-1.8) -- (7.5,-1.8);
  \end{scope}
\end{scope}
\begin{scope}[shift={(29.5,0)}]
  \node[font=\bfseries] at (-0.9,1.0) {(d)};
  \foreach \y in {0.2,-0.8,-1.8,-2.8,-3.8} \draw[wire] (0,\y) -- (6.0,\y);
  \foreach \kk/\xx/\yy in {1/0.8/0.2,2/2.1/-0.8,3/3.4/-1.8,4/4.7/-2.8}
    \node[blk,minimum width=0.6cm,minimum height=0.78cm] at (\xx,\yy-0.5) {\scriptsize$U_{\kk}$};
  \node[font=\scriptsize,anchor=south] at (6.6,0.55) {blocks};
  \foreach \cc/\yy in {1/0.2,2/-0.8,2/-1.8,2/-2.8,1/-3.8} \node[font=\scriptsize] at (6.6,\yy) {\cc};
\end{scope}
\end{tikzpicture}
}
\caption{Circuit cutting turns crossing gates into sums of local operations, and finite local depth bounds the blocks acting on each qubit. (a) A partition $\Pi$ assigns the qubits of a circuit to two subcircuits $A$ and $B$. Gates acting within one subcircuit (gray) are untouched, and the gates $\mathcal V_s$ that cross $\Pi$ (blue) are removed. (b) Each crossing gate is written as a signed sum of products of operations local to $A$ and to $B$, with real coefficients $a_{s,i}$. (c) A branch $\bm\alpha=(i_1,i_2)$ chooses one term for every crossing gate (blue). The two subcircuits are then independent, and the original expectation value is the sum over branches weighted by $c_{\bm\alpha}=\prod_s a_{s,i_s}$. (d) A finite local-depth circuit in the sense of Definition~\ref{def:fldc}, here one sequential sweep of blocks $U_k$ on neighboring spins. Its depth grows with the number of spins, but every block acts on $\beta=2$ spins and every spin is acted on by at most $\chi=2$ blocks, as the counts on the right show.}
\label{fig:partition}
\end{figure*}%
}

\newcommand{\tableProperties}{%
\begin{table*}[t]
\centering
\small
\setlength{\tabcolsep}{5pt}
\renewcommand{\arraystretch}{1.25}
\begin{tabular}{@{}c l p{0.26\textwidth} c p{0.42\textwidth}@{}}
\toprule
& \textbf{Criterion} & \textbf{Quantity or task} & \textbf{Tier} & \textbf{What a pass or a failure shows}\\
\midrule
\multicolumn{5}{@{}l}{\textbf{Is the knitting affordable?} (\textsc{Cut})}\\
\midrule
C1 & interface capacity
 & $\kappa(\Pi)$, bounds on $\kstar(n)$, and $\Gcut$ against $\Gamma_{\max}$
 & \tA
 & A partition within budget establishes feasibility. A large heuristic cut does not show that every partition is unaffordable.\\
C2 & seam entropy
 & $\Ehalf(\rho_A)$ and the cost floor $\Gsep$
 & \tB
 & A large floor excludes exact knitting at the assessed partition. A small floor does not certify affordability.\\
\midrule
\multicolumn{5}{@{}l}{\textbf{Does classical difficulty survive the cut?} (\textsc{Sim})}\\
\midrule
S1 & where hardness lives
 & Branch tasks of each subcircuit
 & \tC
 & Efficient branch tasks together with efficient knitting exclude an advantage for the assessed task. Hardness must survive in the branch tasks.\\
S2 & light cone
 & Backward causal cones of the required observables
 & \tA
 & Bounded cones permit direct simulation. Growing cones do not establish hardness.\\
S3 & task type
 & Task and accuracy covered by the supporting results
 & \tC
 & Results must cover local expectation values. Hardness of sampling alone is insufficient.\\
\midrule
\multicolumn{5}{@{}l}{\textbf{Are gradients resolvable?} (\textsc{Train}, variational families only)}\\
\midrule
T1 & local structure
 & $\beta$ and $\chi$ for a stated block decomposition
 & \tA
 & Bounded values identify an FLDC. A failure is not a general untrainability result.\\
T2 & guarantee constant
 & Applicability, scaling, and magnitude of $g$
 & \tA
 & A usable bound supports gradient resolution. A small or inapplicable bound leaves the variance undecided.\\
T3 & gradient variance
 & $\operatorname{Var}_{\boldsymbol\theta}[\partial_\mu C]$ against $n$
 & \tB
 & Numerical scaling with uncertainty for a stated ensemble. Finite-size evidence is not an asymptotic guarantee.\\
\bottomrule
\end{tabular}
\caption{The screening criteria. The tier records the kind of evidence a criterion needs. \tA\ structural checks computed from the circuit generator. \tB\ numerical diagnostics at sizes reachable by classical simulation. \tC\ theoretical obligations met by a proof or a citation. Outcomes refer to the specified task, partition, parameter regime, and resource budget. For non-variational families the \textsc{Train} criteria do not apply.}
\label{tab:properties}
\end{table*}%
}

\newcommand{\figCrossing}{%
\begin{figure}[tb]
\centering
\resizebox{\linewidth}{!}{%
\begin{tikzpicture}[x=0.056cm,y=0.165cm,>=stealth,font=\small]
  \fill[colBudget] (0,0) rectangle (100,5.9);
  \draw[->] (0,0) -- (168,0) node[below=2pt] {$n$};
  \draw[->] (0,0) -- (0,30) node[above] {gates severed};
  \draw[dashed,thick,colCutLimit] (0,5.9) -- (168,5.9);
  \node[above left,colCutLimit] at (166,5.9) {$\kmax$};
  \draw[dashed,thick,colWidthLimit] (100,0) -- (100,28);
  \node[above,colWidthLimit] at (100,28) {$n_{\mathrm{width}}$};
  \draw[colWidthLimit] (100,0) -- (100,-0.8) node[below] {$Dw$};
  \draw[very thick,colInterface,domain=0.4:160,samples=90] plot (\x,{2*sqrt(\x)});
  \node[colInterface,anchor=east] at (160,27.5) {$\kstar=\Theta(\sqrt n)$};
  \draw[very thick,colInterface] (0,2) -- (155,2);
  \node[colInterface,above right] at (12,2.3) {$\kstar=O(1)$};
  \fill[colCutLimit] (8.7,5.9) circle (2.4pt);
  \draw[colCutLimit,thin] (8.7,5.9) -- (14,19);
  \node[colCutLimit,above right] at (13,19) {cut limited};
  \fill[colWidthLimit] (100,2) circle (2.4pt);
  \draw[colWidthLimit,thin] (100,2) -- (110,10);
  \node[colWidthLimit,above right] at (109,10) {width limited};
\end{tikzpicture}}
\caption{A bounded interface leaves the affordable region only through the device width. The budget fixes a threshold $\kmax$ on the number of severed gates, independent of the circuit, while the circuit family fixes an interface curve $\kstar(n)$, independent of the budget. The shaded region is both affordable and within device width. A bounded interface leaves it at $n_{\mathrm{width}}$, so hardware sets the reach. A growing interface leaves it through the budget threshold first unless the device width is reached earlier. The curves show minimum interfaces, whereas a heuristic search supplies upper bounds, and known gate angles should be assessed with the actual norm.}
\label{fig:crossing}
\end{figure}%
}

\def\tabIIfontsize{\fontsize{8.5pt}{9.5pt}\selectfont}
\def\tabIIstretch{1.0}           %
\def\tabIIblockgap{1pt}          %
\def\tabIIfamgap{2.5pt}            %
\def\tabIIsubgap{3pt}            %
\def\tabIIcolsep{1pt}           %
\def\tabIIwNo{1.0em}             %
\def\tabIInamewidth{0.2\textwidth} %
\def\tabIIwTone{5.8em}   %
\def\tabIIwTtwo{4.4em}   %
\def\tabIIwStwo{2em}     %
\def\tabIIwCone{3.1em}   %
\def\tabIIwTthree{2.5em} %
\def\tabIIwCtwo{3.8em}   %
\def\tabIIwSone{3.0em}   %
\def\tabIIwSthree{3.0em} %
\def\tabIIverdictwidth{\dimexpr\textwidth-\tabIIwNo-\tabIInamewidth
  -\tabIIwTone-\tabIIwTtwo-\tabIIwStwo-\tabIIwCone-\tabIIwTthree-\tabIIwCtwo
  -\tabIIwSone-\tabIIwSthree-20\tabcolsep\relax} %
\def\tabIIjustify{\rightskip=0pt\relax\parfillskip=0pt plus 1fil\relax}
\def\tabIIverdictalign{\tabIIjustify}
\def\famend{\tabularnewline\addlinespace[\tabIIfamgap]}
\def\subend{\tabularnewline\addlinespace[\tabIIsubgap]}

\newcommand{\tableCircuits}{%
\begin{table*}[tp]
\centering
\caption{\label{tab:families}%
Circuit families and regimes assessed against the criteria of Table~\ref{tab:properties}.
Tier-2 entries are evidence at the sizes tested, not proof.
Daggered entries are subject to limitations of the specific task considered, and are discussed in Sec.~\ref{sec:families}.
\fldcyes\ passes the stated criterion, \fldcno\ fails it, \fldcopen\ is unresolved, and \fldcna\ is not applicable.
For non-variational families, $(\beta,\chi)$ are reported as structural descriptors rather than as a trainability assessment.
Highlights mark promising cases:
\auditlegend{colTargetFill}{bounded-interface theoretical candidate},
\auditlegendframed{colTargetLine}{colTargetFill}{application studied here}.
Tier numbers as in Table~\ref{tab:properties}.}
\tabIIfontsize
\setlength{\tabcolsep}{\tabIIcolsep}
\renewcommand{\arraystretch}{\tabIIstretch}
\setlength{\parindent}{0pt}
\setlength{\parskip}{0pt}
\resizebox{\textwidth}{!}{%
\begin{tabular}{@{}
  l p{\tabIInamewidth}
  c c c c
  c c
  c c
  p{\tabIIverdictwidth}@{}}
\toprule
\multirow{2}{*}{\makebox[\tabIIwNo][l]{No.}}
& \multirow{2}{*}{Circuit family/regime}
& \multicolumn{4}{c}{\tA\ \textbf{structural}}
& \multicolumn{2}{c}{\tB\ \textbf{numerical}}
& \multicolumn{2}{c}{\tC\ \textbf{theoretical}}
& \multicolumn{1}{c@{}}{\multirow{3}{*}{Verdict}}\\
\cmidrule(lr){3-6}\cmidrule(lr){7-8}\cmidrule(lr){9-10}
& & \makebox[\tabIIwTone]{$(\beta,\chi)$} & \makebox[\tabIIwTtwo]{$g$} & \makebox[\tabIIwStwo]{cone} & \makebox[\tabIIwCone]{$\kstar(n)$}
  & \makebox[\tabIIwTthree]{$\operatorname{Var}$} & \makebox[\tabIIwCtwo]{$\Ehalf$} & \makebox[\tabIIwSone]{hard} & \makebox[\tabIIwSthree]{task} & \\
& & \tiny T1 & \tiny T2 & \tiny S2 & \tiny C1
  & \tiny T3 & \tiny C2 & \tiny S1 & \tiny S3 & \\
\midrule
\auditblock{A. Conventional variational circuits and scrambling benchmarks}
\auditfamily{Hubbard variational circuits:}
1a & \raggedright (a) EHV ansatz~\cite{Stanisic2022Hubbard,DerbyKlassen2021,Cade2020Hubbard}
  & $(2,\Theta(n))$ & ${\le}10^{-43}$ & \fldcyes & \fldcopen
  & \fldcopen & \fldcopen & \fldcno & \fldcyes
  & \tabIIverdictalign  No FLDC guarantee. Preparation and seam both grow.\subend
1b & \raggedright (b) Hamiltonian variational ansatz on a Hubbard ladder~\cite{Wecker2015HVA,Cade2020Hubbard}
  & $(2,\omega(1))$ & ${\le}10^{-168}$ & \fldcyes & $\omega(1)$
  & \fldcopen & $O(1)$ & \fldcno & \fldcyes
  & \tabIIverdictalign  Low seam entropy, but classically easy: DMRG solves the ladder.\textsuperscript{$\dagger$}\famend
2 & \raggedright $J_1$--$J_2$ Heisenberg ground-state ansatz~\cite{Feulner2022VQE}
  & $(2,\Theta(\sqrt n))$ & ${\le}10^{-86}$ & \fldcyes & $\Theta(\sqrt n)$
  & \fldcopen & $\Theta(\sqrt n)$ & \fldcopen & \fldcyes
  & \tabIIverdictalign No trainability guarantee, and not cuttable.\famend
3 & \raggedright Hardware-efficient ansatz for molecular VQE~\cite{Kandala2017HEA,Cerezo2025Simulability,McClean2018BarrenPlateaus,Leone2024HEA}
  & $(2,\Theta(d))$ & ${\le}10^{-28}$ & \fldcyes & $\omega(1)$
  & \fldcno & \fldcopen & \fldcno & \fldcyes
  & \tabIIverdictalign Gradient suppression in the assessed global-cost regime. \famend
4 & \raggedright ADAPT-VQE with a generalized operator pool~\cite{Grimsley2019ADAPT,Tang2021QubitADAPT,Grimsley2023ADAPTBP}
  & $(\Theta(n),\omega(1))$ & ${\le}10^{-2167}$ & \fldcyes & $\omega(1)$
  & \fldcopen & \fldcopen & \fldcno & \fldcyes
  & \tabIIverdictalign No FLDC guarantee. Tested molecules are exactly solvable.  \famend
5 & \raggedright Routed or dense QAOA~\cite{Harrigan2021QAOA,Farhi2014QAOA,Weidenfeller2022Scaling,Sack2024LargeScaleQAOA,Blekos2024QAOAReview}
  & $(2,\omega(1))$ & $4^{-\omega(1)}$ & \fldcopen & $\Theta(n)$
  & \fldcopen & \fldcopen & \fldcopen & \fldcyes
  & \tabIIverdictalign Growing local depth; branch complexity unresolved. \famend
\auditfamily{Circuits hard only to sample:}
6a & \raggedright (a) deep or scrambling circuits~\cite{Cerezo2025Simulability,Arute2019RCS, Boixo2018Supremacy,Morvan2024RCS,Angrisani2024Classical}
  & $(2,\omega(1))$ & $4^{-\omega(1)}$ & \fldcyes & $\Theta(\sqrt n)$
  & \fldcno & $\Theta(n)$ & \fldcyes & \fldcno
  & \tabIIverdictalign  Hard only to sample, not for expectation values. Volume-law seam, not cuttable.\textsuperscript{$\dagger$}\subend
6b & \raggedright\leavevmode(b) fermionic brick-wall with $k$ particles~\cite{Kerenidis2026Scalable}
  & $(2,\Theta(n))$ & \fldcna & \fldcyes & $\Theta(n)$
  & \fldcyes & \fldcopen & \fldcyes & \fldcno
  & \tabIIverdictalign\leavevmode  Hard only to sample, and the interface grows.\textsuperscript{$\dagger$}\famend

\auditblock{B. Sequential preparation and calibration examples}
7 & \raggedright Sequential MPS preparation~\cite{Schon2005SequentialMPS,Iaconis2024MPS,Barratt2021Holographic,Malz2024LogDepth,Smith2024ConstantDepth}
  & $(\log D,O(1))$ & $D^{-O(1)}$ & \fldcyes & $O(1)$
  & \fldcyes & $O(1)$ & \fldcno & \fldcyes
  & \tabIIverdictalign Cuttable, no advantage.\famend
\auditfamily{Clifford toric code:}
8a & \raggedright (a) bounded contact~\cite{Kitaev2003Anyons,Dennis2002Memory,Liu2022StringNet}
  & $(2,{\le}6)$ & \fldcna & \fldcyes & $O(1)$
  & \fldcna & $O(1)$ & \fldcno & \fldcyes
  & \tabIIverdictalign Cuttable, no advantage.\subend
8b & \raggedright (b) Extensive seam~\cite{Satzinger2021ToricCode,Kitaev2003Anyons,Andersen2023NonAbelian}
  & $(2,{\le}6)$ & \fldcna & \fldcyes & $\Theta(L)$
  & \fldcna & $\Theta(L)$ & \fldcno & \fldcyes
  & \tabIIverdictalign Clifford, so classically easy. Calibrates seam cost $2^L-1$.\famend
9 & \raggedright Sequential $S_3$ quantum double (non-Abelian)~\cite{Lo2026S3,Tantivasadakarn2023Shortest,Iqbal2024NonAbelian}
  & $(6,O(1))$ & \fldcna & \fldcyes & $\Theta(L)$
  & \fldcna & $\Theta(L)$ & \fldcno & \fldcyes
  & \tabIIverdictalign Classically easy: solvable group, braiding not universal. Seam cost $6^{L-1}$.\textsuperscript{$\dagger$}\famend
10 & \raggedright Fibonacci string net (universal braiding)~\cite{Xu2024Fibonacci,Liu2022StringNet}
  & $(5,O(1))$ & \fldcna & \fldcyes & $\Theta(L)$
  & \fldcna & $0.92L$ & \fldcno & \fldcyes
  & \tabIIverdictalign  Classically easy state (braids of length ${\le}3$). Seam cost ${\approx}1.9^L$.\textsuperscript{$\dagger$}\famend
\auditblock{C. Two-dimensional FLDCs with contrasting interfaces and simulability}
\auditfamily{Sequential isoTNS:}
11a & \raggedright (a) transparent seam~\cite{ZaletelPollmann2020isoTNS,MalzTrivedi2025,Soejima2020isoTNS,Wei2022SequentialPEPS}
  & $(\log D,O(1))$ & $D^{-O(1)}$ & \fldcyes & $\Theta(L)$
  & \fldcyes & $\Theta(\sqrt n)$ & \fldcyes & \fldcyes
  & \tabIIverdictalign Hard, but not cuttable: seam entropy grows as $\sqrt n$.\subend
11b & \raggedright (b) Strongly injective ($\eta\ge0.41$), bounded interface~\cite{MalzTrivedi2025,Haghshenas2022TNCircuits,Slattery2021isoTNSCircuits}
  & $(\log D,O(1))$ & $D^{-O(1)}$ & \fldcyes & $O(1)$
  & \fldcopen & $O(1)$ & \fldcno & \fldcyes
  & \tabIIverdictalign Cuttable, no advantage. Noise threatens T3.\subend
\textbf{11c} & \auditfocus{colTargetFill}{\textbf{(c) Weakly injective, bounded interface}~\cite{MalzTrivedi2025,ZaletelPollmann2020isoTNS,Wei2022SequentialPEPS}}
  & $(\log D,O(1))$ & $D^{-O(1)}$ & \fldcyes & $O(1)$
  & \fldcopen & ${\le}\kappa$ & \fldcyes & \fldcyes
  & \tabIIverdictalign \textbf{Candidate.} Worst-case hardness; T3 unresolved. \famend
\auditfamily{Variational toric code:}
12a & \raggedright (a) transparent seam~\cite{ZhangLiuZhang2024FLDC,Satzinger2021ToricCode,Trebst2007Breakdown,Vidal2009LowEnergy,Tupitsyn2010Multicritical,Dusuel2011Robustness}
  & $(2,{\sim}6)$ & ${\sim}10^{-29}$ & \fldcyes & $\Theta(L)$
  & \fldcyes & $L-1$ & \fldcopen & \fldcyes
  & \tabIIverdictalign One code cut along its boundary, unlike 12b. Cost $2^L-1$, so $L\le6$.\subend
\textbf{12b} & \auditfocusframed{colTargetLine}{colTargetFill}{\textbf{(b) Field-perturbed, single bridge (this work)}}
    & $(\mathbf{2},\mathbf{O(1)})$ & \fldcna & \fldcyes & $\mathbf{1}$
  & \fldcyes & ${\le}1$ & \fldcopen & \fldcyes
  & \tabIIverdictalign \textbf{Cheap cut,} but \textbf{S1 open.}\famend
\bottomrule
\end{tabular}%
}%
\end{table*}%
}

\newcommand{\tableClasses}{%
\begin{table*}[t]
\centering
\small
\setlength{\tabcolsep}{5pt}
\renewcommand{\arraystretch}{1.3}
\begin{tabular}{@{}p{0.17\textwidth}p{0.15\textwidth}p{0.49\textwidth}p{0.11\textwidth}@{}}
\toprule
\textbf{Class} & \textbf{Signature} & \textbf{Meaning} & \textbf{Rows}\\
\midrule
\classbox{colClassA}{Cuttable, no advantage}
 & C1 \fldcyes, S1 \fldcno
 & A cheap seam with easy subcircuits. By Theorem~\ref{prop:stitching} the knitted computation is classically easy, so the verdict is structural rather than a suspicion. We expect most bounded-interface proposals to fall here.
 & 7, 8a, 11b\\
\classbox{colClassB}{Advantage, not cuttable}
 & S1 \fldcyes, C2 \fldcno
 & Seam entropy grows with the boundary, so $\Gcut=2^{\Omega(\sqrt n)}$ in two dimensions. By Proposition~\ref{prop:cost-entropy} no exact knitting protocol does better, so the verdict does not reflect a poor decomposition.
 & 11a\\
\classbox{colClassC}{Optimally cut, unaffordable}
 & C2 \fldcno, $\Gcut=\Gsep$
 & The cost is known exactly but exceeds the budget. The transparent toric-code seam costs $2^L-1$, the same as the optimal joint cut of $L-1$ wires with classical communication (App.~\ref{app:cost}). These families calibrate the numerics.
 & 8b, 12a\\
\classbox{colTargetFill}{Cheaply cuttable, still hard}
 & C1, T1, S1 \fldcyes
 & A bounded interface, bounded local depth, and hardness inside a subcircuit hold at once. The three conditions are independent (Sec.~\ref{subsec:independence}). Corollary~\ref{cor:coexistence} shows that the class is not empty.
 & 11c, and 12b if S1 holds\\
\bottomrule
\end{tabular}
\caption{The four classes of screening outcomes. The signature lists the criteria of Table~\ref{tab:properties} that determine each class, and the last column gives examples from Table~\ref{tab:families}. The target of a distributed FLDC is the last class.}
\label{tab:classes}
\end{table*}%
}

\newcommand{\tableKmax}{%
\begin{table}[h]
\centering
\small
\setlength{\tabcolsep}{8pt}
\renewcommand{\arraystretch}{1.2}
\begin{tabular}{@{}lccc@{}}
\toprule
Shot budget $M$ & $10^{6}$ & $10^{9}$ & $10^{12}$\\
\midrule
$\Gamma_{\max}$ & $3.7$ & $116$ & $3682$\\
$\kappa_{\max}^{\mathrm{joint}}$, one time slice & $1$ & $5$ & $10$\\
$\kappa_{\max}^{\mathrm{ind}}$, independent gates & $1$ & $4$ & $7$\\
\bottomrule
\end{tabular}
\caption{Largest affordable number of severed two-qubit rotations from Eq.~\eqref{eq:kmax}, for additive error $\varepsilon=0.01$ and failure probability $\delta=0.05$, using the worst-case norm of each construction as a sufficient norm-budget threshold}.
\label{tab:kmax}
\end{table}%
}

\newcommand{\figExperiment}{%
\begin{figure*}[t]
\centering
\resizebox{\ifdim\width>\textwidth\textwidth\else\width\fi}{!}{%
\begin{tikzpicture}[font=\small,
  box/.style={draw=FLDCOpen!80!black,fill=FLDCOpen!10,rounded corners=1.5pt,minimum height=0.55cm,minimum width=0.8cm,inner sep=2pt},
  pbox/.style={draw=FLDCTheory!90!black,fill=FLDCTheory!18,rounded corners=1.5pt,minimum height=0.55cm,minimum width=0.8cm,inner sep=2pt},
  gbox/.style={draw=black!60,fill=black!5,rounded corners=1.5pt,minimum height=0.55cm,minimum width=0.8cm,inner sep=2pt},
  mbox/.style={draw=black!70,fill=white,minimum height=0.55cm,inner sep=2pt},
  op/.style={draw=FLDCOpen!80!black,fill=white,minimum width=0.72cm,minimum height=0.46cm,inner sep=1pt},
  mz/.style={draw=FLDCTheory!90!black,fill=FLDCTheory!15,minimum width=0.72cm,minimum height=0.46cm,inner sep=1pt},
  paren/.style={black!75,line width=0.6pt},
  lab/.style={font=\bfseries,anchor=south west,inner sep=1pt}]
  \begin{scope}[x=0.72cm,y=0.72cm,local bounding box=pa]
    \begin{scope}
      \fill[black!5,rounded corners=4pt] (-0.35,-0.35) rectangle (3.35,3.35);
      \draw[black!50,thin] (0,0) grid[step=1] (3,3);
      \fill[FLDCTheory!35] (0.4,1.85) rectangle (1.6,2.15);
      \fill[FLDCTheory!35] (0.85,1.4) rectangle (1.15,2.6);
      \foreach \x in {0,1,2} \foreach \y in {0,1,2,3} \fill[FLDCOpen] (\x+0.5,\y) circle (1.9pt);
      \foreach \x in {0,1,2,3} \foreach \y in {0,1,2} \fill[FLDCOpen] (\x,\y+0.5) circle (1.9pt);
      \foreach \p in {(0.5,2),(1.5,2),(1,1.5),(1,2.5)} \draw[FLDCTheory,thick] \p circle (3.2pt);
      \node[FLDCTheory,right] at (1.15,2.45) {$S_v$};
      \node at (1.5,-0.9) {$H_A(h)$};
    \end{scope}
    \begin{scope}[shift={(5.2,0)}]
      \fill[black!5,rounded corners=4pt] (-0.35,-0.35) rectangle (3.35,3.35);
      \draw[black!50,thin] (0,0) grid[step=1] (3,3);
      \foreach \x in {0,1,2} \foreach \y in {0,1,2,3} \fill[FLDCOpen] (\x+0.5,\y) circle (1.9pt);
      \foreach \x in {0,1,2,3} \foreach \y in {0,1,2} \fill[FLDCOpen] (\x,\y+0.5) circle (1.9pt);
      \draw[FLDCOpen,very thick] (1.12,0.12) rectangle (1.88,0.88);
      \node[FLDCOpen] at (1.5,0.5) {$P_f$};
      \node at (1.5,-0.9) {$H_B(h)$};
    \end{scope}
    \draw[FLDCGood,thick,dashed] (3,1.5) -- (5.2,1.5);
    \draw[FLDCGood,thick,fill=white] (3,1.5) circle (3.2pt);
    \draw[FLDCGood,thick,fill=white] (5.2,1.5) circle (3.2pt);
    \node[FLDCGood,above] at (4.1,1.6) {$-tZ_aZ_b$};
    \node[below] at (3.2,1.4) {$a$};
    \node[below] at (5.0,1.45) {$b$};
  \end{scope}
  \begin{scope}[shift={(8.679,1.824)},x=0.62cm,y=0.62cm,local bounding box=pb]
    \foreach \y/\p in {0/A,-2.4/B} {
      \draw (0.4,\y) -- (12.4,\y);
      \draw (0.75,\y-0.15) -- (0.95,\y+0.15);
      \node[left] at (0.4,\y) {$|0\rangle_{\p}$};
      \node[gbox] at (2.0,\y) {$E_{\p}$};
      \node[box] (vpre\p) at (4.0,\y) {$V_{\p}^{\mathrm{pre}}$};
      \node[mbox] at (11.6,\y) {$M_{X}/M_{Z}$};
    }
    \node[box] at (9.4,-2.4) {$V_{B}^{\mathrm{post}}$};
    \draw[draw=FLDCOpen!80!black,fill=FLDCOpen!10,rounded corners=1.5pt] (6.25,-0.55) rectangle (10.25,0.55);
    \node[pbox] at (7.5,0) {$R_{S_v}(\alpha)$};
    \node[inner sep=2pt] at (9.4,0) {$V_{A}^{\mathrm{post}}$};
    \draw[FLDCGood,thick] (5.8,0) -- (5.8,-2.4);
    \fill[FLDCGood] (5.8,0) circle (2.6pt);
    \fill[FLDCGood] (5.8,-2.4) circle (2.6pt);
    \node[draw=FLDCGood,fill=colLocal,rounded corners=1.5pt,inner sep=2pt] at (5.8,-1.2) {$R_{ZZ}(\phi)$};
    \node[FLDCGood,above] at (5.8,0.1) {$a$};
    \node[FLDCGood,below] at (5.8,-2.5) {$b$};
  \end{scope}
  \draw[black!55,dashed,rounded corners=3pt] ([xshift=-3pt,yshift=-3pt]vpreB.south west) rectangle ([xshift=3pt,yshift=3pt]vpreB.north east);
  \begin{scope}[shift={(10.26,-1.40)},x=0.9cm,y=0.9cm,local bounding box=pin]
    \fill[black!4] (0,0) rectangle (2,1);
    \draw[black!55] (0,0) rectangle (2,1);
    \draw[black!55] (1,0) -- (1,1);
    \foreach \dx in {0,1} {
      \begin{scope}[shift={(\dx,0)}]
        \draw[FLDCTheory,line width=1pt,shorten <=2.8pt,shorten >=2.8pt] (0,0.5) .. controls (0.1,0.1) .. (0.5,0);
        \draw[FLDCTheory,line width=1pt,shorten <=2.8pt,shorten >=2.8pt] (0.5,0) .. controls (0.9,0.1) .. (1,0.5);
        \draw[FLDCTheory,line width=1pt,shorten <=2.8pt,shorten >=2.8pt] (1,0.5) .. controls (0.9,0.9) .. (0.5,1);
        \draw[FLDCTheory,line width=1pt,shorten <=2.8pt,shorten >=2.8pt] (0.5,1) .. controls (0.1,0.9) .. (0,0.5);
      \end{scope}}
    \foreach \p in {(0.5,0),(1.5,0),(0.5,1),(1.5,1),(0,0.5),(1,0.5),(2,0.5)} \fill[FLDCOpen] \p circle (1.9pt);
    \draw[FLDCBad,thick] (1,0.5) circle (3.8pt);
    \foreach \dx in {0,1} {
      \begin{scope}[shift={(\dx,0)},every node/.style={font=\tiny,text=FLDCTheory!80!black,inner sep=0pt}]
        \node at (0.30,0.28) {1}; \node at (0.70,0.28) {2};
        \node at (0.70,0.72) {3}; \node at (0.30,0.72) {4};
      \end{scope}}
    \draw[->,black!55,thin] (0.62,-0.17) -- (1.38,-0.17);
    \draw[FLDCBad,thin] (1,0.5) -- (2.35,-0.25);
    \node[right,font=\scriptsize,FLDCBad,inner sep=1pt] at (2.35,-0.25) {$R_Y(\vartheta_1)\,R_Z(\vartheta_3)$};
    \draw[FLDCTheory,thin] (0.2,0.84) -- (-0.3,1.05);
    \node[left,font=\scriptsize,FLDCTheory,inner sep=1pt] at (-0.3,1.05) {$R_{ZZ}(\vartheta_5)$};
  \end{scope}
  \draw[black!45,densely dotted] ([xshift=-3pt,yshift=-3pt]vpreB.south west) -- (10.26,-0.50);
  \draw[black!45,densely dotted] ([xshift=3pt,yshift=-3pt]vpreB.south east) -- (12.06,-0.50);
  \begin{scope}[shift={(0.124,-2.95)},x=0.62cm,y=0.62cm,local bounding box=pc]
    \draw (0,0.7) -- (2.6,0.7);
    \draw (0,-0.7) -- (2.6,-0.7);
    \node[left,font=\scriptsize] at (0,0.7) {$a$};
    \node[left,font=\scriptsize] at (0,-0.7) {$b$};
    \draw[FLDCGood,thick] (1.3,0.7) -- (1.3,-0.7);
    \fill[FLDCGood] (1.3,0.7) circle (2.2pt);
    \fill[FLDCGood] (1.3,-0.7) circle (2.2pt);
    \node[draw=FLDCGood,fill=colLocal,rounded corners=1.5pt,inner sep=1.5pt,font=\footnotesize] at (1.3,0) {$R_{ZZ}(\phi)$};
    \node[anchor=mid] at (3.15,0) {$=$};
    \begin{scope}[shift={(3.55,0)}]
      \node[anchor=mid] at (0.6,0) {$c^2$};
      \draw[paren] (1.2,1.15) .. controls (0.95,0.6) and (0.95,-0.6) .. (1.2,-1.15);
      \draw (1.3,0.7) -- (3.1,0.7);
      \draw (1.3,-0.7) -- (3.1,-0.7);
      \draw[paren] (3.2,1.15) .. controls (3.45,0.6) and (3.45,-0.6) .. (3.2,-1.15);
    \end{scope}
    \foreach \k/\sg/\cf/\ta/\tb/\m in {1/+/s^2/Z/Z/0, 2/+/d/M_Z/R_+/1, 3/-/d/M_Z/R_-/1, 4/+/d/R_+/M_Z/2, 5/-/d/R_-/M_Z/2}{
      \begin{scope}[shift={(3.55+3.9*\k,0)}]
        \node[anchor=mid] at (0,0) {$\sg$};
        \node[anchor=mid] at (0.6,0) {$\cf$};
        \draw[paren] (1.2,1.15) .. controls (0.95,0.6) and (0.95,-0.6) .. (1.2,-1.15);
        \draw (1.3,0.7) -- (3.1,0.7);
        \draw (1.3,-0.7) -- (3.1,-0.7);
        \ifnum\m=1 \node[mz] (m) at (2.2,0.7) {$\ta$}; \draw[->,FLDCTheory] (m.north east) -- ++(0.25,0.22);
        \else \node[op] at (2.2,0.7) {$\ta$}; \fi
        \ifnum\m=2 \node[mz] (m) at (2.2,-0.7) {$\tb$}; \draw[->,FLDCTheory] (m.north east) -- ++(0.25,0.22);
        \else \node[op] at (2.2,-0.7) {$\tb$}; \fi
        \draw[paren] (3.2,1.15) .. controls (3.45,0.6) and (3.45,-0.6) .. (3.2,-1.15);
      \end{scope}}
  \end{scope}
  \node[lab] at (pa.north west) {(a)};
  \node[lab] at (pb.west |- pa.north) {(b)};
  \node[lab] at (pc.north west) {(c)};
\end{tikzpicture}}
\caption{Two toric-code patches joined by one variational gate, the sweeps that prepare their states, and the six branches that knit them.
(a)~Two open toric-code patches with spins on the edges (dots).
A star $S_v$ and a plaquette $P_f$ are highlighted, and the only cross-patch term of the Hamiltonian couples the facing exterior spins $a$ and $b$.
(b)~Circuit order.
Fixed encoders $E_p$ prepare the zero-field ground states, the prefix sweeps $V_p^{\mathrm{pre}}$ act before the bridge $R_{ZZ}(\phi)$, and the anchor rotation $R_{S_v}(\alpha)$ and the suffix sweeps act after it.
The inset shows one sweep on two neighboring faces.
In each face it applies a pair block to each of the four pairs of consecutive edges (arcs) in the order of the numbers, and it visits the faces row by row (arrow), so a spin on an edge shared by two faces (ring) enters four pair blocks per sweep.
A pair block rotates both spins by $R_Y$ and $R_Z$ and couples them by $R_{ZZ}$, Eq.~\eqref{eq:ep-pair}.
The anchor rotation and the suffix sweeps of patch $A$ together form $V_A^{\mathrm{post}}$ of Eq.~\eqref{eq:app-state}, so the anchor is applied once.
Each patch is read out in the all-$X$ or the all-$Z$ basis.
(c)~The bridge as a weighted sum of products of local operations on spin $a$ (top wire) and spin $b$ (bottom wire), read as maps on density matrices.
Here $c=\cos(\phi/2)$, $s=\sin(\phi/2)$, $d=\frac12\sin\phi$, and $R_\pm=R_Z(\pm\pi/2)$.
The terms of weight $c^2$ and $s^2$ apply the identity or $Z$ on both spins, and the four terms of weight $\pm d$ come from the commutator.
A $Z$ measurement $M_Z$ lets the circuit continue and multiplies the outcome of that shot by the recorded sign.
The two patches never interact.}
\label{fig:experiment}
\end{figure*}%
}

\newcommand{\tableGrid}{%
\begin{table}[b]
\centering
\small
\setlength{\tabcolsep}{6pt}
\renewcommand{\arraystretch}{1.12}
\begin{tabular}{@{}ccccccccc@{}}
\toprule
 & & & & & & \multicolumn{3}{c}{two-qubit gates}\\
\cmidrule(l){7-9}
$N$ & $(n_x,n_y)$ & $n$ & $n+1$ & $N_{\mathrm p}$ & $D_0$ & unrouted & routed & routed depth\\
\midrule
8 & $(1,1)$ & 4 & 5 & 162 & 2 & 38 & 38 & 36\\
14 & $(1,2)$ & 7 & 8 & 322 & 2 & 74 & 83 & 56\\
20 & $(1,3)$ & 10 & 11 & 482 & 2 & 110 & 134 & 66\\
24 & $(2,2)$ & 12 & 13 & 642 & 4 & 145 & 173 & 103\\
26 & $(1,4)$ & 13 & 14 & 642 & 2 & 146 & 181 & 77\\
32 & $(1,5)$ & 16 & 17 & 802 & 2 & 182 & 227 & 87\\
34 & $(2,3)$ & 17 & 18 & 962 & 4 & 216 & 258 & 105\\
48 & $(3,3)$ & 24 & 25 & 1{,}442 & 8 & 322 & 382 & 124\\
62 & $(3,4)$ & 31 & 32 & 1{,}922 & 8 & 428 & 512 & 142\\
80 & $(4,4)$ & 40 & 41 & 2{,}562 & 16 & 569 & 1{,}375 & 451\\
98 & $(4,5)$ & 49 & 50 & 3{,}202 & 16 & 710 & 1{,}519 & 406\\
120 & $(5,5)$ & 60 & 61 & 4{,}002 & 32 & 886 & 2{,}436 & 556\\
142 & $(5,6)$ & 71 & 72 & 4{,}802 & 32 & 1{,}062 & 2{,}424 & 542\\
168 & $(6,6)$ & 84 & 85 & 5{,}762 & 64 & 1{,}273 & 2{,}995 & 574\\
\bottomrule
\end{tabular}
\caption{Patch geometries of the studies and of the hardware sweep.
The subcircuit width $n+1$ includes the readout ancilla, which carries no angle, the number of stored angles $N_{\mathrm p}$ follows from Eq.~\eqref{eq:app-params} with two prefix and two suffix sweeps, and $D_0=2^{\min(n_x,n_y)}$ is the exact MPS bond dimension of the zero-field state.
The last three columns give the native two-qubit gate count of the largest branch circuit of the series in Fig.~\ref{fig:hardware}, before and after layout and routing on \texttt{ibm\_phoenix}, and its two-qubit depth after routing.}
\label{tab:app-grid}
\end{table}%
}

\newcommand{\figTransfer}{%
\begin{figure*}[b]
\centering
\resizebox{\textwidth}{!}{%
\begin{tikzpicture}[font=\small]
  \begin{scope}
    \node[anchor=north west] at (-0.8,3.95) {(a)};
    \fill[FLDCOutsideHead!80] (0.95,0) rectangle (1.15,3.3);
    \draw[->] (0,0) -- (4.35,0) node[right,font=\scriptsize] {$h$};
    \draw[->] (0,0) -- (0,3.55) node[above,font=\scriptsize] {$E$};
    \node[below,font=\scriptsize] at (0,0) {$0$};
    \draw (3.8,0.05) -- (3.8,-0.05); \node[below,font=\scriptsize] at (3.8,0) {$1$};
    \draw[FLDCGood,line width=1.1pt] plot[smooth] coordinates {(0,1.05) (0.5,1.0) (1.05,0.9) (1.6,0.75) (2.4,0.55) (3.8,0.45)};
    \draw[FLDCBad,line width=0.9pt] plot[smooth] coordinates {(0,2.55) (0.5,2.3) (0.85,1.55) (1.05,1.0) (1.3,1.45) (1.9,1.95) (2.6,2.1) (3.8,2.15)};
    \node[FLDCGood,font=\scriptsize,below] at (2.9,0.55) {$E_0(N,h)$};
    \node[FLDCBad,font=\scriptsize,above] at (2.9,2.12) {$E_1(N,h)$};
    \draw[<->,black!80] (2.0,0.7) -- (2.0,1.98);
    \node[right,font=\scriptsize] at (2.02,1.3) {$\gamma_N(h)$};
    \draw[->,FLDCTheory,thick] (0.05,1.28) to[bend left=8] (0.85,1.18);
    \node[FLDCTheory,font=\scriptsize] at (0.42,1.66) {$\mathcal U(h)$};
    \fill[FLDCOpen] (0,1.05) circle (1.8pt);
    \fill[FLDCOpen] (3.8,0.45) circle (1.8pt);
    \node[font=\scriptsize,align=center,anchor=north] at (0.62,0.86) {stabilizer\\state};
    \node[font=\scriptsize,align=center,anchor=south] at (3.75,0.5) {product\\state};
    \node[font=\scriptsize,anchor=south,align=center] at (1.05,3.32) {small gap};
  \end{scope}
  \begin{scope}[shift={(6.45,0)}]
    \node[anchor=north west] at (-1.25,3.95) {(b)};
    \fill[FLDCOpen!12] (0,0) rectangle (4.75,1.0);
    \draw[->] (0,0) -- (5.0,0) node[right,font=\scriptsize] {$h$};
    \draw[->] (0,0) -- (0,3.7) node[above,font=\scriptsize] {$N$};
    \node[below,font=\scriptsize] at (0,0) {$0$};
    \draw (4.6,0.06) -- (4.6,-0.06); \node[below,font=\scriptsize] at (4.6,0) {$1$};
    \draw (0.06,1.0) -- (-0.06,1.0); \node[left,font=\scriptsize] at (0,1.0) {$N_{\mathrm{src}}$};
    \draw (0.06,3.0) -- (-0.06,3.0); \node[left,font=\scriptsize] at (0,3.0) {$N_{\mathrm{tar}}$};
    \node[font=\scriptsize,FLDCOpen!80!black,anchor=east] at (4.7,0.42) {exact evaluator, $N\le48$};
    \draw[FLDCGood,line width=2pt] (0,0.05) -- (0,3.45);
    \node[font=\scriptsize,FLDCGood,rotate=90,anchor=south] at (-0.95,1.75) {exact angles at $h=0$};
    \draw[->,FLDCOpen,thick] (0.08,1.0) -- (1.06,1.0);
    \draw[->,FLDCTheory,thick] (1.16,1.1) -- (1.16,2.86);
    \fill[FLDCBad] (1.16,3.0) circle (2.6pt);
    \draw[FLDCBad,thick] (1.16,3.0) circle (5pt);
    \fill[FLDCOpen] (1.16,1.0) circle (2.2pt);
    \draw[densely dotted,black!40] (1.16,0) -- (1.16,0.95);
    \node[below,font=\scriptsize] at (1.16,0) {$h_{\mathrm{tar}}$};
    \draw[->,FLDCOpen,thick] (1.95,3.1) -- (2.4,3.1); \node[right,font=\scriptsize] at (2.45,3.1) {continuation, source size};
    \draw[->,FLDCTheory,thick] (1.95,2.65) -- (2.4,2.65); \node[right,font=\scriptsize] at (2.45,2.65) {transfer at fixed $h$};
    \fill[FLDCBad] (2.175,2.2) circle (2.2pt); \draw[FLDCBad,thick] (2.175,2.2) circle (4.2pt); \node[right,font=\scriptsize] at (2.45,2.2) {used without refinement};
  \end{scope}
  \begin{scope}[shift={(13.35,0.45)},x=0.78cm,y=0.78cm]
    \node[anchor=north west] at (-0.85,3.4) {(c)};
    \foreach \x/\y in {0/0,1/0,0/1,1/1} \fill[FLDCOpen!14] (\x,\y) rectangle ++(1,1);
    \foreach \k in {0,...,3}{\draw[black!60,line width=0.45pt] (\k,0) -- (\k,3); \draw[black!60,line width=0.45pt] (0,\k) -- (3,\k);}
    \draw[FLDCTheory,line width=0.9pt,densely dashed] (-0.25,-0.25) rectangle (1.75,1.75);
    \foreach \a in {0,90,180,270} \draw[FLDCBad,line width=1.2pt] (1,1) -- ++(\a:0.22);
    \foreach \a in {0,90,180} \draw[FLDCBad,line width=1.2pt] (1,0) -- ++(\a:0.22);
    \foreach \a in {0,90,270} \draw[FLDCBad,line width=1.2pt] (0,1) -- ++(\a:0.22);
    \foreach \a in {0,90} \draw[FLDCBad,line width=1.2pt] (0,0) -- ++(\a:0.22);
    \foreach \x in {0,1,2} \foreach \y in {0,1,2,3} \fill[FLDCOpen] (\x+0.5,\y) circle (1.5pt);
    \foreach \x in {0,1,2,3} \foreach \y in {0,1,2} \fill[FLDCOpen] (\x,\y+0.5) circle (1.5pt);
    \draw[FLDCTheory,line width=1.6pt] (1,0.5) to[bend right=40] (0.5,1);
    \fill[FLDCTheory] (1,0.5) circle (3.1pt);
    \fill[FLDCTheory] (0.5,1) circle (3.1pt);
    \node[FLDCTheory] at (0.33,0.36) {$\mu$};
    \begin{scope}[shift={(3.4,0)}]
      \foreach \a in {0,90,180,270} \draw[FLDCBad,line width=1.2pt] (0.2,2.75) -- ++(\a:0.2);
      \node[right,font=\scriptsize] at (0.45,2.75) {full star};
      \foreach \a in {0,90,180} \draw[FLDCBad,line width=1.2pt] (0.2,2.1) -- ++(\a:0.2);
      \node[right,font=\scriptsize] at (0.45,2.1) {truncated star};
      \filldraw[fill=FLDCOpen!14,draw=black!50] (0.04,1.3) rectangle (0.36,1.62);
      \node[right,font=\scriptsize] at (0.45,1.46) {plaquette};
      \draw[FLDCTheory,densely dashed,line width=0.9pt] (0.04,0.65) rectangle (0.36,0.97);
      \node[right,font=\scriptsize] at (0.45,0.81) {window, $r=1$};
    \end{scope}
    \node[font=\scriptsize,anchor=north] at (1.5,-0.3) {class of $\mu$ in a $3\times3$ patch};
  \end{scope}
\end{tikzpicture}}
\caption{Spectrum and angle preparation.
(a)~Sketch of the two lowest levels along the field.
The ground energy $E_0(N,h)$ connects the stabilizer state at $h=0$ to the product state at $h=1$, and the gap $\gamma_N(h)$ to the first excited level is small in the region of the bulk transition (gray).
Where the gap is open, the ground state is the image of the zero-field state under the quasi-adiabatic unitary $\mathcal U(h)$ of Eq.~\eqref{eq:qa-flow}.
The curves place neither a transition nor a gap closing for a finite patch.
(b)~Routes to the angles of a target at $(h_{\mathrm{tar}},N_{\mathrm{tar}})$ in the plane of field $h$ and size $N$.
Exact angles are known at $h=0$ for every size (green line), and optimization with the exact evaluator is possible for $N\le48$ (blue region).
We continue in $h$ at a source size (blue arrow), transfer the angles in $N$ at fixed field (purple arrow), and use them at the target without refinement (red circle).
Continuation at the target size itself is possible only inside the blue region.
(c)~Class of a gate.
The pair block $\mu$ acts on two spins on consecutive edges of a face (purple arc).
Its class records its role in the sweeps and the labeled Hamiltonian terms in a square window of radius $r$ around it (dashed), where a star truncated by the boundary differs from a full one.
Gates of the same class, up to a symmetry of the square lattice, receive the circular mean of the source angles of that class, Eq.~\eqref{eq:app-transfer}.}
\label{fig:transfer}
\end{figure*}%
}

\newcommand{\figTransferExample}{%
\begin{figure*}[t]
\centering
\resizebox{\textwidth}{!}{%
\begin{tikzpicture}[x=0.62cm,y=0.62cm,font=\small]
  \node[anchor=north west] at (-0.7,4.3) {(a)};
  \begin{scope}[shift={(1.0,2.3)}]
    \fill[FLDCFailHead] (0,0) rectangle (1,1); \fill[colLocal] (1,0) rectangle (2,1); \fill[FLDCFailHead] (2,0) rectangle (3,1);
    \foreach \k in {0,...,3}{\draw[black!55,line width=0.4pt] (\k,0) -- (\k,1);} \draw[black!55,line width=0.4pt] (0,0) -- (3,0); \draw[black!55,line width=0.4pt] (0,1) -- (3,1);
    \node[font=\scriptsize] at (0.5,0.5) {E}; \node[font=\scriptsize] at (1.5,0.5) {M}; \node[font=\scriptsize] at (2.5,0.5) {E};
    \node[font=\scriptsize,anchor=south] at (1.5,1.15) {source $1\times3$};
  \end{scope}
  \begin{scope}[shift={(0.0,0.0)}]
    \fill[FLDCFailHead] (0,0) rectangle (1,1); \fill[colLocal] (1,0) rectangle (4,1); \fill[FLDCFailHead] (4,0) rectangle (5,1);
    \foreach \k in {0,...,5}{\draw[black!55,line width=0.4pt] (\k,0) -- (\k,1);} \draw[black!55,line width=0.4pt] (0,0) -- (5,0); \draw[black!55,line width=0.4pt] (0,1) -- (5,1);
    \node[font=\scriptsize] at (0.5,0.5) {E}; \foreach \x in {1.5,2.5,3.5} \node[font=\scriptsize] at (\x,0.5) {M}; \node[font=\scriptsize] at (4.5,0.5) {E};
    \node[font=\scriptsize,anchor=north] at (2.5,-0.15) {target $1\times5$};
  \end{scope}
  \draw[->,FLDCTheory,thick] (1.5,2.25) -- (0.5,1.05);
  \draw[->,FLDCTheory,thick] (2.5,2.25) -- (1.5,1.05);
  \draw[->,FLDCTheory,thick] (2.5,2.25) -- (2.5,1.05);
  \draw[->,FLDCTheory,thick] (2.5,2.25) -- (3.5,1.05);
  \draw[->,FLDCTheory,thick] (3.5,2.25) -- (4.5,1.05);
  \node[anchor=north west] at (6.6,4.3) {(b)};
  \begin{scope}[shift={(7.6,2.3)}]
    \fill[FLDCFailHead] (0,0) rectangle (1,1); \fill[colLocal] (1,0) rectangle (2,1); \fill[FLDCFailHead] (2,0) rectangle (3,1);
    \foreach \k in {0,...,3}{\draw[black!55,line width=0.4pt] (\k,0) -- (\k,1);} \draw[black!55,line width=0.4pt] (0,0) -- (3,0); \draw[black!55,line width=0.4pt] (0,1) -- (3,1);
    \node[font=\scriptsize] at (0.5,0.5) {E}; \node[font=\scriptsize] at (1.5,0.5) {M}; \node[font=\scriptsize] at (2.5,0.5) {E};
    \node[font=\scriptsize,anchor=south] at (1.5,1.15) {source $1\times3$};
  \end{scope}
  \begin{scope}[shift={(8.1,-0.5)}]
    \fill[FLDCOutsideHead] (0,0) rectangle (2,2);
    \foreach \k in {0,1,2}{\draw[black!55,line width=0.4pt] (\k,0) -- (\k,2); \draw[black!55,line width=0.4pt] (0,\k) -- (2,\k);}
    \draw[->,FLDCOpen,line width=1pt] (0.5,0.5) -- (1.35,0.5);
    \draw[->,FLDCOpen,line width=1pt] (0.5,0.5) -- (0.5,1.35);
    \node[font=\scriptsize] at (1.5,1.5) {?};
    \node[font=\scriptsize,anchor=north] at (1,-0.15) {target $2\times2$};
  \end{scope}
  \draw[->,FLDCTheory,thick,densely dashed] (9.1,2.25) -- (9.1,1.55);
  \node[FLDCBad,font=\scriptsize,anchor=west,align=left] at (10.4,0.55) {no source gate has\\two perpendicular\\neighbouring faces};
  \node[anchor=north west] at (14.6,4.3) {(c)};
  \begin{scope}[shift={(16.0,-0.7)}]
    \foreach \x in {0,1,2} \foreach \y in {0,1,2} \fill[black!6] (\x,\y) rectangle ++(1,1);
    \fill[FLDCOpen!30] (1,1) rectangle (2,2);
    \foreach \k in {0,...,3}{\draw[black!55,line width=0.4pt] (\k,0) -- (\k,3); \draw[black!55,line width=0.4pt] (0,\k) -- (3,\k);}
    \draw[FLDCTheory,line width=1.1pt,densely dashed] (0.06,0.06) rectangle (2.94,2.94);
    \node[font=\scriptsize,anchor=north] at (1.5,-0.15) {source $3\times3$};
    \node[font=\scriptsize,align=center,anchor=south] at (1.5,3.15) {the neighboring faces\\touch the boundary};
  \end{scope}
  \begin{scope}[shift={(20.6,-0.7)}]
    \foreach \x in {0,...,4} \foreach \y in {0,...,4} \fill[black!4] (\x,\y) rectangle ++(1,1);
    \fill[FLDCOpen!30] (2,2) rectangle (3,3);
    \foreach \k in {0,...,5}{\draw[black!55,line width=0.4pt] (\k,0) -- (\k,5); \draw[black!55,line width=0.4pt] (0,\k) -- (5,\k);}
    \draw[FLDCTheory,line width=1.1pt,densely dashed] (1.06,1.06) rectangle (3.94,3.94);
    \node[font=\scriptsize,anchor=north] at (2.5,-0.15) {target $5\times5$};
    \node[font=\scriptsize,align=center,anchor=south] at (2.5,5.15) {the same neighborhood\\lies in the bulk};
  \end{scope}
\end{tikzpicture}}
\caption{Coverage of a target by its sources.
(a)~The faces of a one-row strip are of two kinds, end faces E with one neighboring face and middle faces M with two.
The best-found angles of a $1\times3$ source form a table with one entry per kind and role, from which a $1\times5$ target, or any longer strip, is filled (arrows).
(b)~Every face of a $2\times2$ patch has one neighbor beside it and one above or below it, an arrangement that never occurs in a one-row strip, so a strip source offers nothing to the gates of this target.
(c)~The highlighted faces are both called interior, but in the $3\times3$ patch the neighboring faces touch the boundary (dashed box), whereas in the $5\times5$ patch the same neighborhood lies in the bulk.
The class of a gate compares the Hamiltonian terms in a window of stated radius together with its role, and the dashed boxes illustrate such a window.}
\label{fig:transfer-example}
\end{figure*}%
}

\newcommand{\figTransferDiagnostics}{%
\begin{figure}[t]
\centering
\includegraphics[width=\textwidth]{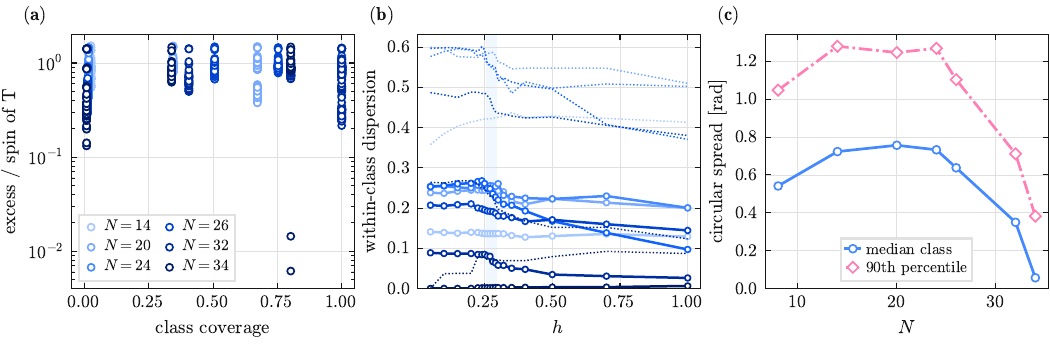}
\caption{Why the transfer fails, first campaign at $h>0$.
(a)~Excess energy per spin of the transferred state against the class coverage of its prediction, one shade per target size.
(b)~Within-class circular dispersion $1-\rho$ of directly optimized angles against $h$ for the median class (solid) and the 90th percentile (dotted).
(c)~The corresponding circular spread against $N$.
At $N=32$ and $34$ the spread is small, which we attribute to runs that hardly left their start.
Panel (a) pools the eighteen fields with $h>0$ of the first campaign.}
\label{fig:transfer-diagnostics}
\end{figure}%
}

\newcommand{\figLandscape}{%
\begin{figure}[b]
\centering
\includegraphics[width=\textwidth]{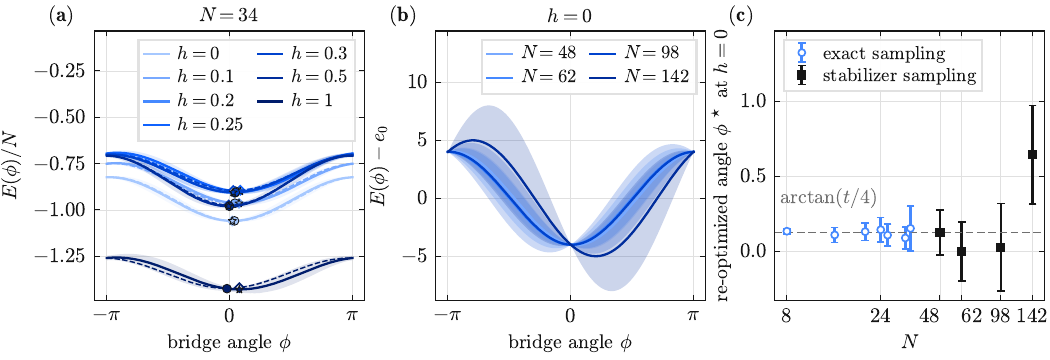}
\caption{Knitted records give the energy at every bridge angle, but shot noise hides the effect of the bridge as the patches grow.
Noiseless simulations of the ideal branch circuits at the shot budget of the device runs, with the patch angles fixed and the same six-branch records reweighted for every bridge angle $\phi$, at $t=0.5$.
(a) Energy per spin $E(\phi)/N$ at $N=34$, one shade per field, reconstructed from exact state-vector samples (solid, with bootstrap bands) and evaluated exactly for the same patch angles (dashed).
Circles mark the trained angles, stars the reconstructed minima, and open diamonds the exact energies at the reconstructed minimizers.
(b) Angle-dependent part $E(\phi)-e_0$ at $h=0$ for $N=48$, 62, 98, and 142, reconstructed from stabilizer samples of the Clifford zero-field branch circuits (App.~\ref{app:endpoints}).
The exact curve, $-D_t\cos(\phi-\phi_0)$ with $D_t=\sqrt{16+t^2}$, is the same at every size, while the bands of the reconstructions widen with $N$.
(c) Estimated minimizer $\phi^\star$ at $h=0$ against $N$, from exact samples for $N\le34$ (open circles) and stabilizer samples beyond (squares), with the exact optimum $\phi_0=\arctan(t/4)$ (dashed).
Bands and error bars are bootstrap uncertainties.}
\label{fig:landscape}
\end{figure}%
}

\newcommand{\figTransferResults}{%
\begin{figure*}[t]
\centering
\includegraphics[width=\textwidth]{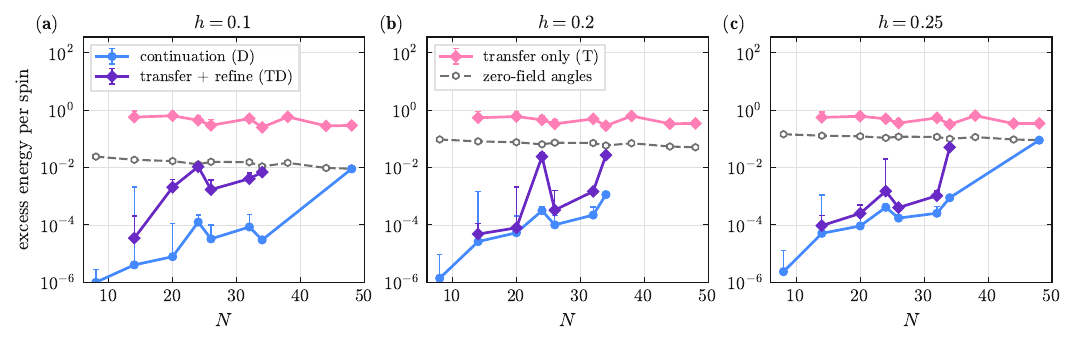}
\caption{Transferred angles against continuation.
Excess energy per spin over the reference against $N$ at (a)~$h=0.1$, (b)~$h=0.2$, and (c)~$h=0.25$.
Blue, continuation in $h$ at the target size.
Magenta, transferred angles before any target update, with the matching rule at the radius that is best at each size and field.
Purple, the same after refinement at the target size.
Gray, dashed, the exact zero-field angles reused unchanged.
Symbols mark the best and bars the worst of three seeds.
At $N=38$ and $44$ only transferred and zero-field angles exist, and at $N=48$ continuation exists at $h=0.1$ and $0.25$ and the refined transfer at $h=0.25$.}
\label{fig:transfer-results}
\end{figure*}%
}

\newcommand{\figTierTwo}{%
\begin{figure*}[t]
\centering
\includegraphics[width=\textwidth]{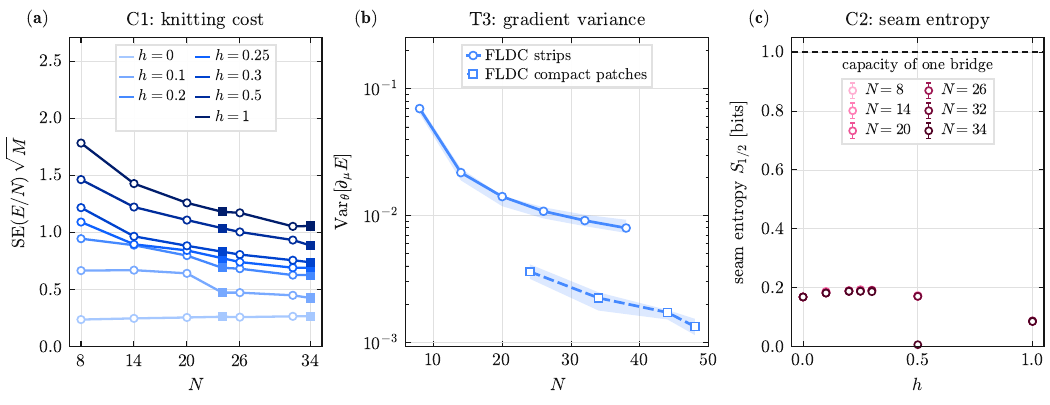}
\caption{
In simulation, the single-bridge application passes the knitting and gradient checks.
One criterion per panel, from classical simulations of the ideal circuits.
(a)~Knitting cost (C1), the bootstrap error of $E/N$ times $\sqrt M$ against $N$, one shade per field, from exact finite-shot simulations with $M=24{,}576$ shots per point, with strips of $1\times k$ plaquettes joined by lines and the compact $2\times2$ and $2\times3$ patches as squares.
(b)~Variance of the energy gradient at random initialization (T3) against $N$, from exact simulation at $h=0.25$ over 64 to 2{,}000 random draws of the angles, for strips (circles) and compact patches (squares).
Points are medians over three groups of parameters sorted by causal distance, with the bridge angle excluded, and bands span the envelope of their 95\% bootstrap intervals (App.~\ref{app:variance}).
(c)~Seam entropy (C2) of the prepared states against the field, one shade per size, with bars of one standard error, the exact ground state for $N\le24$ (gray), and the one bit that a single bridge can carry (dotted).
Classical difficulty of the branch tasks (S1) has no panel and remains open.}
\label{fig:tier2}
\end{figure*}%
}

\newcommand{\figVariance}{%
\begin{figure*}[t]
\centering
\includegraphics[width=\textwidth]{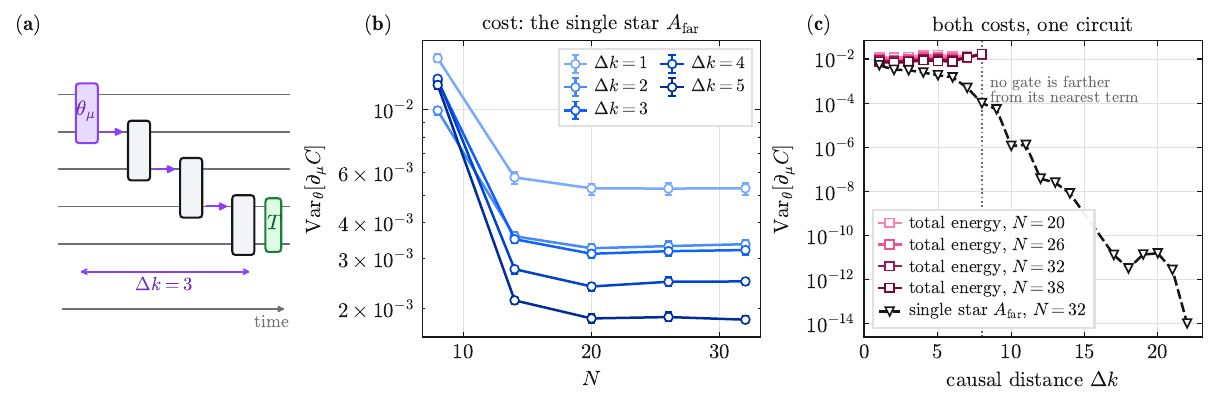}
\caption{Gradient variance at random initialization (T3).
(a)~A sequential sweep.
Blocks (boxes) act on pairs of neighboring spins (lines) in turn, and time runs to the right.
The number of blocks acting on each spin, given at the right, stays at most two however many spins the staircase covers, which is the finite local depth of an FLDC.
The causal distance $\Delta k$ from the parameter $\theta_\mu$ (purple) to the measured term $T$ (green) counts the steps of the shortest chain of blocks to the last block on the support of $T$, here $\Delta k=3$, and it grows with the number of spins.
(b)~Variance of the gradient of a single star $A_{\mathrm{far}}$ at $h=0.25$ over uniformly random angles against $N$, at fixed $\Delta k$ (light to dark), for strips.
$A_{\mathrm{far}}$ is the weight-two star of patch $A$ farthest from the contact, at the open end of the strip.
(c)~The same variance against $\Delta k$, for $A_{\mathrm{far}}$ at $N=32$ (black) and for the total energy at $N=20$, $26$, $32$, and $38$ (magenta, light to dark), where $\Delta k$ is the distance to the nearest Hamiltonian term.
The dotted line marks the largest such distance of any gate.
Bars are 95\% bootstrap intervals.}
\label{fig:variance-locality}
\end{figure*}%
}

\newcommand{\figSignalFraction}{%
\begin{figure}[b]
\centering
\includegraphics[width=0.5\textwidth]{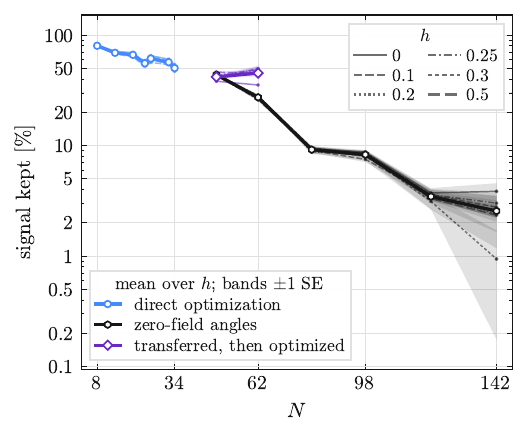}
\caption{Fraction of the signal kept at $h=0$, the measured over the exact energy, against $N$ for each kind of angle vector, with dynamical decoupling and post-selection.
The correction of App.~\ref{app:rescaling} divides every energy of a series by this fraction.}
\label{fig:signal}
\end{figure}%
}

\newcommand{\figRescaling}{%
\begin{figure}[b]
\centering
\includegraphics[width=\textwidth]{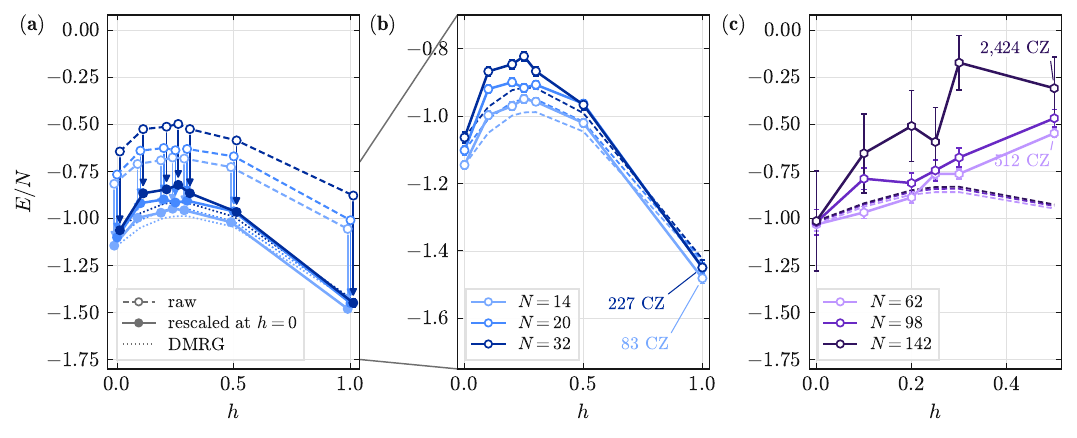}
\caption{Device energies corrected at $h=0$.
(a)~Effect of the correction for $N=14$, 20, and 32 with the directly optimized angles.
Open markers joined by dashed lines are the raw device energies, filled markers joined by solid lines the same series divided by the fraction of the signal that its $h=0$ point kept, and arrows join the two at every field.
(b)~The corrected series of (a), enlarged, with their error bars.
(c)~Corrected energies of the zero-field angles at $N=62$, 98, and 142 for $h\le0.5$.
Thin dotted and dashed lines are the DMRG reference, which the zero-field angles do not reach for $h>0$, so the relevant comparison in (c) is with their closed-form energy, Eq.~\eqref{eq:ep-zf}.
Labels give the largest native two-qubit gate count of one routed subcircuit.
At $h=0$ the corrected point equals the exact energy by construction, and in (a) the sizes are shifted slightly in $h$.}
\label{fig:rescaling}
\end{figure}%
}

\newcommand{\figHardware}{%
\begin{figure*}[t]
\centering
\includegraphics[width=\textwidth]{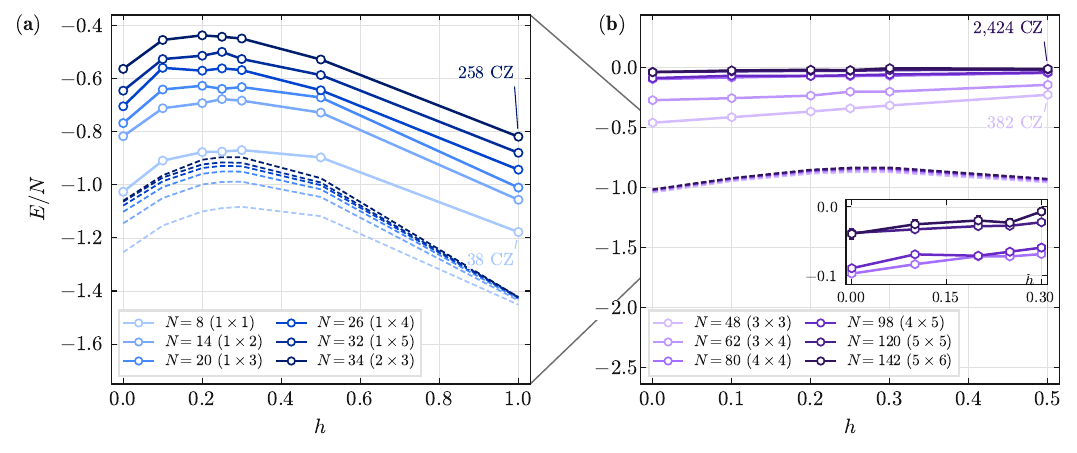}
\caption{
Device energies keep the shape of the reference at small sizes and lose most of their signal at large sizes.
Energy per spin $E/N$ against the field $h$, measured on the IBM Nighthawk processor \texttt{ibm\_phoenix} by knitting the bridge gate, one shade per size from light to dark.
(a)~Patches of $N=8$ to 34 spins with directly optimized angles.
(b)~Patches of $N=48$ to 142 spins with the zero-field angles, for $h\le0.5$, with an inset enlarging the two largest patches.
Markers joined by solid lines are raw device values, and dashed lines the ground-state energy, from exact diagonalization for $N\le24$ and from the density-matrix renormalization group (DMRG) otherwise.
Labels give the largest number of native two-qubit (CZ) gates in one routed subcircuit, and gray lines connect the energy axes of the two panels.
Each point is one run of $M=24{,}576$ shots, and error bars, where larger than the markers, are one bootstrap standard error from shot noise (App.~\ref{app:cut-energy}), computed from the post-selected records (App.~\ref{app:device}).}
\label{fig:hardware}
\end{figure*}%
}

\newcommand{\figTransferBenchmark}{%
\begin{figure}[t]
\centering
\includegraphics[width=\textwidth]{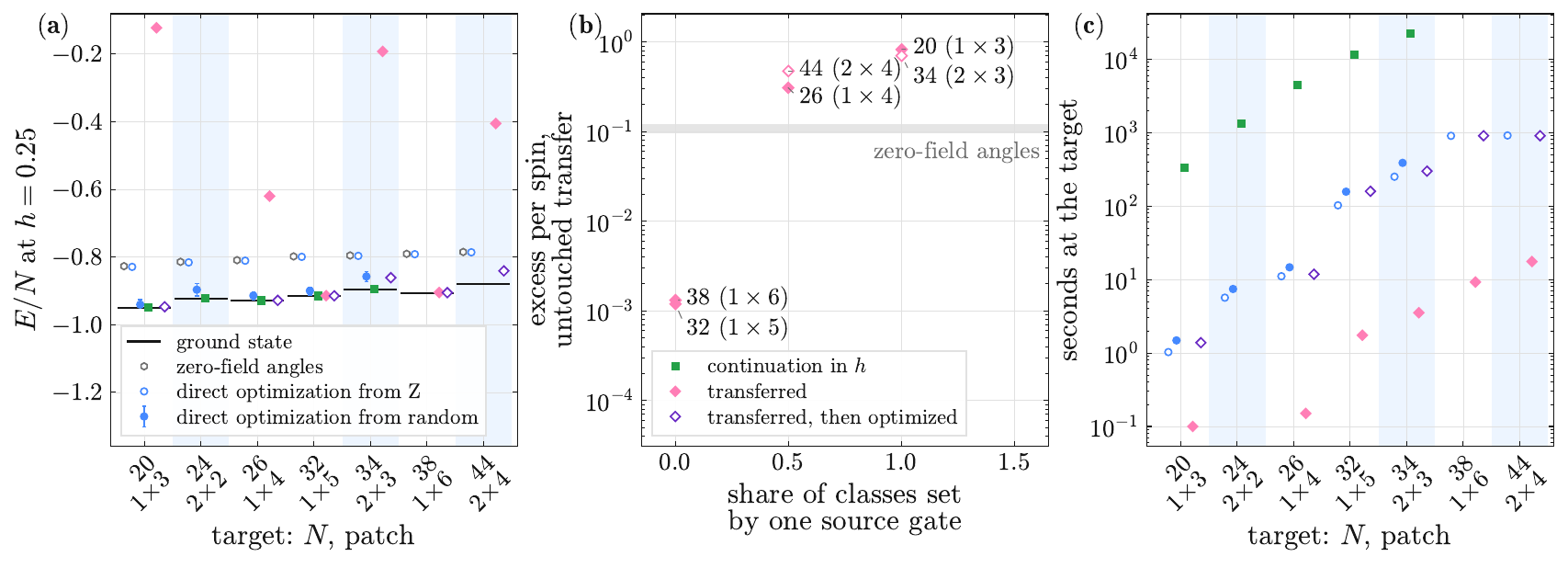}
\caption{Joint class table against optimization at the target, at $h=0.25$.
In (a) and (c) each column is one target patch, with compact patches shaded.
(a)~Energy per spin of every route, with the reference (black bar) and the zero-field angles (gray).
Blue, direct optimization at the target with a budget of 100 iterations or 900\,s, from the zero-field angles (open) and from random angles (filled, mean and range over 10 draws).
Green, continuation in $h$ at the target.
Magenta, the joint class table expanded onto the target without optimization, and purple, the same refined with the direct budget.
(b)~Excess energy per spin of the untouched transfer against the share of target angles whose class a single source gate decides, filled for strips.
The gray band marks the excess of the zero-field angles.
(c)~Seconds spent at the target for the three routes that start from smaller patches.
The one-off training of the table, which serves every target of a family, is not included.}
\label{fig:transfer-benchmark}
\end{figure}%
}

\newcommand{\tableVariancePosition}{%
\begin{table}[t]
\centering
\small
\setlength{\tabcolsep}{7pt}
\renewcommand{\arraystretch}{1.12}
\begin{tabular}{@{}ccccccc@{}}
\toprule
Distance from the nearer end & $N=8$ & $14$ & $20$ & $26$ & $32$ & $38$\\
\midrule
0 & 69.0 & 21.8 & 17.7 & 17.7 & 17.7 & 18.4\\
0.5 & -- & 20.9 & 13.6 & 13.0 & 13.2 & 13.4\\
1 & -- & -- & 8.95 & 7.41 & 7.52 & 7.49\\
1.5 & -- & -- & -- & 4.42 & 4.16 & 4.02\\
2 & -- & -- & -- & -- & 3.34 & 3.12\\
2.5 & -- & -- & -- & -- & -- & 2.73\\
\bottomrule
\end{tabular}
\caption{Variance of the energy gradient on strips at $h=0.25$, averaged over the gates at a given distance from the nearer end of the strip, in plaquettes, and given in units of $10^{-3}$.
The entries of a row agree wherever the strip is long enough to contain that distance.}
\label{tab:variance-position}
\end{table}%
}

\newcommand{\figQuasiLocal}{%
\begin{figure*}[t]
\centering
\resizebox{\textwidth}{!}{%
\begin{tikzpicture}[font=\small]
  \begin{scope}[x=0.72cm,y=0.72cm]
    \node[anchor=north west] at (-3.6,3.4) {(a)};
    \fill[FLDCTheory!12] (0,0) -- (-2.7,2.7) -- (2.7,2.7) -- cycle;
    \fill[FLDCTheory!12] (0,0) -- (-2.7,-2.7) -- (2.7,-2.7) -- cycle;
    \fill[colLocal] (-3.1,-1.0) rectangle (3.1,1.0);
    \fill[FLDCTheory!30] (0,0) -- (-1.0,1.0) -- (1.0,1.0) -- cycle;
    \fill[FLDCTheory!30] (0,0) -- (-1.0,-1.0) -- (1.0,-1.0) -- cycle;
    \draw[FLDCTheory!80,densely dashed] (-2.7,2.7) -- (0,0) -- (2.7,2.7);
    \draw[FLDCTheory!80,densely dashed] (-2.7,-2.7) -- (0,0) -- (2.7,-2.7);
    \draw[->,black!65] (-3.3,0) -- (3.45,0) node[right,font=\scriptsize] {$x$};
    \draw[->,black!65] (0,-2.95) -- (0,3.05) node[above,font=\scriptsize] {$\tau$};
    \fill[FLDCBad] (0,0) circle (2.6pt);
    \node[FLDCBad,font=\scriptsize,anchor=north west,fill=white,inner sep=1pt] at (0.15,-0.1) {$\partial_hH_j$};
    \node[FLDCTheory!80!black,font=\scriptsize,anchor=west] at (1.75,2.15) {$|x|\lesssim v|\tau|$};
    \node[FLDCGood!80!black,font=\scriptsize,anchor=west] at (3.15,0.8) {$|\tau|\lesssim1/\gamma$};
    \draw[<->,FLDCTheory!90!black,thick] (-1.0,1.25) -- (1.0,1.25);
    \node[FLDCTheory!90!black,font=\scriptsize,anchor=south] at (0,1.3) {$\xi\sim v/\gamma$};
    \node[font=\scriptsize,anchor=north,align=center,black!70] at (0,-3.05) {outside the cone the influence is\\suppressed, not zero};
  \end{scope}
  \begin{scope}[shift={(5.9,-2.2)},x=0.72cm,y=0.72cm]
    \node[anchor=north west] at (-0.85,6.4) {(b)};
    \draw[->,black!65] (0,0) -- (5.3,0) node[right,font=\scriptsize] {$|\tau|$};
    \draw[->,black!65] (0,0) -- (0,4.6) node[above,font=\scriptsize] {$|W_\gamma(\tau)|$};
    \fill[colLocal] (0,0) rectangle (1.4,4.2);
    \draw[FLDCGood,line width=1.2pt,domain=0:5,samples=80] plot (\x,{3.9*exp(-1.15*\x)});
    \draw[densely dashed,black!50] (1.4,0) -- (1.4,4.2);
    \node[below,font=\scriptsize] at (1.4,0) {$1/\gamma$};
    \node[FLDCGood!80!black,font=\scriptsize,anchor=west,align=left] at (2.0,2.6) {tails decay beyond $1/\gamma$\\but do not vanish};
    \draw[->,FLDCGood!80!black] (2.4,2.1) -- (3.2,0.22);
  \end{scope}
  \begin{scope}[shift={(11.5,-1.2)},x=0.45cm,y=0.45cm]
    \node[anchor=north west] at (-0.9,7.2) {(c)};
    \begin{scope}[shift={(0,0)}]
      \fill[black!4,rounded corners=2pt] (-0.25,-0.25) rectangle (4.25,4.25);
      \draw[black!70,line width=0.5pt] (0,0) grid[step=1] (4,4);
      \fill[FLDCTheory!22,rounded corners=2pt] (1,1) rectangle (3,3);
      \draw[FLDCTheory,densely dashed,line width=0.9pt,rounded corners=2pt] (1,1) rectangle (3,3);
      \draw[FLDCBad,line width=1.1pt] (1.8,2) -- (2.2,2); \draw[FLDCBad,line width=1.1pt] (2,1.8) -- (2,2.2);
      \node[FLDCBad,font=\scriptsize,anchor=south west] at (2.05,2.1) {$H_j$};
      \node[font=\scriptsize,anchor=north] at (2,-0.4) {source};
    \end{scope}
    \begin{scope}[shift={(6.2,0)}]
      \fill[black!4,rounded corners=2pt] (-0.25,-0.25) rectangle (6.25,4.25);
      \draw[black!70,line width=0.5pt] (0,0) grid[step=1] (6,4);
      \fill[FLDCTheory!22,rounded corners=2pt] (2,1) rectangle (4,3);
      \draw[FLDCTheory,densely dashed,line width=0.9pt,rounded corners=2pt] (2,1) rectangle (4,3);
      \draw[FLDCBad,line width=1.1pt] (2.8,2) -- (3.2,2); \draw[FLDCBad,line width=1.1pt] (3,1.8) -- (3,2.2);
      \node[FLDCBad,font=\scriptsize,anchor=south west] at (3.05,2.1) {$H_k$};
      \node[font=\scriptsize,anchor=north] at (3,-0.4) {target};
    \end{scope}
    \draw[<->,FLDCTheory,thick] (3.15,2) -- (8.05,2);
    \node[FLDCTheory!90!black,font=\scriptsize,anchor=south,align=center] at (5.6,5.0) {same labeled Hamiltonian\\in both windows};
    \draw[FLDCTheory!90!black] (5.6,4.95) -- (5.6,2.1);
    \node[font=\scriptsize,anchor=north,align=center] at (6.1,-1.4) {$\mathcal D_j^{(r)}=\mathcal D_k^{(r)}$ under the identification,\\the full pieces differ by their tails};
  \end{scope}
\end{tikzpicture}}
\caption{How a quasi-local generator arises and why it suggests transfer.
(a)~The Heisenberg evolution of a local term $\partial_hH_j$ spreads within a light cone $|x|\lesssim v|\tau|$ (purple), with influence outside the cone that is exponentially suppressed but not zero, where the Lieb--Robinson velocity $v$ is the largest speed at which local operators spread under the Hamiltonian.
The filter of Eq.~\eqref{eq:transfer-generator} weights evolution times $|\tau|\lesssim1/\gamma$ (green band), so the piece $\mathcal D_j$ is concentrated within a distance $\xi\sim v/\gamma$ of the term (dark region).
(b)~Envelope of the filter, which decays beyond $1/\gamma$ without a sharp cutoff.
The drawn envelope is not a specific filter, and $v/\gamma$ is a crossover scale rather than a support radius.
(c)~Two terms in patches of different size whose windows of radius $r$ carry identical labeled Hamiltonians have identical truncated pieces $\mathcal D^{(r)}$, Eq.~\eqref{eq:transfer-tail}.
Gates in identical windows are therefore candidates for receiving the same angle, which is the hypothesis behind the matching rule of App.~\ref{app:transfer}.}
\label{fig:qac-range}
\end{figure*}%
}

\newcommand{\figContact}{%
\begin{figure}[t]
\centering
\includegraphics[width=\columnwidth]{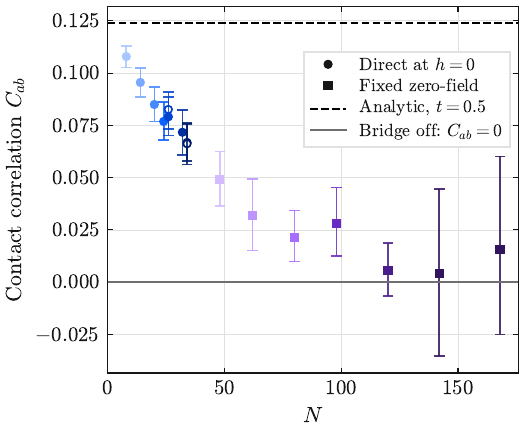}
\caption{Knitting recovers part of the contact correlation that independent execution misses.
Connected contact correlation $C_{ab}$ of Eq.~\eqref{eq:contact-corr} at $h=0$, knitted from the hardware records of Fig.~\ref{fig:hardware}, for the trained angles (circles) and the zero-field angles (squares).
The dashed line is the target $t/\sqrt{16+t^2}$, and the solid line the bridge-off prediction $C_{ab}=0$.
Error bars are 95\% bootstrap intervals from shot noise, with the three expectation values resampled jointly.}
\label{fig:contact}
\end{figure}%
}

\definecolor{scrBridge}{HTML}{4589FF}   %
\definecolor{scrToricA}{HTML}{08BDBA}   %
\definecolor{scrToricB}{HTML}{FA4D56}   %
\definecolor{scrMPS}{HTML}{6929C4}      %
\definecolor{scrBrickLin}{HTML}{FF7EB6} %
\definecolor{scrBrickConst}{HTML}{BE95FF} %
\definecolor{scrHEA}{HTML}{00539A}      %
\definecolor{scrQAOA}{HTML}{24A148}     %
\providecommand{\scrsw}[1]{\textcolor{#1}{\rule[0.5ex]{1.2em}{1.6pt}}}
\providecommand{\scrbad}[1]{\fbox{#1}}
\newcommand{\tableTierOneScreen}{%
\begin{table}[t]
\centering
\caption{\label{tab:tier1-screen}%
Tier-1 screen of Fig.~\ref{fig:tier1-screen}, with the trend recorded in Table~\ref{tab:families} next to the value found at the largest generated size.
All values are taken at about 512 qubits, 513 for row 3 and 517 for row 8b, except for routed QAOA, whose largest size is 192 qubits.
The swatch gives the line of each generator in the figure.
Boxed values fail the criterion or, for T2, collapse with the size.}
\footnotesize
\setlength{\tabcolsep}{2.6pt}
\setlength{\fboxsep}{1.5pt}
\begin{tabular}{@{}l l l cc cc cc cc@{}}
\toprule
 & Row & Generator & \multicolumn{2}{c}{$\chi$ (T1)} & \multicolumn{2}{c}{$\log_{10}g$ (T2)} & \multicolumn{2}{c}{$c/n$ (S2)} & \multicolumn{2}{c}{$\kappa_{\mathrm{found}}$ (C1)}\\
\cmidrule(lr){4-5}\cmidrule(lr){6-7}\cmidrule(lr){8-9}\cmidrule(lr){10-11}
 & & & trend & value & trend & value & trend & value & trend & value\\
\midrule
\scrsw{scrHEA}        & 3   & Hardware-efficient ansatz       & $\Theta(d)$     & \scrbad{342} & ${\le}10^{-28}$  & \scrbad{$-411$} & \fldcyes & 1.00 & $\omega(1)$      & \scrbad{258} \\
\scrsw{scrQAOA}       & 5   & Routed QAOA, $p=2$              & $\omega(1)$     & \scrbad{384} & $4^{-\omega(1)}$ & \scrbad{$g=0$}  & \fldcopen & 1.00 & $\Theta(n)$      & \scrbad{192} \\
\scrsw{scrBrickLin}   & 6a  & Brickwork, linear depth         & $\omega(1)$     & \scrbad{256} & $4^{-\omega(1)}$ & \scrbad{$-616$} & \fldcyes & 0.96 & $\Theta(\sqrt n)$ & \scrbad{128} \\
\scrsw{scrMPS}        & 7   & Sequential MPS, $D=4$           & $O(1)$          & 3            & $D^{-O(1)}$      & $-10$           & \fldcyes & 1.00 & $O(1)$           & 2 \\
\scrsw{scrToricA}     & 8a  & Clifford toric, bounded contact & ${\le}6$        & 3            & \fldcna          & \fldcna         & \fldcyes & 0.83 & $O(1)$           & 1 \\
\scrsw{scrToricB}     & 8b  & Clifford toric, extensive seam  & ${\le}6$        & 2            & \fldcna          & \fldcna         & \fldcyes & 1.00 & $\Theta(L)$      & \scrbad{13} \\
\scrsw{scrBridge}     & 12b & Single-bridge FLDC              & $O(1)$          & 18           & \fldcna          & $-86$           & \fldcyes & 1.00 & 1                & 1 \\
\scrsw{scrBrickConst} & --  & Brickwork, constant depth       & --              & 4            & --               & $-9$            & --       & \scrbad{0.02} & -- & 2 \\
\bottomrule
\end{tabular}
\end{table}%
}

\let\fldcaddcontentsline\addcontentsline
\renewcommand{\addcontentsline}[3]{}

\begin{document}

\title{Structural Conditions for Distributed Quantum Advantage}

\author{Sabina Dr\u{a}goi}
\email{sdragoi@phys.ethz.ch}
\affiliation{Institute for Theoretical Physics, ETH Z\"urich,
Wolfgang-Pauli-Str. 27, 8093 Z\"urich, Switzerland}
\affiliation{IBM Research, S\"{a}umerstrasse 4,
CH-8803 R\"{u}schlikon, Switzerland}

\author{Mar\'ia Gragera Garc\'es}
\email{m.gragera.garces@ed.ac.uk}
\affiliation{Quantum Software Lab, University of Edinburgh,
Edinburgh, United Kingdom}

\author{Lirand\"{e} Pira}
\email{lpira@nus.edu.sg}
\affiliation{Center for Quantum Technologies,
National University of Singapore, Singapore}

\begin{abstract}

Circuit cutting runs a quantum computation on processors smaller than the circuit, at the price of classical knitting whose cost grows exponentially with the number of cut gates.
Keeping this cost affordable is necessary but not sufficient for an advantage, since knitting classically easy subcircuits is itself classically easy.
We formulate three requirements for distributed quantum advantage under a stated budget: affordable knitting, classical difficulty that survives the cut, and, for variational circuits, resolvable gradients.
We prove that under stated assumptions knitting stays affordable when subcircuits grow behind a bounded interface and becomes exponentially expensive when fixed-width subcircuits multiply.
We turn the requirements into screening criteria and apply them to eighteen circuit families.
We identify finite local-depth circuits with bounded interfaces as a setting in which the requirements can coexist, with hardness established only in the worst case.
As a classically verifiable proof of principle, we knit one gate joining two field-perturbed toric-code patches on an IBM Nighthawk processor for up to 142 spins, and confirm that the reconstruction can at least resolve the contact correlation beyond independent execution for up to 98 qubits.
\end{abstract}

\maketitle

\section{Introduction}
\label{sec:intro}

Many applications of quantum computing require more qubits than a single processor currently provides~\cite{Beverland2022Requirements,Cain2026AtomicShor,Shajan2026Protein, Merz2026ProteinLigand}.
As in classical distributed computing~\cite{vanSteen2016Distributed,Dean2008MapReduce}, a natural way to scale up is to distribute the computation across several quantum chips~\cite{Caleffi2024Survey,Barral2025Review,Barratt2021Holographic}.
Quantum links between processors have been demonstrated~\cite{Main2025OpticalLink}, while \emph{circuit cutting}~\cite{BravyiSmithSmolin2016,Tang2021CutQC,Lowe2023RandomizedMeasurements,MitaraiFujii2021VirtualGate, CarreraVazquez2024RealTime,Peng2020CircuitCutting} offers an alternative requiring no quantum link.
It removes the gates that act across a chosen partition of the qubits, runs the resulting subcircuits separately, and classically \emph{knits} their measurement outcomes together.
Circuit cutting thus trades the width of a processor for classical post-processing work.

However, this approach faces two obstacles.
Knitting requires a number of samples that grows exponentially with the number of removed gates~\cite{Ufrecht2024Joint,Schmitt2025Cutting,PiveteauSutter2024Knitting,Brenner2023WireCutting}.
Each removed gate is replaced by a signed combination of local operations, and averaging out the resulting cancellations multiplies the number of samples by a constant factor for every gate.
A cheap cut is not enough either.
If the subcircuits are easy to simulate classically, knitting their results together is easy as well~\cite{PiveteauSutter2024Knitting,BravyiSmithSmolin2016}, and the quantum hardware offers no advantage~\cite{Hangleiter2026Advantage}.
A circuit that promises a quantum advantage is therefore not automatically a good candidate for distributed execution.

Variational algorithms add a third requirement.
They exploit the expressibility of parametrized circuits, which can prepare useful states with few gates, provided their parameters are well optimized~\cite{Peruzzo2014VQE,Kandala2017HEA,Preskill2018NISQ,Sayginel2024FTVQE}.
These parameters are usually trained using gradient-based methods that scale only if the gradients remain resolvable as the system grows, and at random initialization they often do not.
In a \emph{barren plateau}, gradient variances decay exponentially with the number of qubits, so resolving a gradient takes exponentially many measurements~\cite{McClean2018BarrenPlateaus,Cerezo2021CostFunction,Wang2021NoiseBP,Ragone2024Lie,Larocca2025BPReview}.
Keeping circuits shallow can bring their observables within reach of classical simulation in the regimes studied in Refs.~\cite{Cerezo2025Simulability,Angrisani2024Classical,Bermejo2024QCNN}.
The simulation guarantees depend on the architecture, ensemble, accuracy, and access to input-state data; successful simulations of large experimental circuits provide further examples~\cite{Tindall2024Eagle,Begusic2024FastConverged}.
For architectures designed to balance trainability and expressivity, the encoding, readout, and required accuracy must also be specified when assessing classical complexity~\cite{Kerenidis2026Scalable}.
Our earlier work made this conflict precise for a one-dimensional brickwork circuit split into two blocks joined by a fixed number of gates~\cite{GrageraGarces2026EntanglementGeometry}.
In that setting, for area-law input states and local costs at random initialization, avoiding a barren plateau requires depth $O(\log n)$, whereas escaping the matrix-product-state simulation criterion used there requires a superpolynomial bond dimension, which in turn requires depth $\omega(\log n)$.

\figMap

This paper asks which circuit families are worth distributing.
First, we identify three requirements.
Knitting must remain affordable as the system grows, which we call \textsc{Cut}.
The tasks left to each subcircuit must remain classically hard, which we call \textsc{Sim}, since otherwise the knitted result is classically easy too.
For variational circuits, gradients must remain resolvable, which we call \textsc{Train}.
We turn each requirement into criteria that can be checked for a given circuit family and resource budget.
Assessed together, the requirements locate where a distributed advantage can reside.
An efficient cut whose subcircuit tasks are classically easy is classically easy as a whole (Theorem~\ref{prop:stitching}).
For the cutting constructions and target states we consider, growing the subcircuits behind a bounded or logarithmically growing interface keeps knitting affordable, whereas growing the number of fixed-width subcircuits makes it exponentially expensive (Theorem~\ref{prop:regimes}).
An advantage must therefore come from subcircuits that grow behind a bounded interface and remain classically hard.

Second, we apply the criteria to eighteen circuit families and regimes from the literature (Table~\ref{tab:families}), from standard variational ground-state solvers~\cite{Peruzzo2014VQE,Kandala2017HEA} to sequential preparations of topological states~\cite{Satzinger2021ToricCode,Wei2022SequentialPEPS,ZaletelPollmann2020isoTNS}.
The variational ansätze and random circuits all have a local depth that grows with the system size, which removes the only available trainability guarantee.
The topological states cut along an extended seam fail because their interface grows with the length of the seam and the knitting cost grows exponentially with it.
The families that remain admit an affordable cut, but most of them leave subcircuits that a classical computer simulates efficiently, making the advantage disappear.
Only one of the tested families combines an affordable cut with branch tasks that are provably hard, and even there the hardness holds only in the worst case.

These observations point to a more specific architectural requirement. A useful distributed circuit should allow its computational reach to grow with system size without increasing either the local depth that controls the available trainability guarantee or the interface that controls the knitting cost.

\emph{Finite local-depth circuits} (FLDCs) provide a natural way to separate these conditions~\cite{ZhangLiuZhang2024FLDC}.
Instead of the total depth, they bound the size of each local block and the number of blocks acting on each qubit.
For blocks drawn independently from local 2-designs and local costs satisfying the coefficient condition of Sec.~\ref{sec:preliminaries}, this keeps gradient variances bounded away from zero independently of the system size, even when the total depth grows linearly with $n$.
FLDCs can therefore prepare long-range entangled states, and their local observables can depend on a region of the circuit that grows with the system size.
If only a bounded number of gates connect the subcircuits, the knitting cost also stays bounded.
Distributed FLDCs with bounded interfaces thus provide a candidate architecture in which gradient-resolution guarantees are compatible with large total depth, under the stated local-cost and initialization assumptions, and in which affordable knitting, trainability, and classical difficulty might coexist.

Third, we test this candidate on the simplest bounded interface, two open patches of the field-perturbed toric code~\cite{Kitaev2003Anyons,WuDengProkofev2012} joined by one variational entangling gate, extending the benchmark of Ref.~\cite{ZhangLiuZhang2024FLDC}.
Cutting the bridge gate leaves two subcircuits which we run on an IBM Nighthawk processor and recombine classically~\cite{MitaraiFujii2021VirtualGate,CarreraVazquez2024RealTime}.
We simulate up to 48 spins and measure up to 142 on the device, with directly optimized angles up to 34 spins and, beyond, zero-field angles that prepare the exact ground state without the perturbing field.

The hardware experiment is a classically verifiable test of whether an affordable bridge reconstruction recovers physically relevant information that running the two patches independently would miss.
Our simulations confirm that the knitting cost does not grow with the size and that the energy gradients show no exponential decay.
On the device, the knitted energies reproduce the field dependence of the ideal circuits, but their signal falls with the depth of the routed subcircuits, by about $0.6\%$ per layer of two-qubit gates, with no error mitigation beyond dynamical decoupling, twirling, and post-selection~\cite{Viola1999DD,Wallman2016RC,DragoiEgger2026Transpilation}.
Since the bridge changes the energy per spin only at order $1/N$, we also knit the connected correlation across the contact, which independent execution cannot produce.
It stays clearly nonzero up to 98 spins, from about $87\%$ of its target at 8 spins to about a fifth at 80 and 98 spins, so the reconstruction retains part of the cross-patch information within the limits of device noise.

The remaining challenge is state preparation at scales beyond classical simulation.
Resolvable gradients make training possible in principle with polynomially many shots, but exact classical simulation limits it in practice, which is why the larger patches use the zero-field angles.
We propose a parameter transfer motivated by quasi-adiabatic continuation, which carries angles from short strips to longer ones without further optimization, see App.~\ref{app:angles}. We leave the extension to compact patches to future work.

Section~\ref{sec:preliminaries} states the three requirements and fixes the notation, Sec.~\ref{sec:screening} formulates the criteria and applies them, Sec.~\ref{sec:application} presents the toric-code test, and Sec.~\ref{sec:conclusion} discusses what the results imply.
Figure~\ref{fig:map} summarizes this structure, while the appendices collect all technical details.

\section{Requirements for distributed quantum advantage}
\label{sec:preliminaries}

This section states the requirements \textsc{Cut}, \textsc{Sim}, and \textsc{Train} for a circuit family as $n$ grows.
We consider the task of estimating expectation values $\langle O\rangle$ at the output $U|0^n\rangle$ of an $n$-qubit circuit $U$, for observables $O$ with support on at most $r=O(1)$ qubits and norm $\lVert O\rVert\le1$.

\figPartition

Circuit cutting estimates full-circuit expectation values from measurements on smaller subcircuits.
Fix a partition $\Pi$ that assigns each qubit to a subcircuit, see Fig.~\ref{fig:partition}(a).
Each gate channel $\mathcal{V}_s$ crossing $\Pi$ admits a signed \emph{quasiprobability decomposition} into subcircuit-local operations~\cite{Peng2020CircuitCutting,MitaraiFujii2021VirtualGate},
\begin{equation}
  \mathcal{V}_s=\sum_i a_{s,i}\,\mathcal{F}_{s,i},
  \qquad
  \gamma_s=\sum_i |a_{s,i}|,
  \label{eq:prelim-qpd}
\end{equation}
see Fig.~\ref{fig:partition}(b). The coefficients $a_{s,i}$ are real, and each map $\mathcal{F}_{s,i}$ is implemented by local operations, measurements, and classical post-processing.
For $\kappa$ crossing gates decomposed independently, a \emph{branch} $\alpha=(i_1,\ldots,i_\kappa)$ replaces each crossing gate $s$ by $\mathcal{F}_{s,i_s}$ and leaves all other gates unchanged, see Fig.~\ref{fig:partition}(c).
This defines a full-circuit branch map $\mathcal{F}_\alpha$ with coefficient $c_\alpha=\prod_{s=1}^{\kappa}a_{s,i_s}$, and
\begin{equation}
  \mathcal{U}=\sum_\alpha c_\alpha\,\mathcal{F}_\alpha,
  \qquad
  \Gcut=\sum_\alpha |c_\alpha|,
  \label{eq:prelim-circuit-qpd}
\end{equation}
where $\mathcal{U}(\rho)=U\rho U^\dagger$ is the channel of the original circuit.
By linearity, $\langle O\rangle$ is the same signed sum of branch contributions.
Each contribution is obtained by running the subcircuits separately, with the operations of branch $\alpha$ inserted, and knitting their outcomes classically.
We call the estimation of local expectation values of one such subcircuit a \emph{branch task}.

Sampling branch $\alpha$ with probability $p_\alpha=|c_\alpha|/\Gcut$ and weighting its knitted outcome by $\Gcut\operatorname{sgn}(c_\alpha)$ gives an unbiased estimator of $\langle O\rangle$.
For a Pauli-product observable measured in its product eigenbasis, branch outcomes lie in $[-1,1]$, so the weighted outcomes lie in $[-\Gcut,\Gcut]$. 
Hoeffding's inequality~\cite{Hoeffding1963} then shows that
\begin{equation}
  M\ge\frac{2\Gcut^2}{\varepsilon^2}\ln\frac{2}{\delta}
  \label{eq:prelim-shots}
\end{equation}
independent shots suffice for additive error $\varepsilon$ with failure probability at most $\delta$, where a shot runs the subcircuits of one chosen branch once and knits their outcomes.
The sampling cost therefore grows as $\Gcut^2$.
We call a family \emph{cheaply cuttable} when $\Gcut=O(1)$ and \emph{efficiently cuttable} when $\Gcut$ grows at most polynomially with $n$.
\textsc{Cut} requires efficient cuttability.

The value of $\Gcut$ depends on both the crossing gates and the chosen decomposition.
We define the \emph{interface capacity} $\kappa(\Pi)$ as the number of two-qubit gate channels severed by $\Pi$, with repeated gates counted separately.
For crossing gates of the form $\exp[-i\theta_s(P_s\otimes Q_s)/2]$, with $P_s$ and $Q_s$ single-qubit Pauli operators, optimal decompositions of the individual gates give~\cite{Schmitt2025Cutting}
\begin{equation}
  \Gamma_{\mathrm{ind}}=\prod_{s=1}^{\kappa}\bigl(1+2|\sin\theta_s|\bigr)\le 3^\kappa.
  \label{eq:prelim-independent}
\end{equation}
For parallel crossing gates, an optimal \emph{joint decomposition}, which treats them collectively, achieves the smaller value~\cite{Schmitt2025Cutting,Ufrecht2024Joint}
\begin{equation}
  \Gamma_{\mathrm{joint}} =2\prod_{s=1}^{\kappa}\bigl(1+|\sin\theta_s|\bigr)-1\le 2^{\kappa+1}-1.
  \label{eq:prelim-joint}
\end{equation}
In both cases, a bounded interface keeps $\Gcut$ bounded however large the subcircuits become.
Neither decomposition requires classical communication between subcircuits during execution, so the subcircuits can run on separate processors or one after another on the same processor, as in the local-operations setting of Ref.~\cite{CarreraVazquez2024RealTime}.

Entanglement across the partition also constrains $\Gcut$, independently of the chosen decomposition.
For a pure output $|\psi\rangle$ and a bipartition $\Pi=A{:}B$, let $\lambda_a$ be the eigenvalues of the reduced state $\rho_A$.
We define the \emph{seam entropy} as the R\'enyi-$1/2$ entropy
\begin{equation}
  \Ehalf(\rho_A)=2\log_2\Bigl(\sum_a\sqrt{\lambda_a}\Bigr),
  \label{eq:seam-entropy}
\end{equation}
and use it to measure how much entanglement a reconstruction has to reproduce across the partition.
Appendix~\ref{app:cost} collects the bounds that relate it to $\Gcut$.

An efficient cut does not by itself imply an advantage.
We call a branch task \emph{efficiently simulable} when a classical algorithm achieves inverse-polynomial additive error in polynomial time.
\textsc{Sim} requires that at least one branch task entering the knitting is not efficiently simulable.
If every branch task is efficiently simulable, so is the knitted estimate (App.~\ref{app:sim}).

For a parametrized circuit $U(\boldsymbol\theta)$, we take the cost $C(\boldsymbol\theta)=\sum_\nu\lambda_\nu\langle O_\nu\rangle_{\boldsymbol\theta}$ to be a weighted sum of expectation values of observables with support on at most $r$ qubits, and write $\partial_\mu C=\partial C/\partial\theta_\mu$.
\textsc{Train} requires that $\operatorname{Var}_{\boldsymbol\theta}[\partial_\mu C]$ at random initialization decays at most polynomially with $n$, which excludes a barren plateau.

For trainability, we rely on the guarantee available for finite local-depth circuits:
\begin{definition}[Finite local-depth circuit]
  \label{def:fldc}
  For a parametrized $n$-qubit circuit decomposed into unitary blocks,
  \[
    U(\boldsymbol\theta)=U_M(\boldsymbol\theta_M)\cdots U_1(\boldsymbol\theta_1),
  \]
  the maximum block support and the local depth are
  \[
  \begin{aligned}
    \beta &= \max_k |\supp(U_k)|,\\
    \chi &= \max_q \bigl|\{k:q\in\supp(U_k)\}\bigr|,
  \end{aligned}
  \]

  where $\supp(U_k)$ denotes the qubits acted on by block $U_k$, so $\chi$ is the maximum number of blocks acting on any one qubit, see Fig.~\ref{fig:partition}(d).
  A family of such circuits is a finite local-depth circuit (FLDC) family if it admits decompositions for which $\beta$ and $\chi$ remain bounded independently of $n$.
  We call a circuit an FLDC when it belongs to such a family.
\end{definition}

When the blocks are drawn independently from local 2-designs, Zhang, Liu, and Zhang show that $\operatorname{Var}_{\boldsymbol\theta}[\partial_\mu C]\ge g$~\cite{ZhangLiuZhang2024FLDC}. The constant $g$ depends only on $r$, $\beta$, $\chi$, and the coefficients $\lambda_\nu$ of the cost terms whose support overlaps the block containing $\theta_\mu$. It is therefore independent of $n$ whenever the sum of the squares of these coefficients is bounded away from zero. The bound holds for any total depth and is stated in full in App.~\ref{app:train}.

\section{Assessing circuit families for distributed execution}
\label{sec:screening}

This section turns the three requirements into criteria that can be checked for a given circuit family, and then applies them to circuit families from the literature.
Each criterion is a screen, not a certificate.
Passing all of them identifies a candidate, not a demonstrated quantum advantage.
A failure excludes a family only where the criterion states a necessary condition.

\subsection{Assessment framework}
\label{subsec:inputs}

The assessment takes two inputs.
The first is a \emph{circuit generator}, an algorithm that returns, in time polynomial in $n$, the gate list, qubit supports, observables, and parameters of the $n$-qubit circuit.
A generator specifies a family, so the criteria can ask how resources scale with $n$ rather than how a single instance behaves.
The second input is a \emph{device model} $(w,D,M,\varepsilon,\delta)$.
It allows $D$ subcircuits of at most $w$ qubits each and a budget of $M$ shots per expectation value, which is to be estimated to additive error $\varepsilon$ with failure probability at most $\delta$.
An experimental budget must additionally account for the subcircuit executions behind each shot and for the number of observables.

The partition is chosen as part of the assessment.
For a fixed gate representation, the smallest interface over balanced partitions is
\begin{equation}
  \kstar(n)=\min_{\Pi\ \mathrm{balanced}}\kappa(\Pi),
  \label{eq:balanced-cut}
\end{equation}
where balanced means that every subcircuit holds $n/D$ qubits up to a fixed relative tolerance.
Every candidate subcircuit must also fit within $w$ qubits.
Finding this minimum is computationally hard, so a heuristic search returns only an upper bound on $\kstar(n)$.
A small bound exhibits an affordable partition.
A large bound proves nothing, because a better partition may exist.
We therefore report the partition that was found and never treat a failed search as an exclusion.

\subsection{Assessment criteria}
\label{subsec:criteria}

Table~\ref{tab:properties} lists the eight criteria, grouped by the requirement they test.
Each criterion belongs to one of three tiers, which record the kind of evidence it needs.
Tier-1 criteria are structural checks computed from the generator alone.
Tier-2 criteria are numerical diagnostics at sizes that a classical simulator can reach.
Tier-3 criteria are theoretical obligations, met only by a proof or a citation and otherwise recorded as open.
Open entries are kept separate from failures throughout.

\tableProperties

\textit{Knitting cost} (\textsc{Cut}).
Criterion C1 asks whether a partition can be knitted within the budget.
Equation~\eqref{eq:prelim-shots} gives a worst-case sample budget for the stated quasiprobability estimator, and we use its inverse as a conservative norm-budget threshold,
\begin{equation}
  \Gamma_{\max}=\varepsilon\sqrt{\frac{M}{2\ln(2/\delta)}}.
  \label{eq:gamma-budget}
\end{equation}
C1 passes when $\Gcut\le\Gamma_{\max}$, while a larger norm does not, by itself, exclude a task-specific or variance-aware estimator.
Inserting the worst-case norms of Eqs.~\eqref{eq:prelim-independent} and~\eqref{eq:prelim-joint} turns this into a largest number of severed gates, for $\Gamma_{\max}\ge1$,
\begin{equation}
\begin{aligned}
  \kappa_{\max}^{\mathrm{joint}}
    &=\left\lfloor\log_2\frac{\Gamma_{\max}+1}{2}\right\rfloor,\\
  \kappa_{\max}^{\mathrm{ind}}
    &=\left\lfloor\log_3\Gamma_{\max}\right\rfloor.
\end{aligned}
\label{eq:kmax}
\end{equation}
For the ancilla-free joint construction assessed here, the severed gates share a time slice.
Both counts grow only logarithmically with $M$, and App.~\ref{app:cost-regimes} evaluates them for typical budgets.
Each additional severed gate multiplies the worst-case shot count by nine for independent decompositions and by a factor approaching four for joint ones.
When the gate angles are known, the actual norm replaces these worst-case counts.

Criterion C1 answers only half of the practical question, because the subcircuits must also fit on the devices.
For a fixed partitioning strategy, let $n_{\mathrm{cut}}$ be the largest size at which knitting remains affordable and $n_{\mathrm{width}}\le Dw$ the largest size at which the subcircuits fit.
The family can then be run up to $n_{\max}=\min(n_{\mathrm{cut}},n_{\mathrm{width}})$.
We call it \emph{width-limited} when $n_{\mathrm{cut}}\ge n_{\mathrm{width}}$, since the hardware sets the reach.
Figure~\ref{fig:crossing} shows both cases.
A bounded interface whose norm is within budget is width-limited, while an interface whose norm grows with $n$ eventually becomes cut-limited.
Asymptotic and finite-size conclusions can differ.
A constant knitting cost may still exceed the budget, and a growing interface may stay affordable over the sizes of interest.

\figCrossing

Criterion C2 gives a lower bound that no exact state-reconstructing decomposition in the stated class can beat.
It rests on the following identity, proved in App.~\ref{app:cost}.
\begin{proposition}[Cost--entropy identity]
\label{prop:cost-entropy}
For a pure target $|\psi\rangle$ across $A{:}B$, the minimum coefficient one-norm of a signed decomposition into normalized separable states is
\begin{equation}
  \Gsep(\psi)=2\cdot2^{\Ehalf(\rho_A)}-1.
  \label{eq:cost-entropy}
\end{equation}
Any exact quasiprobability knitting from a product input, whose branches use local operations and classical communication, satisfies $\Gcut\ge\Gsep(\psi)$.
\end{proposition}
The first statement is the robustness of entanglement of a pure state~\cite{VidalTarrach1999}, expressed through the seam entropy, and the second is its consequence for exact knitting.
Lower bounds of this kind on the sampling overhead of quasiprobability decompositions are known for gates and wires~\cite{PiveteauSutter2024Knitting,Brenner2023WireCutting}.
The version here is stated for the output state, through its seam entropy, and therefore holds for every decomposition of every gate set at once.
The bound applies to one state and one partition.
If $\Gsep>\Gamma_{\max}$, no exact knitting protocol of this kind is affordable at that partition within the norm-budget model.
Excluding a family requires the bound at every admissible partition, so a large entropy at a single partition is not enough.
Conversely, a small $\Gsep$ does not provide an affordable protocol.
The identity concerns exact reconstruction of the state, not every method that estimates a single observable.

The interface limits the seam entropy but does not fix it.
Starting from a product state, each severed two-qubit Pauli rotation at most doubles the Schmidt rank across the partition, so $\Ehalf(\rho_A)\le\kappa$ whatever the subcircuit-local gates are (Proposition~\ref{prop:port-capacity} in App.~\ref{app:cost-bounded}).
The entropy actually reached can lie anywhere below this bound.
C2 is therefore evaluated on the states the task requires, including the trained states of a variational circuit.

\textit{Classical difficulty} (\textsc{Sim}).
An affordable cut matters only if classical difficulty survives it.
Criterion S1 rests on the following reduction.
\begin{theorem}[Efficient knitting preserves classical simulability]
\label{prop:stitching}
Let a bipartition sever $\kappa=O(\log n)$ gates, each replaced by at most $m=O(1)$ efficiently specified local terms with uniformly bounded coefficient norm and no communication between subcircuits.
If the branch tasks of both subcircuits are in \textsf{BPP}, uniformly over branches and to inverse-polynomial additive error, then estimating $\langle O\rangle$ for observables of bounded support is also in \textsf{BPP}.
\end{theorem}
The proof, in Theorem~\ref{thm:stitch} and Corollary~\ref{cor:trenchcoat} of App.~\ref{app:sim}, expands the observable into products of Pauli operators on each subcircuit and estimates every branch separately.
With $\kappa=O(\log n)$ there are at most $m^\kappa$ branches and $\Gcut$ is polynomial, so polynomially many branch estimates at inverse-polynomial precision suffice.
An efficient cut whose branch tasks are classically easy is thus classically easy as a whole.

Two consequences follow for S1.
Simulating the subcircuits with the severed gates simply deleted does not settle it, because the branch tasks include the inserted local operations.
Likewise, a hardness result supports S1 only if it concerns these branch tasks.
Appendix~\ref{app:hard} constructs a family in which such hardness sits inside one subcircuit of a bounded-interface cut.
The observation behind the theorem, that knitting classically simulable subcircuits yields a classically simulable estimate, is made in Ref.~\cite{PiveteauSutter2024Knitting}, and the hybrid simulation of Ref.~\cite{BravyiSmithSmolin2016} uses the same expansion with a cost exponential in the number of severed gates.
Theorem~\ref{prop:stitching} states it with the branch tasks as the unit of hardness, which is what criterion S1 needs.

Criteria S2 and S3 determine how simulation evidence is read.
S2 records the backward causal cones of the required observables, that is, the gates that can influence them and the qubits these gates act on, including the local operations inserted by the cutting branches.
Uniformly bounded cones permit direct classical simulation, whereas growing cones do not establish hardness (Proposition~\ref{prop:lightcone}).
The generator-level diagnostic in App.~\ref{app:screen} reports the largest probed gate-list cone, whose interpretation depends on the stated circuit representation and probe set.
S3 checks that a cited simulation or hardness result concerns local expectation values at the required accuracy.
Hardness of sampling alone does not meet it.
Together, S2 and S3 separate evidence against an advantage from the mere absence of a known simulation.

\textit{Gradient resolution} (\textsc{Train}).
For variational families, criterion T1 records $\beta$ and $\chi$ for an explicit block decomposition and checks whether the family is an FLDC.
T2 records whether the bound $g$ of Sec.~\ref{sec:preliminaries} applies, how it scales, and how large it is.
A size-independent $g$ can still be too small to help at a given budget.
Neither a failed T1 nor a small $g$ establishes a barren plateau, and Proposition~\ref{prop:design-plateau} gives a sufficient condition for one.
T3 therefore measures the gradient variance directly over a range of sizes, for a stated initialization ensemble and cost normalization.
A variance that does not decay supports trainability only over the tested range.
Random initialization and a warm start answer different questions, and T3 concerns the former.

In practice, one first fixes the task and budget, searches for partitions, and runs the Tier-1 checks.
Tier-2 diagnostics then refine the families that survive, and the Tier-3 obligations decide which claims can be made.
Finally, the reachable sizes must be compared with the sizes at which classical methods still perform the same branch tasks at matched accuracy.

\subsection{Scaling regimes}
\label{subsec:regimes}

The criteria also constrain how a computation may be distributed.
Growing the subcircuits and growing their number, the two ways of scaling a distributed computation, have opposite consequences.
\begin{theorem}[Scaling regimes for distributed circuit cutting]
\label{prop:regimes}
With $D=O(1)$ and $w=\Theta(n)$, an interface capacity $\kappa(\Pi)=O(1)$ gives $\Gcut=O(1)$, and $\kappa(\Pi)=O(\log n)$ gives $\Gcut=\poly(n)$.
With $w=O(1)$ and $D=\Theta(n/w)$, a connected arrangement severs at least $D-1$ gates.
If $|\sin\theta_s|$ is bounded below by a positive constant on a constant fraction of them, independent cutting has $\Gcut=2^{\Omega(D)}$, and the joint decomposition has the same scaling whenever its time-slice condition holds.
Independently, a pure target with $\Ehalf(\rho_A)=\Omega(\sqrt n)$ across a balanced grouping of the subcircuits obeys $\Gcut=2^{\Omega(\sqrt n)}$ for the exact knitting protocols of Proposition~\ref{prop:cost-entropy}.
\end{theorem}
Theorem~\ref{prop:regimes} is proved in App.~\ref{app:cost-regimes}.
Its first two statements concern the norms of the stated constructions and need not bound the best decomposition of a particular known circuit, whereas the entropy statement holds for every exact protocol in the stated class.
For the cutting constructions and target-state assumptions of the theorem, growing the subcircuits behind a bounded or logarithmically growing interface permits efficient reconstruction, with a constant or polynomial knitting cost, whereas the fixed-width regime incurs an exponential overhead.
Many small subcircuits would stay affordable only if the problem itself split into many weakly coupled pieces, which suggests that such a problem is close to classically parallel from the outset.
A cut can exploit a bottleneck that the physics provides, but it cannot create one.
Even in the favorable regime, a bounded interface supplies neither hardness nor trainability, and both must be established separately.

\subsection{Assessment of existing circuit families}
\label{sec:families}
\tableCircuits

Table~\ref{tab:families} applies the criteria to eighteen circuit families and regimes from the literature.
Each column corresponds to one criterion of Table~\ref{tab:properties}, ordered by tier, and each row records the outcome in a specified regime.
Appendix~\ref{app:screen} reports the Tier-1 criteria for seven explicit generators associated with these families and a constant-depth control, over sizes from 7 to 517 qubits (Fig.~\ref{fig:tier1-screen}).

Four patterns emerge among the assessed families.

\emph{Most assessed families lack one or more structural guarantees.}
Every conventional variational circuit in the table (rows 1--6) has a local depth $\chi$ that grows with system size, through repeated ansatz layers or, for QAOA, through the swaps added by routing on a planar device.
For ADAPT-VQE with the assessed generalized operator pool (row ~4), the block support $\beta$ grows as well; the verdict concerns this FLDC screen, not all adaptive variants~\cite{Tang2021QubitADAPT,Grimsley2023ADAPTBP}.
The bounded-cone statement for row~5 applies only at fixed $p$ and bounded logical degree, not to dense graphs or the conservative routed gate-list cone of App.~\ref{app:screen}.
Where evaluated, the formal bound $g$ is strongly suppressed, reaching values as small as $10^{-2167}$ for row~4.
This makes the FLDC bound uninformative at the assessed sizes, since it does not establish a barren plateau.
For the hardware-efficient ansatz and deep scrambling circuits (rows~3 and~6a), gradient suppression is supported separately by the T3 assessment.
Several assessed constructions have an interface that grows with the linear size $L$ of the partition boundary (rows 8b, 9, 10, 11a, 12a).
The $S_3$ example in row~9 is distinct from the related $D_4$ preparation settings of Refs.~\cite{Tantivasadakarn2023Shortest,Iqbal2024NonAbelian}; its fixed-state assessment does not cover fusion-assisted universal computation.
For them, knitting is exponentially expensive in $L$.

\emph{Affordable cutting is typically associated with classically tractable subcircuits.}
Sequential matrix-product-state preparation (row 7), the Clifford toric code with a bounded contact (row 8a), and strongly injective isometric tensor-network states (row 11b) all admit a constant interface.
Their subcircuits, however, are classically tractable, and by Theorem~\ref{prop:stitching} so is the knitted computation.
Row~7 refers to the sequential unitary construction; logarithmic-depth unitary~\cite{Malz2024LogDepth} and adaptive constant-depth~\cite{Smith2024ConstantDepth} preparations are alternative resource models.

For the assessed tasks, efficient classical simulation of the branch contributions also makes the reconstructed expectation values classically tractable.
The easy verdict for the Hubbard ladder (row~1b) concerns the ground-state reference calculation, while efficient simulation of all required variational branches would need a separate argument.

\emph{Few assessed families exhibit both classical hardness and affordable cutting.}
The families with the strongest hardness evidence carry seam entropy that grows with system size.
Sequential isometric tensor-network states with a transparent seam (row 11a) are hard in the worst case, but their seam entropy grows as $\sqrt n$, so no exact knitting protocol is affordable (Proposition~\ref{prop:cost-entropy}).
The hardness ticks in rows~6a and~6b refer to sampling evidence, not a verified S1 result for local-observable branches.
The Pauli-propagation guarantee of Ref.~\cite{Angrisani2024Classical} is ensemble dependent and polynomial at fixed error, becoming quasi-polynomial at inverse-polynomial accuracy; it does not by itself supply the uniform branch simulator required by S1.
Weakly injective isometric tensor-network states with a bounded interface (row 11c) provide a candidate with worst-case hardness; numerical trainability for a specified initialization ensemble remains a separate assessment.
Corollary~\ref{cor:coexistence} in App.~\ref{app:hard} shows that such families exist, but it does not identify natural instances.
The injectivity-dependent local-observable results are those of Ref.~\cite[Sec.~V]{MalzTrivedi2025}, with $\eta=\delta^2D^2$ in that source; Refs.~\cite{Haghshenas2022TNCircuits,Slattery2021isoTNSCircuits} provide related circuit and variational context.

\emph{Topological state families provide a benchmark for seam cost.}
For the Clifford toric code with an extended seam (row 8b), the knitting cost is known and grows exponentially with the seam length (Theorem~\ref{thm:seam-witness} and Corollary~\ref{cor:toric-exact} in App.~\ref{app:cost-transparent}), which follows from the analytic seam entropy $L-1$ of its fixed-point stabilizer state.
The non-Abelian examples (rows 9 and 10) show that richer anyon content does not change this picture.
The easy verdicts in rows~9 and~10 are restricted to the stated fixed-state or short-braid benchmarks, while neither nonuniversal braiding nor a preparation reference alone establishes efficient simulation of arbitrary cutting branches.
The variational toric code with a transparent seam (row 12a) inherits this exponential seam cost when its states retain enough of the fixed-point seam correlations (Corollary~\ref{cor:perturbed}), and at the fixed point the cost $2^L-1$ fits the norm threshold of Eq.~\eqref{eq:gamma-budget} only up to $L\le6$ for $M=10^9$, $\varepsilon=0.01$, and $\delta=0.05$.

Our single-bridge application (row 12b) severs one gate, so its Schmidt rank across the seam is at most two and its knitting norm $1+2|\sin\phi|$ never exceeds three.
Its local depth is bounded, and its sweeps are applied sequentially.
The guarantee constant does not apply because the fixed encoder, the bridge, and the anchor are not drawn from the local 2-designs that the bound of Sec.~\ref{sec:preliminaries} assumes.
We therefore measure the gradient variance of criterion T3 directly (Sec.~\ref{subsec:app-tier2}).

No single criterion decides the question.
A cheap seam does not imply an advantage, while a large seam entropy at every admissible partition excludes exact knitting regardless of the decomposition.
Passing every structural and numerical check still leaves the theoretical obligations open, and it is these obligations, above all S1, that decide whether distribution can deliver an advantage.

\section{A single-bridge distributed ground-state experiment}
\label{sec:application}

\figExperiment

The screen of Sec.~\ref{sec:screening} identifies bounded-interface FLDCs as a candidate class in which affordable knitting, trainability, and classical difficulty can coexist.
Its Tier-1 criteria follow from the circuit generator alone.
The Tier-2 and Tier-3 criteria depend on the states that a circuit prepares and on its branch tasks, so they can only be decided for a concrete model.
This section tests them on the simplest nontrivial interface, two FLDCs joined by one entangling gate (row 12b of Table~\ref{tab:families}).
The bounded interface controls the knitting cost, but it does not determine the quality of the variational state, the resolution of its gradients, or the classical difficulty of its branch tasks.
Because the branch tasks remain classically tractable at the sizes we run, the device results can be compared with classical references, which makes this a verifiable test of what bridge reconstruction recovers beyond independent execution.

Section~\ref{subsec:app-model} fixes the model and the circuit (App.~\ref{app:endpoints}).
Section~\ref{subsec:app-sim} evaluates the circuit and the ground state classically and prepares the angles (App.~\ref{app:angles}).
Section~\ref{subsec:app-tier2} checks the knitting cost and the Tier-2 criteria in simulation (Apps.~\ref{app:cutting} and~\ref{app:variance}).
Section~\ref{subsec:app-knit} reports the hardware results (App.~\ref{app:hardware}).

Throughout this section, $n$ denotes the number of spins in one patch and $N=2n$ the total number of spins.

\subsection{Model and variational circuit}
\label{subsec:app-model}

We consider two patches $A$ and $B$ of Kitaev's toric code in a parallel magnetic field $0\le h\le1$~\cite{Kitaev2003Anyons,Tupitsyn2010Multicritical,Vidal2009LowEnergy,WuDengProkofev2012}, connected by a contact of strength $t=0.5$ between the facing exterior spins $a\in A$ and $b\in B$,
\begin{equation}
\begin{aligned}
  H(h,t)={}&-(1-h)\sum_{p\in\{A,B\}}\Bigl(\sum_v S_{v,p}+\sum_f P_{f,p}\Bigr)\\
  &-h\sum_{j=0}^{N-1}\bigl(X_j+Z_j\bigr)-t\,Z_aZ_b,
\end{aligned}
\label{eq:app-H}
\end{equation}
where, for every vertex $v$ and face $f$ of patch $p\in\{A,B\}$ and with $X$ and $Z$ the Pauli matrices, the star and plaquette operators are
\begin{equation}
  S_{v,p}=\prod_{e\ni v}X_{e,p},
  \qquad
  P_{f,p}=\prod_{e\in\partial f}Z_{e,p}.
  \label{eq:app-stabilisers}
\end{equation}
Figure~\ref{fig:experiment}(a) shows one of each.
A star on the boundary acts on its two or three edges.
A patch is a rectangle of $n_x\times n_y$ square plaquettes with one spin on every edge, so that the number of spins in one patch is
\begin{equation}
  n=2n_xn_y+n_x+n_y .
  \label{eq:app-n}
\end{equation}

Our task is to prepare the ground state of Eq.~\eqref{eq:app-H} variationally and to estimate its energy.
We use the toric code~\cite{Kitaev2003Anyons} because its zero-field, i.e. $h=0$, ground state needs exactly what FLDCs offer, since no constant-depth local circuit prepares it~\cite{BravyiHastingsVerstraete2006,TikkuKim2022}, whereas a sequential circuit of bounded local depth does~\cite{Higgott2021Encoding,Satzinger2021ToricCode}.
The system is solvable at zero field, where the ground state is the common stabilizer state of the two patches mixed by the contact with a single excited state, and at $h=1$, where it is a product state apart from the contact pair (Apps.~\ref{app:ep-zero} and~\ref{app:ep-polarized}).
In between, the field and stabilizer terms do not commute and no closed form is known. 
The field drives the bulk out of the topological phase through a transition estimated between $h\approx0.25$ and $h\approx0.29$~\cite{Tupitsyn2010Multicritical,Vidal2009LowEnergy,WuDengProkofev2012}, where the gap $\gamma_N(h)$ between the two lowest levels is expected to become small.
This region is the natural target of a variational study at large sizes, since classical methods are also expected to struggle there.
Tensor-network evaluations of the circuit and its branch tasks stop converging within our budget at the hardware sizes near and above the transition (Sec.~\ref{subsec:app-sim}).
Criterion S1 therefore remains open.

The variational circuit mirrors the structure of the Hamiltonian~\cite{ZhangLiuZhang2024FLDC}, see Fig.~\ref{fig:experiment}(b).
A fixed encoder $E_p$ of Hadamard and CNOT gates prepares the zero-field ground state of each patch~\cite{Higgott2021Encoding}, and local sweeps then dress this state toward the ground state at $h>0$ (App.~\ref{app:ep-sweeps}).
A sweep proceeds through the faces row by row, and within each face it applies a pair block, made of $R_Y$ and $R_Z$ rotations on two adjacent spins and one $R_{ZZ}$ rotation between them, to each of the four pairs of consecutive edges in a fixed cyclic order, see inset of Fig.~\ref{fig:experiment}(b).
A single gate generated by the contact term couples the patches, as in the Hamiltonian variational ansatz~\cite{Wecker2015HVA},
\begin{equation}
  R_{ZZ}^{ab}(\phi)=\exp\bigl[-i\phi Z_aZ_b/2\bigr].
  \label{eq:app-bridge}
\end{equation}
It is followed on patch $A$ by an anchor rotation $R_{S_v}(\alpha)=e^{-i\alpha S_v/2}$, generated by a star $S_v$ incident to $a$, and on both patches by further sweeps.
Collecting the sweeps into $V_p^{\mathrm{pre}}$ and $V_p^{\mathrm{post}}$, with $V_A^{\mathrm{post}}$ including the anchor rotation, and writing $E_p$ for the encoders, the prepared state is
\begin{equation}
\begin{aligned}
  |\psi(\boldsymbol\theta)\rangle
  ={}&\bigl(V_A^{\mathrm{post}}\otimes V_B^{\mathrm{post}}\bigr)\,R_{ZZ}^{ab}(\phi)\\
  &\times\bigl(V_A^{\mathrm{pre}}E_A\otimes V_B^{\mathrm{pre}}E_B\bigr)|0^N\rangle .
\end{aligned}
  \label{eq:app-state}
\end{equation}

We use two prefix and two suffix sweeps.
A second sweep lets every pair block act again after all its neighbors have been dressed, which a single sequential sweep cannot provide, while the local depth stays bounded.
Each spin enters at most four pair blocks per sweep and a bounded number of encoder, anchor, and bridge gates, so the circuit is an FLDC in the sense of Definition~\ref{def:fldc}.

We reconstruct the inter-patch bridge $R_{ZZ}^{ab}(\phi)$
using the six-branch decomposition of
Ref.~\cite{MitaraiFujii2021VirtualGate}, illustrated in
Fig.~\ref{fig:experiment}(c). Each branch involves only
patch-local operations, and the target expectation value
is recovered as a weighted sum of the branch results.

Writing $c=\cos(\phi/2)$, $s=\sin(\phi/2)$, and
$d=cs=\frac12\sin\phi$, two branches apply either the
identity or a Pauli $Z$ to both bridge spins, with weights
$c^2$ and $s^2$, respectively. The remaining four branches
pair a signed $Z$ measurement on one bridge spin with a
rotation $R_Z(\pm\pi/2)$ on the other. For either choice
of measured spin, the positive and negative rotations
carry weights $d$ and $-d$, respectively. A signed
measurement means that the circuit continues after the
measurement and the corresponding shot contribution
is multiplied by the recorded eigenvalue $m=\pm1$.

The absolute branch weights sum to the knitting norm
\begin{equation}
\Gamma(\phi)
=c^2+s^2+4|d|
=1+2|\sin\phi|,
\label{eq:app-gamma}
\end{equation}
which depends only on the bridge angle and is independent
of the patch size. App.~\ref{app:cut-energy} combines the branches into the knitted estimator.

\figTierTwo

\subsection{Classical evaluation}
\label{subsec:app-sim}

This section describes the classical methods we use to estimate reference ground energies, the energies of the prepared circuits, and the gradients that train them.
We also perform knitted simulations and check the Tier-2 criteria before the device runs.

Because the bridge is $R_{ZZ}(\phi)=c\,\openone-is\,Z_aZ_b$, the prepared state splits into two product states,
\begin{equation}
  |\psi\rangle=c\,|A_0\rangle|B_0\rangle-is\,|A_1\rangle|B_1\rangle ,
\end{equation}
where $|p_1\rangle$ is the state of patch $p$ with a $Z$ inserted on its bridge spin and $|p_0\rangle$ the state without it.
The energy and its gradient therefore follow from quantities computed on one patch at a time, which lowers the cost from $2^N$ to $2^n$.
 
Up to $N=48$, we evaluate the circuit energy $E_{\boldsymbol\theta}=\langle\psi(\boldsymbol\theta)|H|\psi(\boldsymbol\theta)\rangle$ exactly, with state vectors of one patch.
Beyond that we use a matrix-product-state (MPS) evaluator at bond dimension $64$, which matches the exact energies at $N=48$ to about $10^{-3}$ per spin but does not reach this accuracy within our budget at $N=80$ for $h\ge0.25$ and at $N=120$.
Its cost grows polynomially with the bond dimension, but the bond dimension needed for a given accuracy grows with the width of a patch~\cite{Stoudenmire2012DMRG2D}, making the method not suitable for larger compact patches.
For the ground energy, $E_0$, we use the closed forms at $h=0$ and $h=1$ for all $N$, and exact diagonalization at every field up to $N=24$. 
Beyond that, we use the density-matrix renormalization group (DMRG)~\cite{White1992DMRG} at bond dimension $128$, whose energies change by at most $10^{-4}$ per spin between bond dimensions 64 and 128 at every field up to $N=142$.

The energy is knitted from the six branches of Sec.~\ref{subsec:app-model}, shown in Fig.~\ref{fig:experiment}(c).
Each subcircuit is read out in the all-$X$ and the all-$Z$ basis, which together cover every term of Eq.~\eqref{eq:app-H}.
The six branches, two patches, and two bases give $24$ subcircuit settings.
Each energy is estimated from $M=24{,}576$ shots in total, split among the 24 settings so that subcircuits with a larger variance receive more shots (App.~\ref{app:cut-energy}).
Our noiseless simulations of the subcircuits and the knitting protocol reproduce the exact energies of the full system within about one error bar.

Gradients are obtained from the exact evaluator by reverse-mode differentiation~\cite{jones2020gradients}.
This evaluator is used both to train the angles and to compute the gradient variances of criterion T3 up to $N=48$.
Since the MPS evaluator does not provide reliable gradients at larger sizes within our computational budget, training is restricted to $N\le48$.
The larger patches measured on the device therefore use the zero-field angles, which are exact at $h=0$ and require no training (App.~\ref{app:zf-angles}).
A parameter transfer that supplies trained angles to longer strips without optimization at the target size is described in App.~\ref{app:angles}, but not used in the hardware runs.

\subsection{Simulated checks of Tier-2 criteria}
\label{subsec:app-tier2}

\figHardware
\figContact 

\textit{Knitting cost} (C1).
In noiseless simulations of the protocol, the statistical error of the knitted energy per spin at fixed $M$ does not grow with $N$ at any field, and for $h>0$ it even decreases (Fig.~\ref{fig:tier2}(a)).
This is because the knitting norm of the single bridge, Eq.~\eqref{eq:app-gamma}, stays at most three whatever the patch size, while the energy per spin averages over more terms as the patches grow.
The trained bridge angles give $\Gamma(\phi)$ between $1.009$ and $1.535$, well below the worst case of three (trained angles of App.~\ref{app:optimization}).
The number of shots needed for a given accuracy per spin therefore does not increase with the size, so C1 holds.

\textit{Gradient variance} (T3).
We compute the variance of the gradient of the total energy at $h=0.25$ over uniformly random angles with the exact evaluator, following Ref.~\cite{ZhangLiuZhang2024FLDC}, for strips of $8$ to $38$ spins and compact patches of $24$ to $48$ spins\footnote{Strips have plaquettes dimensions $1\times k$ with $k=1,\dots,6$, giving $N=8$, 14, 20, 26, 32, and 38, and compact patches have $2\times2$, $2\times3$, $2\times4$, and $3\times3$ plaquettes, giving $N=24$, 34, 44, and 48.}.
The variance falls with $N$ but not exponentially over these sizes (Fig.~\ref{fig:tier2}(b)), which is an improvement over circuits whose local depths grow with the patch and whose variances fall by a fifth with every added spin (App.~\ref{app:variance}).
The remaining decrease is a boundary effect, due to the gates near the open ends of a strip, which have a larger variance than gates in the bulk.
Thus, longer strips, which add only bulk gates, have an average variance that levels off with increasing $N$, whereas compact patches, which have a larger share of bulk gates, have a lower variance that levels off only at larger sizes (App.~\ref{app:variance}).
Over the simulated sizes the gradients therefore remain resolvable with polynomially many shots, so T3 holds.
 
\textit{Seam entropy} (C2).
Criterion C2 asks whether the entanglement across the seam stays small as the system grows, since by Proposition~\ref{prop:cost-entropy} it sets a floor on the knitting cost.
We evaluate it for the ground state, which is what a single bridge has to support, and for the prepared states, which are what we knit.

For the ground state, we evaluate Eq.~\eqref{eq:seam-entropy} on its Schmidt spectrum from exact diagonalization, for $N\le24$.
Its seam entropy stays below about $0.2$ bits, and almost all of its weight lies on its two largest Schmidt components, which mirrors the single bridge construction.
For the prepared states, the seam entropy follows from the bridge angle $\phi$ and the expectation values $z_a$ and $z_b$ of $Z_a$ and $Z_b$ just before the bridge (App.~\ref{app:cut-seam}),
\begin{equation}
  \Ehalf(\rho_A)=\log_2\Bigl[1+|\sin\phi|\sqrt{(1-z_a^2)(1-z_b^2)}\,\Bigr].
  \label{eq:app-seam}
\end{equation}
The prepared states also use at most about a fifth of the one bit that Proposition~\ref{prop:port-capacity} allows for a single bridge, with almost no dependence on $N$ (Fig.~\ref{fig:tier2}(c)), so C2 holds.

\subsection{Hardware setup and results}
\label{subsec:app-knit}

We run the knitted circuits on the IBM Nighthawk processor \texttt{ibm\_phoenix} for two sets of sizes, with the angles and the knitting protocol of Sec.~\ref{subsec:app-sim}.
The first set, $N=8$, 14, 20, 24, 26, 32, and 34, uses the trained angles, and the second, the $3\times3$ to $6\times6$ patches with $N=48$, 62, 80, 98, 120, and 142, uses the zero-field angles at every field.
Every size is measured at $h\in\{0,0.1,0.2,0.25,0.3,0.5,1\}$ and $t=0.5$.
In the four branches that measure a bridge spin, the measurement is deferred onto one ancilla, so every subcircuit has $n+1$ qubits, from 5 to 85, with 38 to 2{,}995 two-qubit gates.

To suppress errors, we protect idle qubits against decoherence through dynamical decoupling sequences~\cite{Viola1999DD}, and two-qubit gates through Pauli twirling~\cite{Wallman2016RC}, which turns coherent gate errors into stochastic Pauli errors.
Additionally, we also post-select shots using the inversion test of Ref.~\cite{DragoiEgger2026Transpilation}, which discards shots with readout or relaxation errors at the end of the circuit (App.~\ref{app:device}).
Post-selection keeps most shots at $N=8$ but only a few percent at the largest sizes, so the number of shots that enter an estimate decreases with the size.
No further error mitigation is applied to the raw energies.

The hardware runs are designed as a classically verifiable test of whether an affordable bridge reconstruction recovers physically relevant information that running the two patches independently would miss.
While the measured energies assess how accurately the device reproduces the prescribed coupled circuit, it is the connected contact correlation that tests the information that a bridge-off product state lacks.
The experiment is thus a proof of principle for the value and the limits of bounded-interface reconstruction, not a demonstration of quantum advantage or a faster ground-state solver.

\textit{Energies.}
On the device, the knitted energies reproduce the field dependence of the ideal circuits at every size (Fig.~\ref{fig:hardware}(a,b)).
For the trained angles, they follow the ideal curves, including their maximum near the transition and their ordering by size, and for the zero-field angles they follow the exact linear dependence on $h$ (Eq.~\eqref{eq:ep-zf}).
The magnitudes, however, are reduced by noise.
For the trained angles, the ideal energies lie within $2\times10^{-3}$ per spin of the ground-state reference, so the preparation is accurate and the reduction comes from the device.
At $h=0$, the fraction of the energy signal that survives falls from about $0.8$ at $N=8$ to about $0.04$ at $N=142$, steadily with the routed two-qubit depth of the subcircuits (App.~\ref{app:depth}).
A supplementary rescaling calibrated at $h=0$ is assessed separately in App.~\ref{app:rescaling}, and it is not applied to the raw energies of Fig.~\ref{fig:hardware}.

The energy is, however, a poor measure of what the bridge contributes, since switching the bridge off changes the energy per spin by only $0.031/N$ at zero field, far below the device error (App.~\ref{app:ep-zero}).
The reconstructed dependence of the energy on the bridge angle shows the same limitation (App.~\ref{app:cut-angle}).

\textit{Contact correlation.}
We therefore measure the connected contact correlation of the final state
\begin{equation}
  C_{ab}=\langle Z_aZ_b\rangle-\langle Z_a\rangle\langle Z_b\rangle ,
  \label{eq:contact-corr}
\end{equation}
which vanishes for any product of patch states.
Independent execution thus predicts $C_{ab}=0$ whatever the local parameters, and a nonzero value reveals correlations only the bridge can create.
The energy per spin barely reflects it, since $C_{ab}$ changes it by only $-tC_{ab}/N$.
We knit its three terms from the same branch measurements and combine them into $C_{ab}$ (App.~\ref{app:cut-energy}).
At zero field, its target is $t/\sqrt{16+t^2}\approx0.124$ at every size (App.~\ref{app:ep-zero}).

On the device, knitting recovers a clearly nonzero contact correlation at every size up to $N=98$, which independent execution would miss entirely (Fig.~\ref{fig:contact}).
At $N=8$ the central value reaches about $87\%$ of the target, so knitting removes most of the error of independent execution.
The recovered fraction declines as the subcircuits grow deeper, to about $20\%$ at $N=80$ and 98, broadly following the loss of energy signal, and from $N=120$ on the correlation is no longer distinguishable from zero.
Bounded-interface reconstruction thus recovers a measurable part of the contact information on hardware, within limits that device noise tightens as the subcircuits grow.

\section{Discussion and outlook}
\label{sec:conclusion}

Assessed jointly, affordable knitting, classical difficulty of the branch tasks, and resolvable gradients narrow where a distributed quantum advantage can be sought.
An efficient cut whose branch tasks are classically easy is classically easy as a whole (Theorem~\ref{prop:stitching}), so a low knitting cost is never sufficient on its own.
For the cutting constructions and target states we consider, knitting stays affordable only when the subcircuits grow behind a bounded or logarithmically growing interface (Theorem~\ref{prop:regimes}), since every severed gate multiplies the shot count and even a generous budget pays for only a handful of them.
Among the eighteen assessed families, only one combines such a cut with a hardness result for its branch tasks, and that result holds in the worst case.
The outcomes of Sec.~\ref{sec:screening} fall into the four classes of Table~\ref{tab:classes}.

Local depth, interface size, and branch-task complexity require separate assessment.
A bounded interface controls reconstruction overhead, whereas finite local depth supports a gradient-variance bound only under the stated cost and initialization assumptions.
Extensive causal cones do not establish hardness, as the classically simulable examples illustrate.
\label{subsec:independence}

We tested the bounded-interface FLDC in its simplest form, two toric-code patches joined by a single variational gate and executed as separate subcircuits.
The experiment verifies what the bridge reconstruction recovers beyond running the patches independently, with a classical reference for every result.
This makes it a testbed for practical distributed quantum advantage, since it checks the necessary ingredient that correlations across the cut survive the reconstruction as the subcircuits grow. 
The results show that the favorable geometry behaves as the criteria predict.
In simulation, the single bridge keeps the knitting cost independent of the size, the prepared states use only a small part of the entanglement it can carry, and the gradients remain resolvable.
On the device, knitting recovers a contact correlation that independent execution misses up to 98 spins, and the energies reproduce the field dependence of the ideal circuits up to 142 spins, with a signal that noise attenuates as the routed subcircuits grow deeper.

Moving from these observations toward a distributed advantage requires addressing three aspects.
The first asks whether the circuits one would actually run are hard to simulate classically.
Our application has bounded seam and local depth, but the hardness of its branch tasks is not established.
Tensor-network calculations that stop converging at the hardware sizes are no evidence of hardness, and for bounded-interface subcircuits hardness is known only in the worst case over instances engineered to be hard (Theorem~\ref{thm:hardness}).
What is still missing is a physically motivated Hamiltonian whose bounded-interface circuits keep their branch tasks hard.

The second asks how many processors a single computation can usefully be split across.
Theorem~\ref{prop:regimes} keeps knitting affordable for at most logarithmically many subcircuits joined by single gates, and it is open whether any hard problem can be divided this way.
Splitting into more subcircuits gives every balanced grouping of them a long boundary, and the theorem makes this exponentially expensive for every exact protocol only if the seam entropy of the target grows as $\sqrt n$ across that boundary.
Excluding a family therefore needs such lower bounds, since a heuristic search for partitions can only show that a cheap cut exists, never that none does.

The third asks how to prepare good states once the patches outgrow classical simulation.
A single bridge keeps knitting affordable for any angles.
A vanishing bridge angle, however, also decouples the patches, so affordability alone says nothing about the quality of the prepared state.
Testing trainability and hardness together therefore needs scalable preparation and validation beyond classical references.
\label{subsec:open}

Candidate computations for distributed quantum advantage should therefore be assessed jointly by their reconstruction cost, the classical difficulty of the required branch tasks, and, for variational approaches, gradient resolution, since none of these three properties can be inferred from the others.

\tableClasses

\section{Acknowledgments}
SD acknowledges funding from the ETH Zurich Quantum Center.
MGG acknowledges funding from the EPSRC UK Quantum Technologies Programme under grant EP/T001062/1 and VeriQloud.
LP is supported by the National Research Foundation, Singapore through the National Quantum Office, hosted in A*STAR, under its Centre for Quantum Technologies Funding Initiative (S24Q2d0009). LP is also partially supported by A*STAR under its Young Investigator Research Grant (YIRG) M25N8c0131.

Claude Opus 5.5, Fable 5.1 (Anthropic) and GPT Astra (OpenAI) were used to draft and revise text, suggest references, write and debug analysis and plotting code, and check the manuscript for internal consistency.
The tools also proposed two parts of the additional data analysis in the appendices, namely the transfer of angles between patch sizes (App.~\ref{app:angles}), and the post-correction method used at $h=0$ (App.~\ref{app:rescaling}).
The authors reviewed all outputs and take full responsibility for the content.
The research questions, results, proofs, and experimental design are the authors'.

\section*{Data and code availability}

The code used to generate the circuits and analyze the data will be made publicly available upon publication and is available from the corresponding author upon request.

\newpage
\bibliography{refs}

\appendix
\onecolumngrid
\let\addcontentsline\fldcaddcontentsline
\tableofcontents

\section{Knitting cost and seam entanglement}
\label{app:cost}

This appendix proves Proposition~\ref{prop:cost-entropy} and Theorem~\ref{prop:regimes} and collects every bound on $\Gcut$ used by criteria C1 and C2.
The cost--entropy identity rests on the robustness of entanglement of pure states~\cite{VidalTarrach1999}, and the scaling regimes combine it with the norms of Eqs.~\eqref{eq:prelim-independent} and~\eqref{eq:prelim-joint}.
Throughout, $\rho_A=\Tr_B\rho$, logarithms are base two, and $S(\rho)=-\Tr\rho\log_2\rho$.

\subsection{The cost--entropy identity}
\label{app:cost-identity}

We define the cost of a cut at the level of the output state and relate it to the circuit-level norm $\Gcut$ and to the seam entropy of Eq.~\eqref{eq:seam-entropy}.

\begin{definition}[State-level cutting cost]
\label{def:gsep}
For a bipartite state $\rho_{AB}$,
\begin{equation}
  \Gsep(\rho)=\inf\Bigl\{\textstyle\sum_i|c_i|:\
  \rho=\sum_ic_i\,\sigma_i^A\otimes\tau_i^B,\
  \sigma_i^A,\tau_i^B\ \text{states}\Bigr\}.
  \label{eq:app-gsep-def}
\end{equation}
\end{definition}

\begin{lemma}[Hierarchy]
\label{lem:hierarchy}
Let a circuit with product input $|0^n\rangle$ have its cross-device
channels replaced by quasiprobability decompositions of the form of
Eq.~\eqref{eq:prelim-qpd} whose terms are local operations with classical
communication, possibly including measurements whose outcomes multiply the
estimator by a sign, and let $\Gcut$ be the infimum of their coefficient
norms. Then
\begin{equation}
  \Gsep(\rho_{\rm out})\le\Gcut\le\prod_a\gamma_a,
\end{equation}
where $\rho_{\rm out}$ is the output of the uncut circuit and $\gamma_a$
are the norms of independent per-channel decompositions.
\end{lemma}

\begin{proof}
Substituting a decomposition with coefficients $a_i$ into the circuit gives
$\rho_{\rm out}=\sum_ia_i(\rho_{i,+}-\rho_{i,-})$, where $\rho_{i,\pm}$ is
the unnormalized state prepared by branch $i$ when the product of its
outcome signs is $\pm1$. Each branch is an instrument whose outcomes sum to
a trace-preserving map, so $\Tr\rho_{i,+}+\Tr\rho_{i,-}=1$. It acts on a
product input by subcircuit-local circuits interleaved with local operations
and classical communication, so each $\rho_{i,\pm}$ is separable, that is,
a nonnegative combination of product states with weights summing to
$\Tr\rho_{i,\pm}$. Inserting these expansions gives a decomposition of the
form of Eq.~\eqref{eq:app-gsep-def} with coefficient norm $\sum_i|a_i|$.
Taking infima gives the first inequality; the product of independent
per-channel decompositions is one admissible decomposition, which gives the
second.
\end{proof}

\begin{theorem}[Exact pure-state cutting cost]
\label{thm:exact-gsep}
For any state $\rho$, $\Gsep(\rho)=1+2R(\rho)$, where $R$ is the robustness
of entanglement relative to separable states~\cite{VidalTarrach1999}. For a
pure state $\psi$ with Schmidt probabilities $\{\lambda_a\}$,
\begin{equation}
  \Gsep(\psi)=2\Bigl(\sum_a\sqrt{\lambda_a}\Bigr)^2-1
  =2\cdot2^{\Ehalf(\rho_A)}-1 .
  \label{eq:app-exact-gsep}
\end{equation}
\end{theorem}

\begin{proof}
Let $\rho=\sum_ic_i\,\sigma_i^A\otimes\tau_i^B$. Taking traces,
$\sum_ic_i=1$. Collecting positive and negative coefficients,
$\rho=p_+\omega_+-p_-\omega_-$ with separable states $\omega_\pm$,
$p_+-p_-=1$, and coefficient norm $p_++p_-=1+2p_-$. Conversely, any such
two-term form is a valid decomposition once $\omega_\pm$ are expanded into
products. Hence $\Gsep=1+2\min p_-$, and $\min p_-$ is by definition the
robustness $R(\rho)$: the least weight of a separable state $\omega_-$ such
that $(\rho+p_-\omega_-)/(1+p_-)$ is separable. For pure states Vidal and
Tarrach showed $R(\psi)=(\sum_a\sqrt{\lambda_a})^2-1$~\cite{VidalTarrach1999},
which gives Eq.~\eqref{eq:app-exact-gsep}.
\end{proof}

\begin{proof}[Proof of Proposition~\ref{prop:cost-entropy}]
The identity is Theorem~\ref{thm:exact-gsep}, and the inequality is
Lemma~\ref{lem:hierarchy}.
\end{proof}

For pure targets the identity reads $\Ehalf(\rho_A)=\log_2[(\Gsep+1)/2]$, so state-level cutting cost and seam entropy determine each other.
Attaining $\Gsep$ may require a decomposition that is hard to find, but the lower bounds below do not depend on finding it.

\begin{corollary}[Area-law no-go]
\label{cor:area-law}
Let $\{\rho_n\}$ be pure outputs of distributed circuits. If
$\Gsep(\rho_n)=\poly(n)$, in particular if $\Gcut=\poly(n)$, then
\begin{equation}
  S(\rho_A)\le\Ehalf(\rho_A)=\log_2\frac{\Gsep+1}{2}=O(\log n).
\end{equation}
Hence no family whose targets carry seam entanglement $S(\rho_A)=\Theta(L)$
with $L=\omega(\log n)$ is efficiently cuttable. Conversely,
$\Ehalf=O(\log n)$ implies $\Gsep=\poly(n)$.
\end{corollary}

\begin{proof}
Combine Lemma~\ref{lem:hierarchy} and Theorem~\ref{thm:exact-gsep} with
$S\le\Ehalf$, which holds because R\'enyi entropies do not increase with
their index.
\end{proof}

Subcircuits are executed unmodified, so the entanglement structure inside
each subcircuit, including any area law within it, is untouched by cutting;
the corollary concerns only the entanglement that reconstruction must
supply across the seam. Read in the other direction, a cut of $\kappa$
channels with per-channel norm at most $\bar\gamma$ gives
$\Ehalf(\rho_A)\le\log_2[(\bar\gamma^\kappa+1)/2]\le\kappa\log_2\bar\gamma$,
so a bounded interface forces bounded seam entropy.

Subcircuits run unmodified, so the entanglement inside each of them is untouched by cutting.
The corollary concerns only the entanglement that knitting must supply across the seam.
Read in the other direction, a cut of $\kappa$ gates with per-gate norm at most $\bar\gamma$ gives $\Ehalf(\rho_A)\le\log_2[(\bar\gamma^\kappa+1)/2]\le\kappa\log_2\bar\gamma$, so a bounded interface forces bounded seam entropy.

\begin{example}[Long-range entanglement is not the obstruction]
\label{ex:ghz}
The $n$-qubit GHZ state split into two halves has Schmidt probabilities
$\{\frac12,\frac12\}$, hence $\Gsep=2(\sqrt{1/2}+\sqrt{1/2})^2-1=3$,
independently of $n$. More generally, any state of Schmidt rank $R$ across
the seam has $\Gsep\le2R-1$. Long-range entanglement alone does not
determine the cost; seam capacity does.
\end{example}

A cheap cut constrains the whole Schmidt spectrum across the seam, not only its entropy.
\begin{corollary}[Cheap seams are low-rank seams]
\label{cor:lowrank}
Let $\psi_{AB}$ be pure with Schmidt coefficients $\{\sqrt{\lambda_a}\}$ across $A{:}B$, and let $\varepsilon\in(0,1]$.
There is a set $K$ of Schmidt indices with
\begin{equation}
  |K|\le\frac{2^{\Ehalf(\rho_A)}}{\varepsilon^2}
  =\frac{\Gsep(\psi)+1}{2\varepsilon^2}
  \le\frac{\Gcut+1}{2\varepsilon^2}
  \label{eq:lowrank}
\end{equation}
such that the normalized truncation $\psi_K$ of $\psi$ to $K$ satisfies $\lVert\psi-\psi_K\rVert_1\le2\varepsilon$.
Hence $|\langle O\rangle_\psi-\langle O\rangle_{\psi_K}|\le2\varepsilon\lVert O\rVert$ for every observable, and an efficiently cuttable family has a seam of polynomial effective Schmidt rank at any fixed accuracy.
\end{corollary}
\begin{proof}
Write $F=\sum_a\sqrt{\lambda_a}=2^{\Ehalf/2}\ge1$ and keep $K=\{a:\lambda_a\ge\eta\}$ with $\sqrt\eta=\varepsilon^2/F\le1$.
Every kept index has $\sqrt{\lambda_a}\ge\sqrt\eta$, and the $\sqrt{\lambda_a}$ sum to $F$, so $|K|\le F/\sqrt\eta=F^2/\varepsilon^2$.
The discarded weight obeys $\sum_{a\notin K}\lambda_a\le\sqrt\eta\sum_a\sqrt{\lambda_a}=\varepsilon^2$, so the retained weight is $p\ge1-\varepsilon^2$ and $|\langle\psi|\psi_K\rangle|=\sqrt p$.
For pure states $\lVert\psi-\psi_K\rVert_1=2\sqrt{1-p}\le2\varepsilon$.
The remaining forms of Eq.~\eqref{eq:lowrank} follow from Eq.~\eqref{eq:cost-entropy} and $\Gsep\le\Gcut$.
\end{proof}
The corollary concerns the state, not an algorithm.
Estimating $\langle O_A\otimes O_B\rangle$ from the truncated state would require transition elements between Schmidt vectors, which the branch tasks do not provide.
The stitching theorem below avoids the Schmidt basis by reusing the cutting decomposition, whose branches are already products of subcircuit states.

\subsection{Bounded interfaces}
\label{app:cost-bounded}

The single bridge of Sec.~\ref{sec:application} is the case $\kappa=1$ of the following two results, which bound the seam entropy and the knitting cost of any circuit whose subcircuits interact through $\kappa$ Pauli rotations.

\begin{proposition}[Capacity bounds seam entropy]
\label{prop:port-capacity}
Let $U$ consist of subcircuit-local gates and $\kappa$ contact gates $\exp[-i\phi_jP^A_jP^B_j/2]$, with $P^A_j$ and $P^B_j$ Pauli strings supported in $A$ and $B$, applied to a product input.
The state of Eq.~\eqref{eq:app-state} has this form, with $\kappa=1$.
The output then has Schmidt rank at most $2^{\kappa}$ across $A{:}B$, and
\begin{equation}
  \Ehalf(\rho_A)\le\kappa,
  \qquad
  \Gsep(\psi)\le2\cdot2^{\kappa}-1,
  \label{eq:capacity}
\end{equation}
uniformly in the angles $\phi_j$, in $n$, and in the subcircuit-local gates.
\end{proposition}
\begin{proof}
Each contact gate equals $\cos(\phi_j/2)\,\openone-i\sin(\phi_j/2)\,P^A_j\otimes P^B_j$, an operator of Schmidt rank at most two across $A{:}B$, so it at most doubles the Schmidt rank of the state.
Subcircuit-local gates leave the rank unchanged, and the product input has rank one, so the output rank is $R\le2^{\kappa}$.
The Cauchy--Schwarz inequality gives $\sum_{a=1}^{R}\sqrt{\lambda_a}\le\sqrt R$, hence $\Ehalf\le\log_2R\le\kappa$, and Eq.~\eqref{eq:cost-entropy} gives the bound on $\Gsep$.
\end{proof}
The bound follows from the operator Schmidt rank of the contact gates alone, so it holds before any decomposition is chosen.
Corollary~\ref{cor:area-law} reaches a similar bound from the opposite direction, through the cutting cost.

\begin{proposition}[Cost of cutting bounded interfaces]
\label{prop:port-cost}
The contact gate $\exp[-i\phi\,P^AP^B/2]$ admits a local quasiprobability
decomposition with $\gamma(\phi)=1+2|\sin\phi|$, optimal among local
decompositions with or without classical
communication~\cite{MitaraiFujii2021VirtualGate,PiveteauSutter2024Knitting}. Cutting $\kappa$
contacts independently therefore costs
\begin{equation}
  \Gcut(\bm\phi)\le\prod_{j=1}^\kappa(1+2|\sin\phi_j|)
  \le\min\bigl\{3^\kappa,\ e^{2\sum_j|\phi_j|}\bigr\},
\end{equation}
and near a warm start, $\phi_j\to0$, the overhead tends to one: weakly
coupled contacts are asymptotically free, consistently with the seam
entropy they generate also vanishing with $\bm\phi$.
\end{proposition}

\subsection{Transparent toric-code seam}
\label{app:cost-transparent}

We now turn to the opposite geometry, a single toric-code patch divided by a straight cut across which anyons can pass anywhere along the seam.
This is the seam of rows 8b and 12a of Table~\ref{tab:families}.
Two lemmas turn a witness operator into a lower bound on the coefficient norm.
Applied to the stabilizers straddling the seam, they give an exponential lower bound certified by operators of bounded support, together with the exact cost of the seam.

\begin{lemma}[Witness ratio]
\label{lem:witness}
For any Hermitian $W$,
\begin{equation}
  \Gsep(\rho)\ge\frac{|\Tr(W\rho)|}
  {\sup_{\sigma^A,\tau^B}|\Tr[W(\sigma^A\otimes\tau^B)]|}.
\end{equation}
The same bound holds for the coefficient norm of any signed product
decomposition of any operator whose $W$-expectation matches $\Tr(W\rho)$ to
additive error $\varepsilon$, with numerator $|\Tr(W\rho)|-\varepsilon$.
\end{lemma}

\begin{proof}
$|\Tr(W\rho)|=|\sum_ic_i\Tr[W(\sigma_i\otimes\tau_i)]|
\le\sum_i|c_i|\,\sup|\Tr[W(\sigma\otimes\tau)]|$.
\end{proof}

\begin{lemma}[Maximally entangled overlap]
\label{lem:overlap}
Let $\Phi_D$ be maximally entangled of Schmidt rank $D$ across $A{:}B$. For
any state $\omega$ of Schmidt number $r$, in particular any product state
with $r=1$, $\langle\Phi_D|\omega|\Phi_D\rangle\le r/D$.
\end{lemma}

\begin{proof}
For a pure state $|\varphi\rangle=\sum_{b\le r}s_b|u_b\rangle|v_b\rangle$,
each $|\langle\Phi_D|u_bv_b\rangle|\le D^{-1/2}$, so
$|\langle\Phi_D|\varphi\rangle|\le D^{-1/2}\sum_bs_b\le\sqrt{r/D}$ by the
Cauchy--Schwarz inequality. Mixed states of Schmidt number $r$ are
mixtures of such states.
\end{proof}

Consider a planar toric-code patch of $W_p\times H_p$ plaquettes with
qubits on edges, a star operator on every vertex, truncated at the
boundary, and a plaquette operator on every face, so that its stabilizers
fix a unique state $\psi$.

\begin{definition}[Straddling stabilizer group]
\label{def:straddling}
For a vertical seam through vertex column $c$, $0<c<W_p$, assign to $A$
every qubit whose edge midpoint has $x<c$ and to $B$ the rest, and let
$L=H_p+1$ be the number of seam vertices. Let $\mathcal G_\partial$ be the
set of star and plaquette operators whose support meets both $A$ and $B$,
and define the strip witness
\begin{equation}
  W=\prod_{g\in\mathcal G_\partial}\frac{\openone+g}{2}.
  \label{eq:app-strip-witness}
\end{equation}
\end{definition}

\begin{theorem}[Strip-local robust seam witness]
\label{thm:seam-witness}
For the state $\psi$ of Definition~\ref{def:straddling}:
\begin{enumerate}[label=(\roman*),leftmargin=*,nosep]
\item \emph{(locality)} $|\mathcal G_\partial|=2L-1$, consisting of the $L$
 stars centered on the seam and the $L-1$ plaquettes immediately to its left.
 Every element acts on at most four qubits within one lattice spacing of
 the seam and is localized at a single seam position, and all elements
 commute, so $W$ is a product of $2L-1$ commuting projectors of bounded
 support;
\item \emph{(normalization)} $\Tr(W\psi)=1$;
\item \emph{(separable ceiling)}
 $\sup_{\sigma^A,\tau^B}\Tr[W(\sigma^A\otimes\tau^B)]\le2^{-(L-1)}$;
\item \emph{(robust bound)} any signed product decomposition, in particular
 any cutting protocol of Lemma~\ref{lem:hierarchy}, whose reconstruction
 reproduces $\langle W\rangle$ to additive error $\varepsilon$ obeys
 $\Gamma\ge(1-\varepsilon)\,2^{L-1}$.
\end{enumerate}
Moreover $W=2^{-(2L-1)}\sum_{g\in\langle\mathcal G_\partial\rangle}g$ is the
uniform average of the $2^{2L-1}$ products of seam-crossing star and
plaquette operators, the closed $e$- and $m$-string configurations threading
the seam, so reproducing all of those correlators to error $\varepsilon$
reproduces $\langle W\rangle$ to error $\varepsilon$.
\end{theorem}

\begin{proof}
(i) The stars centered at the $L$ seam vertices have one edge in $A$ and the
rest in $B$; the plaquettes in the column immediately left of the seam have
three edges in $A$ and one in $B$; every other star or plaquette lies
wholly on one side. Stars and plaquettes share an even number of edges, so
$\mathcal G_\partial$ is abelian. (ii) Each $g$ stabilizes $\psi$.

(iii) Define the alternating form $\omega(g,h)$ on the $\mathbb F_2$ space
$\langle\mathcal G_\partial\rangle$ recording whether the restrictions
$g|_A$ and $h|_A$ anticommute. The star at seam row $y$ restricts to $X$ on
the single horizontal edge $e_y$ entering the seam at that row, and the
plaquette at row $y'$ restricts to $Z$ on $e_{y'}$, $e_{y'+1}$ and one
vertical edge. Hence
$\omega(\mathrm{star}_y,\mathrm{plaq}_{y'})=\delta_{y,y'}+\delta_{y,y'+1}$,
a bidiagonal $L\times(L-1)$ matrix of rank $L-1$, while stars and
plaquettes restrict to $X$-type and $Z$-type operators that commute among
themselves. So $\omega$ has rank $2(L-1)$. Choose a symplectic basis
$(u_i,v_i)_{i=1}^{L-1}$ for $\omega$ and a complement spanning its radical,
so that
\begin{equation}
  W=\Pi_0\prod_{i=1}^{L-1}\frac{(\openone+u_i)(\openone+v_i)}{4}
  \le\prod_{i=1}^{L-1}\frac{(\openone+u_i)(\openone+v_i)}{4},
\end{equation}
with $\Pi_0$ the commuting projector associated with the radical. The
$A$-restrictions of $\{u_i,v_i\}$ obey the commutation relations of Pauli
$X$ and $Z$ on $L-1$ qubits, and because $u_i$ and $v_i$ commute, so do the
$B$-restrictions. A local Clifford unitary $V_A\otimes V_B$ therefore maps
the right-hand side to $\Phi_{2^{L-1}}\otimes\openone$ with
$\Phi_{2^{L-1}}$ a maximally entangled projector~\cite{Fattal2004Stabilizer}.
Local unitaries map products to products, and partial traces of products
are products, so Lemma~\ref{lem:overlap} gives the ceiling $2^{-(L-1)}$.
(iv) Lemma~\ref{lem:witness}.
\end{proof}

The symplectic basis used in (iii) is in general supported along the whole
seam, but it enters only the proof of the ceiling. By (i) the witness
itself is a product of commuting projectors of bounded support, each
attached to one seam site; a disentangling circuit never has to be
constructed, only shown to exist.

\begin{corollary}[Exact transparent-seam cost]
\label{cor:toric-exact}
For the state of Definition~\ref{def:straddling}, $\rho_A$ has a flat
spectrum of rank $2^{L-1}$, so $\Gsep(\psi)=2\cdot2^{L-1}-1=2^L-1$,
and every cutting protocol satisfies $\Gcut\ge2^L-1$, with a sufficient
shot count scaling as $\Gcut^2\sim4^L$. More generally, $k$ Bell pairs have
$\Gsep=2^{k+1}-1$, which coincides with the coefficient norm of the optimal
joint cut of $k$ wires with classical communication~\cite{Brenner2023WireCutting},
so that protocol is optimal already at the state level.
\end{corollary}

\begin{proof}
For a stabilizer state, $\rho_A$ is proportional to a projector and
$S(\rho_A)$ is half the $\mathbb F_2$ rank of the anticommutation form of
the $A$-restrictions of a full set of stabilizer
generators~\cite{Fattal2004Stabilizer}. A generator supported wholly in $A$
or in $B$ has a trivial form with every other generator, so only
$\mathcal G_\partial$ contributes, and step (iii) of
Theorem~\ref{thm:seam-witness} gives $S(\rho_A)=L-1$ with a flat
spectrum. Then $\Ehalf=L-1$, and Theorem~\ref{thm:exact-gsep} and
Lemma~\ref{lem:hierarchy} give the claims. For $k$ Bell pairs,
$(\sum_a\sqrt{\lambda_a})^2=2^k$.
\end{proof}

\begin{corollary}[Seam-crossing loop degradation]
\label{cor:loop-degradation}
Any signed product decomposition of coefficient norm $\Gamma$ assigns the
witness a value $\langle W\rangle_{\rm rec}\le\Gamma\,2^{-(L-1)}$. At fixed
subexponential budget, the joint correlator of the seam-crossing string
algebra therefore cannot be reproduced beyond a value exponentially small in
$L$, whereas observables supported in one device are reproduced without
overhead, since subcircuits run unmodified.
\end{corollary}

Because $W$ is a witness rather than a fidelity, it applies away from the
fixed point. The overlap $|\langle\psi_0|\psi_h\rangle|^2$ between the
fixed-point and field-perturbed states decays exponentially in the total
system size, whereas $\Tr(W\rho_h)$ is controlled by a strip of width $O(1)$
around the seam.

\begin{corollary}[Perturbed reduction]
\label{cor:perturbed}
For any state $\rho$ on the patch, a decomposition reproducing
$\langle W\rangle_\rho$ has $\Gamma\ge\Tr(W\rho)\,2^{L-1}$. Hence if the
field-perturbed ground state obeys $\Tr(W\rho_h)\ge e^{-c(h)L}$, then
$\Gamma\ge2^{(1-c(h)/\ln2)L-1}$, and the exponential bound survives at
every field with $c(h)<\ln2$.
\end{corollary}

\begin{proof}
Lemma~\ref{lem:witness} with numerator $\Tr(W\rho)$ and
Theorem~\ref{thm:seam-witness}(iii).
\end{proof}

\subsection{Scaling regimes and affordable interface sizes}
\label{app:cost-regimes}

This subsection proves Theorem~\ref{prop:regimes}.
Let $\{U_n\}$ be cut into $D(n)$ subcircuits of width $w(n)$, severing $K(n)$ gates in total.
Two lower bounds on $\Gcut$ apply, one for rotation gates and one for every exact protocol.

For the first, let the severed gates be two-qubit rotations with $|\sin\theta_s|\ge c>0$ on a fraction $f>0$ of them.
Independent cutting gives $\Gcut=\prod_{s=1}^{K}(1+2|\sin\theta_s|)\ge(1+2c)^{fK}=2^{\Omega(K)}$ by Eq.~\eqref{eq:prelim-independent}.
When the time-slice condition holds, the joint decomposition of Eq.~\eqref{eq:prelim-joint} gives $\Gcut\ge2(1+c)^{fK}-1=2^{\Omega(K)}$.
A connected arrangement of $D$ subcircuits has $K\ge D-1$, so both costs are $2^{\Omega(D)}$.

For the second, group the subcircuits into two halves $A$ and $B$ and write $s(n)=\Ehalf(\rho_A)$ for the pure target.
Proposition~\ref{prop:cost-entropy} gives $\Gcut\ge2\cdot2^{s(n)}-1$ for every exact protocol.

\begin{proof}[Proof of Theorem~\ref{prop:regimes}]
If $D=O(1)$ and $w=\Theta(n)$, the total interface $K$ may stay bounded, and then $\Gcut\le3^K=O(1)$ by Eq.~\eqref{eq:prelim-independent}.
The asymptotics in $n$ then describe subcircuits that grow behind a fixed interface.
If $w=O(1)$ and $D=\Theta(n/w)$, then $K\ge D-1=\Omega(n/w)$, and the first bound gives $2^{\Omega(D)}$ under its angle condition.
Independently, a target with $s(n)=\Omega(\sqrt n)$ across a balanced grouping has $\Gcut=2^{\Omega(\sqrt n)}$ by the second bound.
\end{proof}

The entropy hypothesis holds, for example, for a geometrically local two-dimensional target whose entanglement grows with the boundary.
Every balanced grouping of the subcircuits then has a boundary $\partial$ of $\Omega(\sqrt n)$ qubits, and $S(\rho_A)\ge a|\partial|-b$ with $a>0$ gives $s(n)\ge S(\rho_A)=\Omega(\sqrt n)$, since $\Ehalf\ge S$.
An area-law upper bound alone does not imply the hypothesis, which is why Theorem~\ref{prop:regimes} states it explicitly.

We now evaluate Eq.~\eqref{eq:kmax} for representative budgets.
We take $\varepsilon=0.01$ and $\delta=0.05$, for which Eq.~\eqref{eq:gamma-budget} gives $\Gamma_{\max}\approx3.68\times10^{-3}\sqrt M$.
Table~\ref{tab:kmax} lists the resulting counts.

\tableKmax

Each additional severed gate costs a factor of nine in shots for independent decompositions and a factor approaching four for joint ones.
The affordable count therefore grows only as $\frac12\log_2M$ and remains a small integer for any realistic budget.
Interface gates interleaved through the circuit, rather than confined to one time slice, must use the independent count.
For the single bridge gate of Sec.~\ref{sec:application}, $\kappa=1$ and the joint and independent norms coincide.

\section{Classical simulability}
\label{app:sim}

This appendix proves Theorem~\ref{prop:stitching}, which underlies criterion S1, and gives the two sufficient conditions for efficient simulation that criteria S2 and S1 use.
{The theorem formalizes an observation of Refs.~\cite{PiveteauSutter2024Knitting,BravyiSmithSmolin2016}, that knitting classically simulable subcircuits yields a classically simulable estimate, with the branch tasks as the unit of hardness.

\subsection{Knitting easy subcircuits is easy}
\label{app:sim-stitching}

This subsection proves Theorem~\ref{prop:stitching}.
Each branch of a cutting decomposition without classical communication is a product of subcircuit states, so a global expectation value is a signed sum of products of subcircuit expectation values.

\begin{theorem}[Stitching]
\label{thm:stitch}
Let $U_n$ carry a cut of $\kappa(n)$ gates, each replaced by at most $m=O(1)$ local terms of total coefficient norm at most $\bar\gamma$, none requiring classical communication.
Then
\begin{equation}
  \rho_{\mathrm{out}}=\sum_\alpha c_\alpha\,\sigma^A_\alpha\otimes\tau^B_\alpha,
  \qquad
  \sum_\alpha|c_\alpha|=\Gamma\le\bar\gamma^{\kappa},
  \qquad
  \#\{\alpha\}\le m^{\kappa},
  \label{eq:stitch-branches}
\end{equation}
where $\sigma^A_\alpha$ and $\tau^B_\alpha$ are prepared by the two subcircuits with the operations of branch $\alpha$ inserted and weighted by the signs of any measurement outcomes.
Let $O$ have support at most $r=O(1)$ and $\lVert O\rVert\le1$.
Suppose the branch tasks, estimating $\Tr(O'\sigma^A_\alpha)$ and $\Tr(O''\tau^B_\alpha)$ to additive error $\eta$ for any $\alpha$ and any $O'$, $O''$ of support at most $r$ and norm at most one, are solvable classically in time $T(n,1/\eta)$.
Then $\langle O\rangle$ can be estimated to additive error $\varepsilon$ in time
\begin{equation}
  O\bigl(4^r m^{\kappa}\,T(n,3\cdot4^r\Gamma/\varepsilon)\bigr).
\end{equation}
In particular, if $\kappa=O(\log n)$ and both branch tasks are in \textsf{BPP}, then estimating $\langle O\rangle$ to inverse-polynomial additive error is in \textsf{BPP}.
\end{theorem}
\begin{proof}
Expand $O=\sum_j o_j P^j_A\otimes P^j_B$ in the Pauli basis of its support, with at most $4^r$ terms and $|o_j|\le\lVert O\rVert\le1$.
By Eq.~\eqref{eq:stitch-branches}, $\langle O\rangle=\sum_j o_j\sum_\alpha c_\alpha\Tr(P^j_A\sigma^A_\alpha)\Tr(P^j_B\tau^B_\alpha)$.
Estimate each of the $2\cdot4^r m^{\kappa}$ factors to additive error $\eta\le1$.
The exact factors lie in $[-1,1]$, so each product carries error at most $2\eta+\eta^2\le3\eta$, and the total error is at most $3\cdot4^r\Gamma\eta$.
Choosing $\eta=\varepsilon/(3\cdot4^r\Gamma)$ gives error $\varepsilon$, and the running time is the number of factors times the cost of each.
For randomized subroutines, repeating each call $O(\log(4^r m^{\kappa}/\delta))$ times and taking the median makes all estimates succeed together with probability at least $1-\delta$.
If $\kappa=O(\log n)$, then $m^{\kappa}$ and $\Gamma\le\bar\gamma^{\kappa}$ are polynomial, so the branch tasks are called polynomially often at inverse-polynomial precision.
\end{proof}

\begin{proof}[Proof of Theorem~\ref{prop:stitching}]
This is the last statement of Theorem~\ref{thm:stitch}.
\end{proof}

\begin{corollary}[The trench coat]
\label{cor:trenchcoat}
Under the hypotheses of Theorem~\ref{thm:stitch} with $\kappa=O(\log n)$, if estimating local expectation values of $\{U_n\}$ is not in \textsf{BPP}, then the branch task of at least one subcircuit is not in \textsf{BPP}.
A family that is efficiently cuttable by such a protocol, and whose subcircuits are both classically easy, is therefore classically easy as a whole.
\end{corollary}

The theorem uses only the branch tasks of criterion S1, with no Schmidt basis and no transition amplitudes.
It covers structured protocols with boundedly many local terms per severed gate and no classical communication, which includes the six-branch decomposition of the bridge gate used in Sec.~\ref{sec:application}.
Allowing classical communication with a message alphabet of constant size per gate multiplies the branch count by a further factor $O(1)^{\kappa}$ and leaves the conclusion unchanged.
A protocol reaching $\Gamma=\poly(n)$ with superpolynomially many branches is not covered.
Corollary~\ref{cor:lowrank} still shows that its seam has polynomial effective Schmidt rank, but turning that into a simulation would require transition elements, so this case remains open.

\subsection{Sufficient conditions for easiness}
\label{app:sim-easy}

\begin{proposition}[Light-cone easiness]
\label{prop:lightcone}
Let $U$ be a geometrically local circuit on a two-dimensional lattice with
conventional depth $d$, and let $O$ be an observable supported in a region
of diameter $r$ with $\|O\|\le1$. Then
$\langle0^n|U^\dagger OU|0^n\rangle$ can be computed in time
$2^{O((r+d)^2)}$, independent of $n$. In particular, constant-depth
two-dimensional families fail S2.
\end{proposition}

\begin{proof}
The backward light cone of $\supp(O)$ under $d$ layers of geometrically
local gates is contained in a disk of radius $O(r+d)$, hence contains
$O((r+d)^2)$ qubits. Gates outside the cone cancel between $U$ and
$U^\dagger$, and a state-vector simulation of the cone circuit on its
product input gives the result.
\end{proof}

The argument uses only the number $c$ of qubits in the backward light cone
of $\supp(O)$ in the gate DAG, and gives cost $\poly(n)\,2^{O(c)}$ for any
circuit; this is the quantity S2 measures. A sequential FLDC can have
conventional depth $\Theta(n)$ and cones that grow with $n$, so bounded
local depth does not imply a bounded cone.

\begin{proposition}[Quasi-one-dimensional easiness]
\label{prop:quasi1d}
Let a circuit of two-qubit gates act on a $w\times n$ strip, with gates
crossing any vertical cut at most $t$ times in total. Then the output has
Schmidt rank at most $4^t$ across every vertical cut, admits a
matrix-product description of bond dimension $2^{O(w+t)}$, and its local
expectation values can be computed classically in time
$\poly(n)\,2^{O(w+t)}$. For depth $d$, $t\le wd$, so $wd=O(\log n)$ implies
polynomial cost. Constant-$p$ QAOA on graphs of bounded maximum degree $\Delta$ is
likewise easy for local observables, since the backward causal cone of an $r$-local
observable contains at most $r\Delta^{O(p)}$ qubits~\cite{Farhi2014QAOA}.
\end{proposition}

\begin{proof}
A two-qubit gate crossing a cut has operator Schmidt rank at most four
across it, so it multiplies the Schmidt rank of the state by at most four,
and gates not crossing the cut leave it unchanged; this gives $4^t$ at
every time. Ordering the qubits column by column, a cut inside a column
exceeds the adjacent vertical cut by at most $w$ qubits, which gives the
bond dimension, and a time-evolving matrix-product simulation~\cite{Vidal2003MPSClassical}
gives the running time. The QAOA statement is the light-cone argument of
Proposition~\ref{prop:lightcone} applied on the interaction graph.
\end{proof}

\section{Trainability}
\label{app:train}

This appendix states the finite-local-depth gradient bound behind criteria T1 and T2 with its hypotheses, gives the sufficient condition for a barren plateau behind the fail conditions of T2 and T3, and shows that gradients can be estimated through a cut without losing the bound.

\subsection{The gradient bound and its hypotheses}
\label{app:train-bound}

Recall from Definition~\ref{def:fldc} that $U(\boldsymbol\theta)=U_M(\boldsymbol\theta_M)\cdots U_1(\boldsymbol\theta_1)$ consists of blocks that act on at most $\beta$ qubits each, with at most $\chi$ blocks acting on any one qubit.
For
\[
  C(\boldsymbol\theta)
  =\sum_\nu\lambda_\nu\langle O_\nu\rangle_{\boldsymbol\theta},
  \qquad |\operatorname{supp}(O_\nu)|\le r=O(1),
\]
and a parameter in block $U_{k(\mu)}$, the bound used in
Sec.~\ref{sec:preliminaries} is
\begin{equation}
\begin{split}
  \operatorname{Var}_{\boldsymbol\theta}[\partial_\mu C]
  &\ge g,\\
  g&=2\,4^{-r\chi\beta}
    \sum_{\substack{\nu:\,\operatorname{supp}(O_\nu)\,\cap\\
      \operatorname{supp}(U_{k(\mu)})\ne\emptyset}}\lambda_\nu^2.
\end{split}
\label{eq:prelim-gradient}
\end{equation}
The bound holds under the hypotheses of Zhang, Liu, and
Zhang~\cite{ZhangLiuZhang2024FLDC}: the blocks are drawn independently from
local $2$-designs, $\theta_\mu$ parametrizes a gate inside block
$U_{k(\mu)}$, and only cost terms whose support overlaps that block
contribute to the sum. The bound $g$ depends on the parameter $\mu$ even
though that index is suppressed in the notation.

The quantities $\beta$ and $\chi$ in Eq.~\eqref{eq:prelim-gradient} need some care.
They are not the circuit depth.
A sequential FLDC can have conventional depth $\Theta(n)$ while $\beta$ and $\chi$ stay bounded, as for the toric-code circuits of Table~\ref{tab:families}, and this gap is what the FLDC construction exploits.
They are also not intrinsic to a circuit, which has no canonical block decomposition.
The finest choice, one block per gate, gives $\beta=2$ with $\chi$ equal to the largest number of gates on any qubit, and coarser blocks trade a larger $\beta$ for a smaller $\chi$.
Criterion T1 therefore records both quantities for a stated decomposition.
Finally, every circuit has some $\beta$ and $\chi$, and an FLDC is distinguished only by both remaining bounded as $n$ grows.

\subsection{Barren plateaus from unitary designs}
\label{app:train-plateau}

\begin{proposition}[Design-induced barren plateau]
\label{prop:design-plateau}
Let $C(\theta)=\Tr[OU\rho_0U^\dagger]$ with $\|O\|\le1$, and write
$U=V_2e^{-i\theta G/2}V_1$ with $G$ a Pauli string. If $V_1$ or $V_2$ is
drawn from a unitary $2$-design on the full $n$-qubit register,
independently of the other, then
\begin{equation}
  \mathbb E[\partial_\theta C]=0,
  \qquad
  \operatorname{Var}[\partial_\theta C]\le\frac{\Tr(O^2)}{4^n-1}
  \le\frac{2^n}{4^n-1}=2^{-\Omega(n)}.
  \label{eq:app-design-plateau}
\end{equation}
Consequently, FLDC families whose blocks are enlarged until the circuit, or
the sub-circuit on either side of a parameter, forms a global $2$-design
lose trainability, and the guarantee of Eq.~\eqref{eq:prelim-gradient},
which carries the factor $4^{-r\chi\beta}$, becomes vacuous as soon as
$r\chi\beta=\omega(\log n)$.
\end{proposition}

\begin{proof}
Write $d=2^n$, $R=e^{-i\theta G/2}$, $O'=V_2^\dagger OV_2$ and
$\sigma'=RV_1\rho_0V_1^\dagger R^\dagger$. Differentiating,
$\partial_\theta C=\Tr(A\sigma')=\Tr(O'B)$ with the traceless Hermitian
operators $A=\frac i2[G,O']$ and $B=\frac i2[\sigma',G]$, which obey
$\Tr(A^2)\le\Tr(O^2)$ and $\Tr(B^2)\le\Tr(\sigma'^2)\le1$. For $V$ drawn
from a $2$-design and traceless $Y$, the standard second-moment
formulas~\cite{McClean2018BarrenPlateaus} give $\mathbb E[\Tr(VXV^\dagger Y)]=0$ and
$\mathbb E[\Tr(VXV^\dagger Y)^2]
=\Tr(Y^2)[\Tr(X^2)-\Tr(X)^2/d]/(d^2-1)\le\Tr(X^2)\Tr(Y^2)/(d^2-1)$. If $V_1$
is the design, apply this with $V=RV_1$, $X=\rho_0$ and $Y=A$; if $V_2$ is,
apply it with $V=V_2^\dagger$, $X=O$ and $Y=B$. Both cases give
Eq.~\eqref{eq:app-design-plateau}.
\end{proof}

The vacuous case is the fail condition of T2, and the straight line on a log plot is the fail condition of T3.

\subsection{Gradient estimation through a cut}
\label{app:train-cut}

The following result shows that cutting leaves the landscape and its gradient bound unchanged, and that the sampling cost of a gradient component is set by the interface norm.
It is the guarantee that the diagnostic of App.~\ref{app:variance} would test.

\begin{theorem}[Cutting guarantee]
\label{thm:fldc-train}
Let $U$ be a distributed FLDC with block support $\beta$ and local depth
$\chi$, cut through $\kappa$ channels with per-channel quasiprobability
norms $\gamma_a\le\bar\gamma$, and let $O$ be an observable with
$\|O\|\le1$. The product estimator
$\widehat o=\Gamma\,\operatorname{sgn}(a_{\bm\alpha})\,x$, with
$\Gamma=\prod_a\gamma_a\le\bar\gamma^\kappa$, is unbiased with
$|\widehat o|\le\Gamma$, so Eq.~\eqref{eq:prelim-shots} bounds the number of shots that resolve $\langle O\rangle$ to error $\varepsilon$.
A gradient component given by the parameter-shift rule,
$\partial_\mu C=\frac12[C(\theta_\mu+\frac\pi2)-C(\theta_\mu-\frac\pi2)]$,
is estimated to error $\varepsilon$ by estimating both shifted costs to
error $\varepsilon$. Removing interface gates deletes blocks, so neither
subcircuit's $\chi$ nor $\beta$ increases, and cutting leaves the landscape
pointwise unchanged, so for blocks drawn from local $2$-designs the bound of
Eq.~\eqref{eq:prelim-gradient}~\cite{ZhangLiuZhang2024FLDC} applies
verbatim. Consequently, for a cost $C=\sum_\nu\lambda_\nu O_\nu$ with
$|\supp(O_\nu)|\le r$, $\|O_\nu\|\le1$ and $\Lambda=\sum_\nu|\lambda_\nu|$,
resolving $\partial_\mu C$ to a constant fraction of its root-mean-square
value is achieved with
\begin{equation}
  N=O\Bigl(\Lambda^2\,\bar\gamma^{2\kappa}\,4^{r\chi\beta}
  \Bigl(\sum_\nu\lambda_\nu^2\Bigr)^{-1}\log\frac1\delta\Bigr)
\end{equation}
shots per shifted cost, with the sum as in Eq.~\eqref{eq:prelim-gradient};
this is independent of $n$ whenever $\Lambda$ is.
\end{theorem}

\begin{proof}
Unbiasedness and $|\widehat o|\le\Gamma$ follow from sampling $\bm\alpha$
with probability $|a_{\bm\alpha}|/\Gamma$ and $|x|\le\|O\|$. The landscape statement is linearity of the
reconstruction in expectation. For the cost, each shot lies in
$[-\Lambda\Gamma,\Lambda\Gamma]$, and combining
Eq.~\eqref{eq:prelim-shots} at $\varepsilon=\Theta(v^{1/2})$ with the
variance bound $v\ge2\cdot4^{-r\chi\beta}\sum_\nu\lambda_\nu^2$ of
Eq.~\eqref{eq:prelim-gradient} gives the shot count.
\end{proof}

\section{Hardness at a bounded interface}
\label{app:hard}

This appendix shows that the three structural properties targeted in the
main text are mutually compatible. In particular, finite local depth and a
bounded interface do not by themselves force the branch tasks to be
classically easy. The result is an existence statement: it places a
worst-case hard FLDC inside one subcircuit and connects it to another
subcircuit through a constant-size interface. It does not identify a natural
average-case hard instance.

\begin{theorem}[Subcircuit hardness at bounded interface]
\label{thm:hardness}
Local-observable estimation to constant additive accuracy for
two-dimensional sequential FLDCs is \textsf{BQP}-complete in the worst
case~\cite[Sec.~V]{MalzTrivedi2025}. Consequently, for every constant $\kappa$
there exists a $\kappa$-interface distributed FLDC family whose
subcircuit-local estimation task remains \textsf{BQP}-hard after the
interface gates are removed.
\end{theorem}

\begin{proof}
Membership in \textsf{BQP} is immediate, since a quantum computer can run
the corresponding subcircuit.

For hardness, let $\{U_A^{(x)}\}$ be the worst-case \textsf{BQP}-hard
family of two-dimensional sequential FLDCs from
Ref.~\cite{MalzTrivedi2025}. Each $U_A^{(x)}$ has bounded block support and
bounded local depth. Place this family on subcircuit $A$. Let $U_B$ be any
fixed FLDC on subcircuit $B$, and add a fixed set $U_{\rm port}$ of
$\kappa=O(1)$ contact gates between designated boundary ancillas of $A$ and
$B$. Define
\begin{equation}
  U^{(x)}=(U_A^{(x)}\otimes U_B)\,U_{\rm port}.
\end{equation}
The resulting distributed circuit is still an FLDC: the contact gates
increase the local depth only on $O(\kappa)$ boundary ancillas. The
interface has constant size by construction.

Now remove the interface gates and consider the branch task on subcircuit
$A$. It contains the original family $\{U_A^{(x)}\}$. Therefore, a
classical polynomial-time algorithm for the branch task would give a
classical polynomial-time algorithm for the \textsf{BQP}-hard local
observable problem of Ref.~\cite{MalzTrivedi2025}. Hence the branch task is
\textsf{BQP}-hard in the worst case.
\end{proof}

\begin{corollary}[Coexistence]
\label{cor:coexistence}
There exist two-dimensional distributed FLDC families with
$\kappa,\chi,\beta=O(1)$ for which all of the following hold.

\begin{enumerate}[label=(\roman*),leftmargin=*]
\item The interface is cheaply cuttable:
\[
  \Gcut\le 3^\kappa=O(1),
\]
so $O(9^\kappa\varepsilon^{-2}\log\delta^{-1})$ shots suffice per
bounded-norm observable.

\item Under the local-$2$-design initialization assumptions of
Ref.~\cite{ZhangLiuZhang2024FLDC}, local costs with $n$-independent
coefficients have gradients with variance bounded below by
\[
  2\,4^{-r\chi\beta}
  \sum_\nu \lambda_\nu^2
  =\Omega(1)
\]
whenever the sum is bounded below.

\item The subcircuit-local observable-estimation task is
\textsf{BQP}-hard in the worst case.

\item Every state prepared by the architecture has seam entropy
$\Ehalf\le\kappa$, by Proposition~\ref{prop:port-capacity}.
\end{enumerate}

Thus bounded interface, finite local depth, and worst-case hard branch
tasks can coexist. Conversely, a family whose target states have seam
entropy growing with $n$ cannot satisfy the first item for exact knitting,
by Corollary~\ref{cor:area-law}.
\end{corollary}

The trainability statement in item (ii) and the hardness statement in
item (iii) refer to different kinds of evidence: random initialization for
the former, and worst-case instances for the latter.

\section{Tier-1 screen of explicit generators}
\label{app:screen}

This appendix illustrates the Tier-1 entries of Table~\ref{tab:families} on explicit circuit generators, computing the quantities of criteria T1, T2, S2, and C1 and checking that they follow the trends recorded there.
We screen eight generators with 7 to 517 qubits (Fig.~\ref{fig:tier1-screen} and Table~\ref{tab:tier1-screen}).
Seven of them instantiate rows of Table~\ref{tab:families}, and constant-depth brickwork serves as a shallow control.

The single bridge is grown as a $1\times k$ strip, the bounded-contact toric code as two such strips joined by one contact gate, the extensive-seam toric code as patches from $2\times1$ to $22\times11$ plaquettes cut along a straight seam, and the chain generators and the hardware-efficient ansatz on chains and $3\times k$ grids.
Routed QAOA applies $p=2$ layers of all $n(n-1)/2$ pairwise interactions through a linear SWAP network, as its recorded gate counts, exactly $n(n-1)$, show.
Cones are backward gate-list light cones of up to twelve probed cost observables on the uncut generator, with every recorded gate, SWAP gates included, treated as acting on all its qubits, and they bound the cones of the cutting branches, whose inserted operations are local.
Partitions come from a spectral bisection with Fiduccia--Mattheyses refinement, replaced by a multilevel variant at some sizes of the grid and seam generators, with 2 to 10 seeded restarts, and $g$ is evaluated with the smallest sum of squared cost coefficients over parameterized blocks.
All quantities follow directly from the generated gate lists, without state simulation, so they are exact properties of the generators, except for the balanced partition, which a seeded heuristic search provides.

\begin{figure}[t]
\centering
\includegraphics[width=\textwidth]{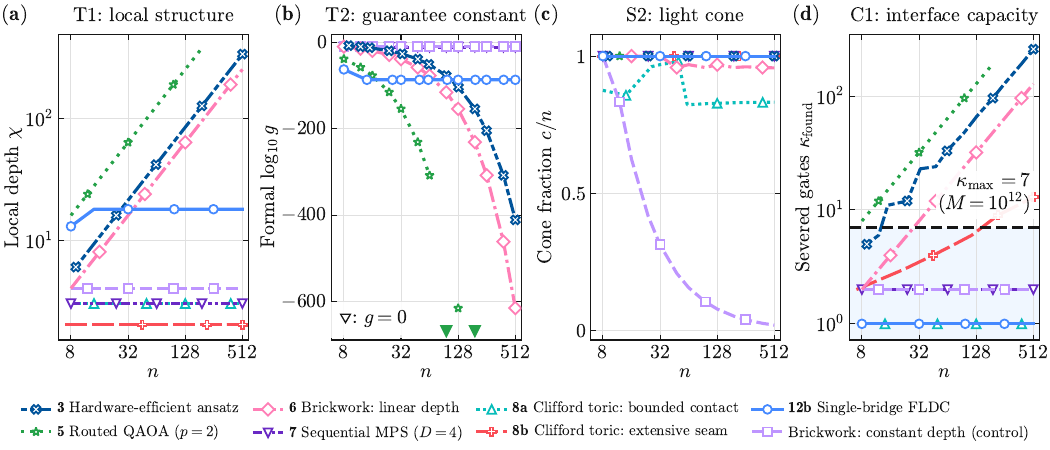}
\caption{Tier-1 screen of eight explicit generators against the number of
qubits $n$.
(a)~Local depth $\chi$ (T1).
(b)~Formal $\log_{10}g$ of Eq.~\eqref{eq:prelim-gradient} (T2), with no
curve for the Clifford generators, which have no variational blocks, and
filled triangles at an arbitrary height where $g=0$.
(c)~Fraction $c/n$ of qubits in the largest probed backward causal cone (S2).
(d)~Severed gates $\kappa_{\mathrm{found}}$ of the partition found by the
heuristic search (C1), an upper bound on $\kstar(n)$, with the
independent-cut limit $\kappa_{\max}^{\mathrm{ind}}=7$ for $10^{12}$ shots,
$\varepsilon=0.01$, and $\delta=0.05$ (dashed) and the affordable region
below it (shaded).
All values are structural counts or the result of a seeded heuristic search,
so there are no statistical error bars. All horizontal axes are logarithmic
in base 2, the vertical axes of (a) and (d) are logarithmic, the panel
titles name the criteria of Table~\ref{tab:properties}, and the legend
below the panels applies to all four, with the row of
Table~\ref{tab:families} in front of each entry and (control) marking the
constant-depth brickwork, which has no row.}
\label{fig:tier1-screen}
\end{figure}

\tableTierOneScreen

\textit{Local depth} (T1).
For the stated block decompositions, all generators have bounded block support, $2\leq\beta\leq4$, so T1 reduces to checking whether $\chi$ remains bounded as $n$ grows.
The hardware-efficient ansatz (row~3) and linear-depth brickwork (row~6a) fail this check because their increasing layer counts give $\chi=\Theta(n)$.
Routed QAOA (row~5) also has $\chi=\Theta(n)$, but here the growth comes from routing swaps, despite the fixed number of QAOA layers, $p=2$ (Fig.~\ref{fig:tier1-screen}(a)).
Both Clifford toric-code generators (rows~8a and~8b) have bounded local depth, although their lack of variational parameters makes them structural comparisons rather than trainability tests.
In contrast to the three growing-depth generators, the single-bridge FLDC (row~12b), sequential MPS preparation at fixed $D=4$ (row~7), and constant-depth brickwork (control) have $\chi=O(1)$ and pass T1.

\textit{Guarantee constant} (T2).
For the hardware-efficient ansatz (row~3) and linear-depth brickwork (row~6a), growing $\chi$ suppresses the formal bound through $4^{-r\chi\beta}$, giving $\log_{10}g\approx-411$ and $-616$, respectively, at the largest sizes (Fig.~\ref{fig:tier1-screen}(b)).
Routed QAOA (row~5) instead gives $g=0$ for parameters whose blocks overlap no cost term, making the coefficient sum in Eq.~(\ref{eq:prelim-gradient}) vanish.
Neither a collapsing nor a zero lower bound establishes a barren plateau.

The single-bridge FLDC (row~12b) has bounded local depth but does not satisfy the local-2-design assumptions: the encoder is fixed, and the bridge and anchor are not drawn from local 2-designs.
Its formal value $g\sim10^{-86}$ is therefore only formal, and trainability is assessed directly through T3 (Sec.~\ref{subsec:app-tier2} and App.~\ref{app:variance}).
The Clifford generators (rows~8a and~8b) have no variational parameters, so the \textsc{Train} criteria do not apply.

In contrast, sequential MPS preparation at fixed $D=4$ (row~7) and constant-depth brickwork (control) pass the scaling check: their formal bounds remain nonzero and independent of $n$, with $\log_{10}g\approx-10$ and $-9$, respectively.
Under the local-2-design assumptions, these would provide size-independent variance lower bounds, although their small magnitudes do not ensure gradient resolution within a practical shot budget.

\textit{Light cone} (S2).
Constant-depth brickwork (control) is the only family which fails the cone-growth check: its largest probed cone saturates at 10 qubits for $n\geq12$, permitting efficient classical simulation by Proposition~\ref{prop:lightcone} (Fig.~\ref{fig:tier1-screen}(c)).
All other generators have extensive probed cones, but this does not establish hardness, as the classically simulable Clifford codes (rows~8a and~8b) and fixed-bond MPS preparation (row~7) illustrate.

\textit{Interface} (C1).
We compare the heuristic upper bound $\kappa_{\mathrm{found}}\geq\kstar(n)$ with the independent-cut threshold $\kappa_{\max}^{\mathrm{ind}}=7$ of Eq.~\eqref{eq:kmax}, for $M=10^{12}$ shots, $\varepsilon=0.01$, and $\delta=0.05$ (Fig.~\ref{fig:tier1-screen}(d)).
The growing-local-depth generators (rows~3, 5, and~6a) exceed this threshold by $n=32$, and the extensive-seam Clifford code (row~8b) near $n=200$.
In contrast, the single-bridge FLDC (row~12b) and bounded-contact Clifford code (row~8a) pass with $\kappa_{\mathrm{found}}=1$, while fixed-bond MPS preparation (row~7) and constant-depth brickwork (control) pass with $\kappa_{\mathrm{found}}=2$ at every tested size.
Exceeding this sufficient threshold does not establish unaffordability: angle-dependent norms, a joint decomposition where applicable, or a better partition may still satisfy C1.

The Tier-1 screen illustrates the independence of local depth, interface size, and cone growth.
Bounded local depth does not imply a bounded interface or cone, while extensive cones can coexist with classically easy branch tasks.

\section{Model, circuit, and exact endpoints}
\label{app:endpoints}

The test of Sec.~\ref{sec:application} rests on exactly known states, and this appendix derives them.
We first show that the circuit, built from the encoder and the pair blocks of the sweeps, has bounded local depth and is therefore an FLDC.
At $h=0$, we find the exact ground state and the angles that prepare it in closed form, which also provides the contact benchmark of Sec.~\ref{subsec:app-knit}.
Since we reuse the zero-field angles at nonzero field to avoid training, we also give their energy there in closed form and compare it with the ground energy.
At $h=1$, the ground state is again known in closed form, which gives an exact reference at the other end of the field sweep.

\subsection{Encoder and sweeps}
\label{app:ep-sweeps}

A star and a plaquette share zero or two edges, so all stabilizers commute.
A patch has $(n_x+1)(n_y+1)$ stars and $n_xn_y$ plaquettes, $n+1$ generators with one relation, since the product of all stars is the identity.
The stabilizer group therefore has rank $n$, and an isolated patch at zero field has a unique ground state $|\Omega_p\rangle$, the common $+1$ eigenstate of all stars and plaquettes.
Since $|0^n\rangle$ is a $+1$ eigenstate of every plaquette, applying the independent star projectors prepares it,
\begin{equation}
  E_p|0^n\rangle=|\Omega_p\rangle=\prod_{v\neq v_0}\frac{\openone+S_{v,p}}{\sqrt2}\,|0^n\rangle ,
  \label{eq:ep-encoder}
\end{equation}
where $v_0$ is the top-right vertex.
The factor $2^{-1/2}$ per star normalizes the state, because distinct products of stars map $|0^n\rangle$ to orthogonal basis states.
A circuit for Eq.~\eqref{eq:ep-encoder} treats the vertices row by row from the bottom-left corner.
For each $v\neq v_0$ it applies a Hadamard gate to a pivot spin on the edge from $v$ to its right or upper neighbor, followed by controlled-NOT gates from the pivot to the other spins of $S_{v,p}$.
The pivot lies on no earlier star, so it is still in $|0\rangle$, and the step applies $(\openone+S_{v,p})/\sqrt2$.
Every spin enters at most two such steps, so the encoder has bounded local depth.
The encoder used in the simulations follows Ref.~\cite{Higgott2021Encoding} and prepares the same state.

For two consecutive edge spins $(u,v)$, the pair block of Fig.~\ref{fig:experiment}(b) is
\begin{equation}
  U_{uv}(\boldsymbol\vartheta)
  =R_{ZZ}^{uv}(\vartheta_5)\,R_Z^{v}(\vartheta_4)\,R_Z^{u}(\vartheta_3)\,R_Y^{v}(\vartheta_2)\,R_Y^{u}(\vartheta_1),
  \label{eq:ep-pair}
\end{equation}
with $R_P(\vartheta)=\exp(-i\vartheta P/2)$ for a Pauli string $P$.
The local depth of Sec.~\ref{subsec:app-model} counts the blocks acting on each spin in the generated circuits.
It is dominated by the sixteen pair blocks that four sweeps apply to an interior spin, and it does not grow with the length of the patch, which we checked on strips of up to 512 spins.

\subsection{Zero field}
\label{app:ep-zero}

Let $|\Omega_0\rangle$ be the product of the two patch ground states, so that $\mathcal H_{\mathrm s}|\Omega_0\rangle=(N+2)|\Omega_0\rangle$.
The contact $Q$ anticommutes with the two stars incident to $a$ and the two incident to $b$, and commutes with every other stabilizer.
Hence $|\Omega_1\rangle=Q|\Omega_0\rangle$ is orthogonal to $|\Omega_0\rangle$, with $\mathcal H_{\mathrm s}|\Omega_1\rangle=(N-6)|\Omega_1\rangle$.
Their span is invariant under $H(0,t)$, which acts on it as
\begin{equation}
  H(0,t)=(-N+2)\,\openone+
  \begin{pmatrix}
    -4 & -t\\
    -t & 4
  \end{pmatrix}.
  \label{eq:ep-block}
\end{equation}
With $D_t=\sqrt{16+t^2}$ and $\tan\phi=t/4$, the lower eigenpair is
\begin{equation}
  E_0(0,t)=-N+2-D_t,
  \qquad
  |\Omega(0)\rangle=\cos\frac{\phi}{2}\,|\Omega_0\rangle+\sin\frac{\phi}{2}\,|\Omega_1\rangle .
  \label{eq:ep-zero}
\end{equation}
Relative to the bridge-off state $|\Omega_0\rangle$, whose energy is $-N-2$, the coupling lowers the energy by $D_t-4\approx0.031$ at $t=0.5$, independently of the size.
It is the global ground state.
The four stars $S_1,\ldots,S_4$ that anticommute with $Q$ form, with the contact, the part $K=-\sum_iS_i-tQ$, and the anticommutation gives $K^2=(\sum_iS_i)^2+t^2\openone\le D_t^2\openone$, so $K\ge-D_t$.
Every other stabilizer contributes at least $-1$, and Eq.~\eqref{eq:ep-zero} attains all these bounds, which proves optimality and uniqueness.

The circuit prepares this state exactly.
The bridge gives $R_{ZZ}^{ab}(\phi)|\Omega_0\rangle=\cos\frac{\phi}{2}|\Omega_0\rangle-i\sin\frac{\phi}{2}|\Omega_1\rangle$, with the right amplitudes but the wrong relative phase.
A star incident to $a$ has eigenvalue $+1$ on $|\Omega_0\rangle$ and $-1$ on $|\Omega_1\rangle$, so the anchor rotation at $\alpha=\pi/2$ corrects the phase and yields $e^{-i\pi/4}|\Omega(0)\rangle$.
Every single-spin Pauli operator has zero expectation value in $|\Omega(0)\rangle$, because it anticommutes with a stabilizer of each component and cannot map one component onto the other.
In particular $z_a=z_b=0$, and since $Q$ exchanges the two components, the contact correlation of Eq.~\eqref{eq:contact-corr} is $C_{ab}=\langle Q\rangle=\sin\phi=t/D_t$ at every size.

\subsection{Zero-field angles at nonzero field}
\label{app:zf-angles}

The zero-field angles prepare the same state at every field.
Since $\langle\mathcal H_{\mathrm f}\rangle=0$ and $\langle\mathcal H_{\mathrm s}\rangle=N-2+16/D_t$ in this state, their energy is exactly linear in $h$,
\begin{equation}
  E_{\mathrm{zf}}(h,t)=-N+2-D_t+h\Bigl(N-2+\frac{16}{D_t}\Bigr),
  \label{eq:ep-zf}
\end{equation}
which is the value the device should return for these angles.
By the Hellmann--Feynman identity $\partial_hE_0=\langle\partial_hH\rangle$ with $\partial_hH=\mathcal H_{\mathrm s}-\mathcal H_{\mathrm f}$, this line is also the first-order expansion of the exact ground energy around $h=0$.
It is a variational upper bound for all $h$ and differs from the ground energy at second order in $h$.

Relative to the references of Sec.~\ref{subsec:app-sim}, its excess at $N=48$ is $9.1\times10^{-3}$, $0.051$, $0.091$, and $0.14$ per spin at $h=0.1$, $0.2$, $0.25$, and $0.3$, respectively.
The excess decreases slowly with $N$, a boundary effect; for example, at $N=80$ it is $0.008$ per spin at $h=0.1$ and $0.081$ per spin at $h=0.25$.

With no dressing and $\alpha=\pi/2$, the knitting decomposition inserts only Pauli operators, $Z$ measurements, and $R_Z(\pm\pi/2)$ rotations, so every branch circuit is a Clifford circuit at every field.

\subsection{Polarized limit}
\label{app:ep-polarized}

At $h=1$ the stabilizers drop out, and $H(1,t)=-\sum_{j\notin\{a,b\}}(X_j+Z_j)+K_{ab}$ with $K_{ab}=-(X_a+Z_a+X_b+Z_b)-tQ$.
Every spin other than $a$ and $b$ is in the single-spin ground state $|q_+\rangle=\cos(\pi/8)|0\rangle+\sin(\pi/8)|1\rangle$, with $(X+Z)|q_+\rangle=\sqrt2\,|q_+\rangle$.
The antisymmetric state of the contact pair has energy $t$, and in the symmetric basis $\{|00\rangle,(|01\rangle+|10\rangle)/\sqrt2,|11\rangle\}$
\begin{equation}
  K_{ab}^{\mathrm{sym}}=
  \begin{pmatrix}
    -2-t & -\sqrt2 & 0\\
    -\sqrt2 & t & -\sqrt2\\
    0 & -\sqrt2 & 2-t
  \end{pmatrix},
  \label{eq:ep-sym}
\end{equation}
whose eigenvalues solve $\lambda^3+t\lambda^2-(t^2+8)\lambda-t^3=0$.
With $\lambda_0(t)$ the smallest root and $|\varphi_0\rangle$ its pair state,
\begin{equation}
  E_0(1,t)=-(N-2)\sqrt2+\lambda_0(t),
  \qquad
  |\Omega(1)\rangle=|\varphi_0\rangle_{ab}\otimes|q_+\rangle^{\otimes(N-2)},
  \label{eq:ep-one}
\end{equation}
with $\lambda_0=-3.1262$ at $t=0.5$.
The ground state is a product state except for the contact pair, which is entangled for every $t>0$.
\section{Transferring angles between patch sizes}
\label{app:angles}

This appendix develops a method to transfer trained angles from small patches to larger ones, which avoids optimization at the target size.
Such a method is needed because exact classical methods have limited reach, with our exact evaluator being able to train angles only up to $N=48$.
Transferring variational parameters between problem sizes has precedents in the quantum approximate optimization algorithm, where angles optimized on small graphs can remain effective on larger graphs with similar local structure~\cite{Galda2021Transfer, WurtzLykov2021FixedAngles, Shaydulin2023ParameterTransfer,Sureshbabu2024ParameterSetting}.
That argument rests on a causal neighborhood of fixed depth, which a sequential FLDC does not have, so we motivate transfer instead by the locality of the ground-state response to the field (Figure~\ref{fig:transfer}(a)).
Figure~\ref{fig:transfer}(b) shows the route we take: we first continue in $h$ at a source size, then transfer the angles in $N$ at fixed field.
Appendix~\ref{app:qac} gives the locality argument, and App.~\ref{app:transfer} tests a rule that matches local environments, which fails for generic patches, but works strips. We leave expansion to compact patches as future work.

\figTransfer

\subsection{Why angles might transfer between sizes}
\label{subsec:app-angles}\label{app:qac}

The transfer rests on how the ground state responds to a change of the field.
Let $|\Omega(h)\rangle$ be the nondegenerate ground state of $H(h)$ with energy $E_0(h)$, and $|m(h)\rangle$ an excited eigenstate with energy $E_m(h)$.
Differentiating the eigenvalue equation $H|\Omega\rangle=E_0|\Omega\rangle$ with respect to $h$ and projecting onto $\langle m|$ gives the first-order response
\begin{equation}
  \langle m|\partial_h\Omega\rangle=-\frac{\langle m|\partial_hH|\Omega\rangle}{E_m-E_0}.
  \label{eq:transfer-response}
\end{equation}
Each excited state contributes with a weight inversely proportional to its excitation energy, so the state is most sensitive to the field where the gap $\gamma_N(h)$ between the two lowest levels is small (Fig.~\ref{fig:transfer}(a)).
A statement that holds uniformly in the size therefore needs bounded interactions and a gap that stays open along the whole path $I=[0,h]$ for every patch size,
\begin{equation}
  \gamma_{\min}=\inf_N\inf_{h'\in I}\gamma_N(h')>0 ,
  \label{eq:transfer-uniform-gap}
\end{equation}
which gaps computed for finitely many patches cannot establish.

Under this condition, quasi-adiabatic continuation~\cite{HastingsWen2005,Bachmann2012Automorphic} replaces the sum over excited states in Eq.~\eqref{eq:transfer-response} by a single operator.
It constructs a Hermitian operator $\mathcal D(h)$, the quasi-adiabatic generator, whose commutator with the ground-state projector $\Pi_0(h)=|\Omega(h)\rangle\langle\Omega(h)|$ gives the change of the projector with the field, together with the unitary $\mathcal U(h)$ that this generator produces,
\begin{equation}
  \partial_h\Pi_0=i\bigl[\mathcal D,\Pi_0\bigr],
  \qquad
  \partial_h\mathcal U=i\mathcal D\,\mathcal U,
  \qquad
  \mathcal U(0)=\openone .
  \label{eq:qa-flow}
\end{equation}
Integrating this flow from $h=0$ gives $\Pi_0(h)=\mathcal U(h)\Pi_0(0)\mathcal U(h)^\dagger$, so $\mathcal U(h)$ carries the exactly known zero-field ground state to the ground state at field $h$.
This is the dressing that the variational sweeps have to approximate.
For $H=\sum_jH_j$ and $0<\gamma\le\gamma_{\min}$, the generator is a sum $\mathcal D=\sum_j\mathcal D_j$ of pieces, one for each term of the Hamiltonian,
\begin{equation}
  \mathcal D_j(h)=\int_{-\infty}^{\infty}\!d\tau\,W_\gamma(\tau)\,e^{iH(h)\tau}\,\partial_hH_j(h)\,e^{-iH(h)\tau}.
  \label{eq:transfer-generator}
\end{equation}
Each piece takes the derivative $\partial_hH_j$ of one Hamiltonian term, evolves it in time under the full Hamiltonian, and averages the result over evolution times $\tau$ with a filter function $W_\gamma$.
The filter is real and odd, decays rapidly for $|\tau|\gtrsim1/\gamma$, and is chosen so that its Fourier transform equals $i/\omega$ at every frequency at least as large as the gap,
\begin{equation}
  \widehat W_\gamma(\omega)=\int_{-\infty}^{\infty}\!d\tau\,W_\gamma(\tau)\,e^{i\omega\tau}=\frac{i}{\omega}
  \qquad(|\omega|\ge\gamma).
  \label{eq:transfer-filter}
\end{equation}
Between the ground state and an excited state, the time evolution in Eq.~\eqref{eq:transfer-generator} multiplies the matrix element by $e^{i\omega\tau}$ with $\omega=E_m-E_0\ge\gamma$, and the integral over $\tau$ then yields the factor of Eq.~\eqref{eq:transfer-filter},
\begin{equation}
  \langle m|i\mathcal D|\Omega\rangle=-\frac{\langle m|\partial_hH|\Omega\rangle}{E_m-E_0},
  \label{eq:transfer-generator-check}
\end{equation}
which is Eq.~\eqref{eq:transfer-response}.
The generator therefore encodes the full response of the ground state, but, unlike the sum over excited states, it is built from the local terms of the Hamiltonian.

\figQuasiLocal

The locality of the pieces follows from the Lieb--Robinson bound~\cite{LiebRobinson1972}.
In a lattice system with bounded local interactions, a local operator evolved in time spreads at most with a finite velocity $v$, the Lieb--Robinson velocity.
After a time $\tau$, its weight outside a distance $v|\tau|$ from its initial support is exponentially small, although not zero, which defines an effective light cone (Fig.~\ref{fig:qac-range}(a)).
Because the filter suppresses evolution times beyond about $1/\gamma$ (Fig.~\ref{fig:qac-range}(b)), each piece $\mathcal D_j$ only involves the evolution of $\partial_hH_j$ up to those times, and it is therefore concentrated within a distance of order $\xi\sim v/\gamma$ of its term.
To make this quantitative, let $A_j(\tau)=e^{iH\tau}(\partial_hH_j)e^{-iH\tau}$, and let $A_j^{(r)}(\tau)$ be the same evolution under the Hamiltonian restricted to a distance $r$ beyond the support of $H_j$. Define $\mathcal D_j^{(r)}$ to be the piece built from $A_j^{(r)}$ with the same filter.
Splitting the integral in Eq.~\eqref{eq:transfer-generator} at a cutoff time $T>0$ gives
\begin{equation}
\begin{aligned}
  \bigl\|\mathcal D_j-\mathcal D_j^{(r)}\bigr\|
  \le{}&\int_{|\tau|\le T}\!d\tau\,|W_\gamma(\tau)|\,\bigl\|A_j(\tau)-A_j^{(r)}(\tau)\bigr\|\\
  &+2\,\|\partial_hH_j\|\int_{|\tau|>T}\!d\tau\,|W_\gamma(\tau)| .
\end{aligned}
  \label{eq:transfer-tail}
\end{equation}
The first term compares the full and the truncated evolution at short times, and the Lieb--Robinson bound makes their difference exponentially small as long as the light cone has not reached the truncation radius, $T<r/v$.
The second term collects the long times, which the decay of the filter suppresses.
Choosing $T$ just below $r/v$ therefore makes $\mathcal D_j^{(r)}$ a good approximation of $\mathcal D_j$, with an error that decays rapidly with $r$, but does not vanish~\cite{Bachmann2012Automorphic}.

This has a direct consequence for transfer.
Two terms, possibly in patches of different size, whose neighborhoods of radius $r$ carry identical labeled Hamiltonians along the whole path have identical truncated pieces $\mathcal D^{(r)}$ (Fig.~\ref{fig:qac-range}(c)).
The quasi-adiabatic dressing near such terms is therefore the same up to these tails, whatever the size of the patch, which suggests that gates whose surroundings agree should receive similar angles.
This motivates the transfer rules below, but is susceptible to errors  when an optimizer does not select circuit coordinates that follow the local pieces.

\subsection{Transfer rules}
\label{app:transfer}\label{app:transfer-joint}

We consider two rules that build the angles of a target patch from the best-found angles of smaller trained patches at the same field.
Both give every dressing angle a class and assign the same angle to gates of the same class, but they define the classes differently and obtain the source angles differently.

\textit{Classes.}
The role $\tau(\mu)$ of a dressing angle $\mu$ records the prefix or suffix stage, the sweep, the slot of its pair within the face, and the rotation within the pair block, but not the position of the face.
Its window of radius $r\in\{0.5,1,1.5,2.5\}$ is the square of half-width $r$ lattice units centered on the gate, see Fig.~\ref{fig:transfer}(c).
A class combines the role with the labeled Hamiltonian terms in the window, where a star truncated by the boundary differs from a full one.

\textit{Rule 1 (environment matching).}
Rule 1 takes its source angles $\widehat\theta^{\mathrm{src}}$ from patches optimized independently, and it identifies two classes when they agree up to one of the eight symmetries of the square lattice.
This identification ignores that a lattice symmetry need not preserve the order of the sweeps.
Let $G_r(\mu)$ be the set of source angles in the class of a target angle $\mu$, with resultant
\begin{equation}
  m_\mu=\frac{1}{|G_r(\mu)|}\sum_{\mu'\in G_r(\mu)}e^{i\widehat\theta^{\mathrm{src}}_{\mu'}},
  \qquad
  \rho_\mu=|m_\mu| .
  \label{eq:transfer-resultant}
\end{equation}
The transferred angle is the circular mean
\begin{equation}
  \theta^{\mathrm{tr}}_\mu=\operatorname{Arg}m_\mu ,
  \label{eq:app-transfer}
\end{equation}
which respects the $2\pi$ periodicity of a rotation angle.
Angles $0.02$ and $2\pi-0.02$ have circular mean zero, whereas their arithmetic mean is $\pi$.
A resultant length $\rho_\mu$ near one indicates tightly clustered source angles.
A target angle without a match receives the circular mean over all source angles of the same role.
The bridge and anchor angles carry no spatial environment, so each forms a class of its own and receives the circular mean over all sources.
Nothing overrides the rule at $h=0$, where we know analytic formulas.

\textit{Coverage.}
A rule can only transfer classes that occur in the sources.
For Rule 1, the coverage
\begin{equation}
  \mathcal C_r=\frac{\#\{\mu\text{ in the target dressing}:G_r(\mu)\neq\varnothing\}}{N_{\mathrm p}^{\mathrm{tar}}-2}
  \label{eq:transfer-coverage}
\end{equation}
is the fraction of the $N_{\mathrm p}^{\mathrm{tar}}-2$ dressing angles of the target, which exclude the bridge and anchor angles, whose class occurs in at least one source.
The remaining fraction $1-\mathcal C_r$ falls back to a role mean.
Coverage depends on which neighborhoods the sources contain, and it decreases with the window radius, because a larger window distinguishes more environments.
Figure~\ref{fig:transfer-example} illustrates this.
A short strip covers every longer strip, since both contain only end faces and middle faces (panel a).
A strip offers nothing to a compact patch, whose faces have neighbors in two directions (panel b).
A face that lies in the interior of a small patch can still see the boundary within its window, unlike the same face in a larger patch (panel c).
For the $3\times3$ target, all smaller patches together cover every angle at $r=0.5$, thanks to the symmetry identification, but almost none at $r\ge1.5$.

Full coverage, however, does not make Rule 1 accurate (App.~\ref{app:transfer-performance}), so coverage is necessary but not sufficient.
In particular, Rule 1 averages angles from patches that were optimized separately.
Many different angle vectors prepare the same state, so each optimization settles on its own convention, and averaging across patches mixes conventions that do not fit together.
Rule 2 avoids this by fixing one convention for all patches at once.

\figTransferExample

\textit{Rule 2 (joint class table).}
Rule 2 builds a lookup table with one angle per class and trains this single table on several small patches together.
The classes are those of Fig.~\ref{fig:transfer}(c) at the smallest window, $r=0.5$, except that we keep track of orientation, that is, of the direction in which the pair block points and of which neighbors lie on which side.
We do not merge a class with its rotated or mirrored copy, because in a sequential circuit these are different gates.
We train the table by minimizing the average energy per spin of all source patches, starting from four random tables and keeping the best, since the zero-field angles form a stationary point from which the optimizer does not move.
Training on all sources together forces each entry to work in every place where its class appears, rather than in one patch only.
The trained table then fills the target directly, without any further optimization.

\subsection{Performance of the transfer rules}
\label{app:transfer-performance}

We test both rules on targets that the exact evaluator can still handle, so that every transferred state can be compared with the best state found by continuation.
The results follow the picture of App.~\ref{app:transfer}.
Rule 1 does not carry the angles over, whereas Rule 2 transfers them accurately to strips, but not yet to compact patches.

\figTransferResults

\textit{Rule 1.}
Figure~\ref{fig:transfer-results} shows Rule 1 on targets of 14 to 48 spins.
Before any update at the target, the transferred angles lie about $0.8$ per spin above the ground energy, for every target, nonzero field, and window radius.
Copying each angle from the gate at the same relative position in a source, which we call the position copy, behaves the same way.
On an energy scale of about one per spin, this is close to a random state and further from the ground state than the zero-field angles reused unchanged, which stay within about $0.1$ per spin up to $h=0.25$ at every size up to $N=80$ (App.~\ref{app:zf-angles}).
Refining the transferred vector at the target recovers most of this gap, but it ends at or above continuation there and is no cheaper from $N=26$ on.
The diagnostics below show that the obstacle is not missing coverage but the mismatch of conventions anticipated in App.~\ref{app:transfer}.

\figTransferDiagnostics

\textit{Why Rule 1 does not transfer.}
The step from Eq.~\eqref{eq:qa-flow} to circuit angles is not a theorem, and the diagnostics show where it breaks (Fig.~\ref{fig:transfer-diagnostics}(a)).
On the $3\times3$ target at $h=0.1$ to $0.3$ the fully covered rules leave $0.53$ to $0.71$ per spin and the almost uncovered ones, which fall back to role means, $0.29$ to $0.38$.
The resultant length explains why.
Among directly optimized angles of one class the median resultant length is $0.81$, a circular spread of $0.65$~rad, and for one class in ten it lies below $0.51$, a spread above $1.16$~rad, see Fig.~\ref{fig:transfer-diagnostics}(b,c).
A class mean therefore shifts every gate by several tenths of a radian, enough to change a state built from four sweeps of such gates substantially.
Tying all gates of a class to one shared angle leaves a reduced circuit of $122$ free angles instead of $642$, and training it directly at $N=24$ leaves $0.10$ per spin at $h=0.25$ and $0.13$ at $h=0.3$, against $10^{-4}$ per spin for the free optimization at the same size and budget.
The tested optimizations did not find an accurate vector within these symmetrized classes at the available budget, which does not exclude more accurate vectors or better optimizers in the same parametrization, and the oriented classes of Rule 2 do reach good accuracy.
The other expectations behind Rule 1 are not borne out either.
Complete coverage does not improve the transferred energy ($0.87$ against $0.78$ per spin at partial coverage), the transfer is least accurate in the polarized limit ($1.43$ per spin) rather than in the transition window ($0.72$), and the trained angles are not smooth within classes.

\figTransferBenchmark

\textit{Rule 2.}
Rule 2 needs more than coverage.
Every table entry that the target uses must have been trained on at least two gates in the sources.
An entry trained on a single gate simply inherits the convention that one optimization happened to pick, which is the problem of Rule 1 again.
An entry shared by several gates, in patches of different sizes, has to work for all of them, and this pushes it toward the local dressing suggested by Eq.~\eqref{eq:qa-flow}.
Like coverage, the condition can be checked by counting gates, before any training.
Strips pass it easily, because middle faces appear in every strip source, so strip sources of up to four plaquettes suffice for every strip of five or more.

On strips, Rule 2 transfers accurately (Fig.~\ref{fig:transfer-benchmark}).
On the $1\times5$ strip at $h=0.25$, with shorter strips as sources, the transferred angles lie $1.2\times10^{-3}$ per spin above the reference without any update, at least an order of magnitude closer than direct optimization at the target, and only continuation in $h$ does better, at $2.6\times10^{-4}$.
The transfer needs a single energy evaluation at the target, after an eight-minute training of the table that serves every longer strip, whereas continuation takes about three hours per field sweep.
At $N=38$, beyond our continuation runs, the table still leaves $1.3\times10^{-3}$ per spin, against $0.115$ for direct optimization.
On every tested target, the transfer is accurate exactly where the condition above holds (Fig.~\ref{fig:transfer-benchmark}(b)), which we use as an empirical criterion rather than a guarantee.

Compact patches do not yet meet the condition, because each of our sources, of at most 34 spins, contains their interior faces only once, so for them transfer remains to be demonstrated.
Adding the $2\times4$ patch as a source would meet it for the $3\times3$ target, although it covers only 89\% of its angles, and we have not tested it.
The compact patches of the device runs, from $4\times4$ on, would need a source with two interior faces, the $3\times4$ patch with 62 spins, which is beyond our exact evaluator.

\section{Knitted estimators and seam entropy}
\label{app:cutting}

This appendix gives the knitted estimators of the energy and the contact correlation, the bridge-angle dependence of the energy, and the seam entropy of the prepared states.
The reconstruction uses the local-operation protocol and estimator framework of Refs.~\cite{MitaraiFujii2021VirtualGate,CarreraVazquez2024RealTime}.

\subsection{Energy and contact correlation}
\label{app:cut-energy}

A setting $(i,p,g)$ runs branch $i$ of Fig.~\ref{fig:experiment}(c), with weight $w_i\in\{c^2,s^2,d,-d,d,-d\}$, on patch $p$ with readout basis $g\in\{X,Z\}$.
For a Hamiltonian Pauli term $P_\nu$ of group $g$, let $\bar y^{(i)}_{\nu,p}$ be the mean over shots of the product of its measured spin values on patch $p$, multiplied by the recorded sign if branch $i$ measures on patch $p$.
Empty products equal one, so a term on one patch picks up the mean sign of the branch on the other.
The energy estimate is
\begin{equation}
  \widehat E
  =
  \sum_{i=1}^{6}
  w_i
  \sum_\nu
  J_\nu\,
  \bar y^{(i)}_{\nu,A}\,
  \bar y^{(i)}_{\nu,B},
  \label{eq:cut-estimator}
\end{equation}
which is unbiased for the ideal circuit because the two patches are sampled independently and every branch state is a product across the seam.
Each of the 24 settings first receives 32 shots, and the rest of $M=24{,}576$ is allocated in proportion to $|w_i|$ times the estimated standard deviation of the setting.
Error bars come from bootstrap resamples of whole bit strings within each setting.

\textit{Contact correlation.}
We reconstruct $\langle Z_aZ_b\rangle$, $\langle Z_a\rangle$, and $\langle Z_b\rangle$ from the all-$Z$ branch records using Eq.~\eqref{eq:cut-estimator}, with the Hamiltonian replaced by the corresponding observable.
We then form $C_{ab}=\langle Z_aZ_b\rangle-\langle Z_a\rangle\langle Z_b\rangle$, rather than combining branchwise connected correlations.
For each bootstrap replicate, the same resampled records are used for all three expectations before forming $C_{ab}$, preserving their statistical dependence.
The three expectations are unbiased for the ideal circuit, but the product $\langle Z_a\rangle\langle Z_b\rangle$ makes the finite-shot estimate of $C_{ab}$ nonlinear, so we do not claim it is unbiased.

Applied to this model, the sampling bound of Eq.~\eqref{eq:prelim-shots} gives a size-independent guarantee.
For an estimator that samples one branch with probability $|w_i|/\Gamma$ and one readout group in proportion to its Hamiltonian weight, each observation of $E/N$ is bounded by $\Gamma[1+h+(2(1-h)+t)/N]$, which does not grow with $N$.
Explicitly, the total weight of the Hamiltonian is
\begin{equation}
  W=\sum_\nu|J_\nu|=(1-h)(N+2)+2hN+t ,
  \label{eq:cut-weight}
\end{equation}
each observation of $E/N$ lies in $[-A_N,A_N]$ with $A_N=\Gamma W/N$, and Eq.~\eqref{eq:prelim-shots} gives
\begin{equation}
  \Pr\bigl[|\overline e-E/N|\ge\varepsilon\bigr]
  \le
  2\exp\Bigl[-\frac{M\varepsilon^2}{2A_N^2}\Bigr],
  \label{eq:cut-bound}
\end{equation}
with $A_N\approx1.3$ at $h=0$, several times the measured error constants of Fig.~\ref{fig:tier2}(a).

\subsection{Bridge angle}
\label{app:cut-angle}
\figLandscape

This subsection shows that the same measurement outcomes can be used to reconstruct the energy at any bridge angle, and asks whether the reconstruction is precise enough to locate the best bridge angle and thus resolve what the bridge contributes as the patches grow.

In the uncut circuit, the bridge angle $\phi$ must be fixed before the bridge gate is applied.
After cutting, the six branch circuits contain only the patch angles and fixed local operations, and $\phi$ enters only through the weights $w(\phi)=(c^2,s^2,d,-d,d,-d)$ of Eq.~\eqref{eq:cut-estimator}.
Running each branch circuit on both patches gives its unweighted contribution $E^{(i)}=\sum_\nu J_\nu\,\bar y^{(i)}_{\nu,A}\bar y^{(i)}_{\nu,B}$, including any recorded measurement signs.
Reweighting these six numbers gives the energy at any bridge angle, for fixed patch angles.
Since $c^2=(1+\cos\phi)/2$, $s^2=(1-\cos\phi)/2$, and $d=\frac12\sin\phi$,
\begin{equation}
  E(\phi)=\sum_{i=1}^{6}w_i(\phi)\,E^{(i)}=e_0+e_1\cos\phi+e_2\sin\phi ,
  \label{eq:cut-landscape}
\end{equation}
with $e_0=\frac12(E^{(1)}+E^{(2)})$, $e_1=\frac12(E^{(1)}-E^{(2)})$, and $e_2=\frac12(E^{(3)}-E^{(4)}+E^{(5)}-E^{(6)})$.
For nonzero amplitude, $e_1^2+e_2^2>0$, its minimum lies at
\begin{equation}
  \phi^\star=\operatorname{atan2}(-e_2,-e_1),
  \qquad
  E(\phi^\star)=e_0-\sqrt{e_1^2+e_2^2}.
  \label{eq:cut-optimum}
\end{equation}
If $e_1=e_2=0$, every angle minimizes $E$, and from sampled records we write $\widehat e_1$, $\widehat e_2$, and $\widehat\phi^\star$ for the estimates.

Once records covering all six branches are available, thus determining $e_1$ and $e_2$, we can therefore derive the best bridge angle for the chosen patch angles, without new angle-dependent circuits, to a precision set by the number of shots.
We test this precision by sampling the ideal branch circuits with the shot budget of the device runs.
At $N=34$, the reconstructed curves (solid) follow the exact (dashed) $E(\phi)$ of the same patch angles at every field $h$, but their minima are not a reliable guide to the best angle (Fig.~\ref{fig:landscape}(a)).
Minimizing a noisy curve favors angles at which the noise lowers the estimate, so the reconstructed minima lie up to two standard errors, about $0.01$ per spin, below the exact energy at the selected angle.
That exact energy stays within $4\times10^{-3}$ per spin of the best value along $\phi$, so the apparent gain from re-optimizing the bridge is statistical rather than real.

As the patches grow, the precision of the reconstruction degrades even though the effect of the bridge does not.
The zero-field circuit makes this visible, because its landscape is known exactly at every size.
With no dressing and $\alpha=\pi/2$, App.~\ref{app:endpoints} gives
\begin{equation}
  E(\phi)=-N+2-D_t\cos(\phi-\phi_0),
  \qquad
  D_t=\sqrt{16+t^2},\quad \phi_0=\arctan(t/4),
  \label{eq:cut-landscape-zero-field}
\end{equation}
so the offset $e_0=-N+2$ grows with the size, whereas the amplitude $D_t$ and the optimum $\phi_0$ do not.
The bridge itself lowers the energy by only $D_t-4\approx0.031$ relative to $\phi=0$, a small fraction of the oscillation amplitude $D_t\approx4.03$.
After subtracting $e_0$, the exact curves of all sizes coincide and share the optimum $\phi_0$ (Fig.~\ref{fig:landscape}(b,c)).
The reconstructed curves in Fig.~\ref{fig:landscape}(b) scatter around this common curve, with bands that widen as $N$ grows.
Likewise, the estimated minimizers $\phi^\star$ of Eq.~\eqref{eq:cut-optimum} scatter around $\phi_0$ more strongly at larger sizes (Fig.~\ref{fig:landscape}(c)).

Scanning the bridge angle is therefore free, but it cannot improve the bridge, already at $N=34$ and increasingly at larger sizes, due to statistical fluctuations in the angle-dependent coefficients $\widehat e_1$ and $\widehat e_2$.
The bridge changes the energy by an amount that does not grow with the size, and which the extensive noise of the full energy hides.
This is why Sec.~\ref{subsec:app-knit} measures the contact correlation, which isolates the interface, rather than relying on the energy.

\subsection{Seam entropy}
\label{app:cut-seam}

For a single severed gate, the seam entropy follows exactly from the states of the two patches just before the bridge.

\begin{proposition}[Seam entropy of a single bridge]
\label{prop:single-bridge}
Let $|u_A\rangle|u_B\rangle$ be a product state, let
$z_a=\langle u_A|Z_a|u_A\rangle$ and
$z_b=\langle u_B|Z_b|u_B\rangle$, and let
\begin{equation}
  |\psi\rangle
  =
  (W_A\otimes W_B)
  R_{ZZ}^{ab}(\phi)
  |u_A\rangle|u_B\rangle
\end{equation}
for local unitaries $W_A$ and $W_B$. Then the at most two nonzero eigenvalues of
$\rho_A$ satisfy
\begin{equation}
  \lambda_++\lambda_-=1,
  \qquad
  \lambda_+\lambda_-
  =
  c^2s^2(1-z_a^2)(1-z_b^2),
  \label{eq:cut-schmidt}
\end{equation}
and
\begin{equation}
  \Gsep(\psi)
  =
  2\bigl(\Tr\sqrt{\rho_A}\bigr)^2-1
  =
  1+2|\sin\phi|\sqrt{(1-z_a^2)(1-z_b^2)}
  \le
  \Gamma(\phi).
  \label{eq:cut-state-floor}
\end{equation}
Equivalently, the seam entropy is $\Ehalf=2\log_2\Tr\sqrt{\rho_A}$.
\end{proposition}

\begin{proof}
Local unitaries $W_A$ and $W_B$ do not change the spectrum of $\rho_A$, so
we set them to the identity. For each patch,
\begin{equation}
  Z_{b_p}|u_p\rangle
  =
  z_p|u_p\rangle
  +
  \sqrt{1-z_p^2}\,|u_p^\perp\rangle,
\end{equation}
where $|u_p^\perp\rangle$ is normalized and orthogonal to $|u_p\rangle$.
In the orthonormal product basis built from $|u_p\rangle$ and
$|u_p^\perp\rangle$, the coefficient matrix of $|\psi\rangle$ is
\begin{equation}
  C=
  \begin{pmatrix}
    c-isz_az_b & -isz_a\sqrt{1-z_b^2}\\
    -isz_b\sqrt{1-z_a^2}
    &
    -is\sqrt{(1-z_a^2)(1-z_b^2)}
  \end{pmatrix},
\end{equation}
with determinant
\begin{equation}
  \det C
  =
  -ics\sqrt{(1-z_a^2)(1-z_b^2)} .
\end{equation}
The nonzero eigenvalues of $\rho_A=CC^\dagger$ have sum
$\Tr CC^\dagger=1$ and product $|\det C|^2$, giving
Eq.~\eqref{eq:cut-schmidt}. Since
\begin{equation}
  \bigl(\sqrt{\lambda_+}+\sqrt{\lambda_-}\bigr)^2
  =
  1+2\sqrt{\lambda_+\lambda_-}
\end{equation}
and $2|cs|=|\sin\phi|$, Eq.~\eqref{eq:cut-state-floor} follows from the
pure-state formula
$\Gsep(\psi)=2(\Tr\sqrt{\rho_A})^2-1$~\cite{VidalTarrach1999}.
The result extends continuously when one of the two local spans becomes
one-dimensional.
\end{proof}

In the simulations, $z_a$ and $z_b$ are estimated by sampling prefix circuits, which run the encoder and the prefix sweeps of one patch and measure its bridge spin.
For the ground state, any state of Schmidt rank two across the seam satisfies
\begin{equation}
  |\langle\Omega|\psi\rangle|^2
  \le
  \lambda^{(0)}_1+\lambda^{(0)}_2 ,
  \label{eq:cut-rank-two-fidelity}
\end{equation}
where $\lambda^{(0)}_1\ge\lambda^{(0)}_2$ are the two largest Schmidt weights of the exact ground state $|\Omega\rangle$, because such a state occupies a two-dimensional subspace of patch $A$.
For $N\le24$, the weight beyond these two stays below $10^{-6}$ at every field, so a single bridge does not limit the fidelity.

\section{Gradient-variance diagnostic}
\label{app:variance}

The trainability guarantee for FLDCs~\cite{ZhangLiuZhang2024FLDC} assumes a circuit of local blocks with bounded support $\beta$, each spin acted on by a bounded number $\chi$ of blocks, and every block drawn independently from a local 2-design.
For a cost made of local terms, it then bounds the gradient variance from below independently of the system size.
However, the gradient variance of a single local observable still decays exponentially with its causal distance $\Delta k$ to the parameter, the number of blocks on the shortest chain that carries the effect of the parameter forward to the observable (Fig.~\ref{fig:variance-locality}(a)), but not with the total size.
Our circuit has bounded support, $\beta=2$, bounded local depth (Sec.~\ref{subsec:app-model}), and a cost made of local terms, but its pair blocks, built from fixed rotations with independent angles, are not drawn from a 2-design.
This appendix therefore checks the two consequences directly, that a single observable in our circuit follows the same trend, and additionally confirms that the total energy, the cost we actually train, does not decay with $\Delta k$ at all.

\textit{Setup.}
We use strips with $N=8$ to 38 and compact patches with $N=24$ to 48.
For each random draw of the angles, the exact evaluator of Sec.~\ref{subsec:app-sim} returns the full gradient in one reverse-mode pass.
The number of draws falls from 2{,}000 at $N\le20$ to 64 at $N=48$, because each draw requires one exact gradient, whose cost grows exponentially with the patch size.
The variance of each partial derivative is taken over the draws.
All runs use $h=0.25$ and $t=0.5$, and we consider the total energy $E$ rather than $E/N$, as in the main text, to align with Ref.~\cite{ZhangLiuZhang2024FLDC}.
We draw every angle independently and uniformly from $[0,2\pi)$, the natural random initialization of our circuit.
We additionally surround each probed $R_{ZZ}$ rotation with Haar-random two-qubit unitaries, to mimic the FLDC local 2-design assumption.
The variance then falls faster, from $7.6\times10^{-2}$ at $N=8$ to $1.2\times10^{-3}$ at $N=26$, indicating that the weak size dependence we report is therefore specific to the natural initialization of our circuit, not a property of random blocks in general.

\figVariance

\textit{A single observable.}
We first take the weight-two star $A_{\mathrm{far}}$ of patch $A$ farthest from the contact.
Its causal distance grows by three with every added plaquette, from 10 at $N=8$ to 25 at $N=38$.
At fixed $\Delta k$, its gradient variance does not depend on the length of the strip, once the strip is longer than one plaquette (Fig.~\ref{fig:variance-locality}(b)).
Along $\Delta k$, it falls by a factor of about $0.28$ per block at $N=32$ (Fig.~\ref{fig:variance-locality}(c)), which is the same behavior found for FLDCs in Ref.~\cite{ZhangLiuZhang2024FLDC}.

\textit{The total energy.}
The energy gradient is a sum over local terms, $\partial_\mu E=\sum_a\partial_\mu\langle A_a\rangle$, and every parameter has some term within eight blocks on strips and nine on compact patches.
Unlike that of the single star, its variance shows no decay with $\Delta k$ (Fig.~\ref{fig:variance-locality}(c)).
We attribute this to the nearby terms dominating the sum, so that the energy does not inherit the long-distance decay of a single far observable.

\tableVariancePosition 

\textit{Size dependence.}
The average energy-gradient variance decreases with size, but not exponentially, because the decrease comes from the open ends of the patches, which have limited impact at larger sizes.
We locate each gate by its distance, in plaquettes, from the closer end within a strip.
Table~\ref{tab:variance-position} lists the variance averaged over the gates at each such distance, with one column per strip length.
Along a row, the variance of a gate depends on its position and not on the length of the strip, assuming a long enough strip to contain that distance.
Down a column, it falls from about $1.8\times10^{-2}$ at the ends to about $2.7\times10^{-3}$ two and a half plaquettes inside.
A longer strip therefore only adds gates with this bulk value, and its average levels off instead of decaying exponentially.
Compact patches, however, are bounded on all sides.
Their variance also falls away from the boundary, but it has not reached a bulk value at the sizes we can simulate, so whether it levels off remains open (Fig.~\ref{fig:tier2}(b)).
For contrast, a control circuit that is identical except that its number of sweeps grows with the linear size of the patch has local depth 13 at $N=8$ and 82 at $N=32$.
On strips from $N=8$ to 32, a fit of $\ln\operatorname{Var}=c-rN$ gives $r=0.233\pm0.004$ per spin, so its variance falls by about a fifth with every added spin.

\section{Hardware circuits and settings}
\label{app:hardware}

This appendix gives the details of the device runs that are not part of the knitting protocol, namely the circuit sizes and trained angles, the compilation and error suppression, and the empirical correction calibrated at $h=0$ together with the depth dependence of the signal it measures.
 
\subsection{Circuit coordinates and trained angles}
\label{app:optimization}
 
The angles of a circuit are stored as one vector, ordered as the prefix sweeps of patches $A$ and $B$, the bridge angle $\phi$ and the anchor angle $\alpha$, and the suffix sweeps of patches $A$ and $B$.
With $s_{\mathrm{pre}}$ prefix and $s_{\mathrm{post}}$ suffix sweeps, its length is
\begin{equation}
  N_{\mathrm p}=40\,n_xn_y\,(s_{\mathrm{pre}}+s_{\mathrm{post}})+2 ,
  \label{eq:app-params}
\end{equation}
and we use $s_{\mathrm{pre}}=s_{\mathrm{post}}=2$.
Table~\ref{tab:app-grid} lists every patch geometry, with the two-qubit gate count and depth of the larger subcircuit, maximized over its branch circuits.

\tableGrid

The circuit structure is the same at every field, so we train by continuation in $h$.
Starting from the exact zero-field angles and from an optimized $h=1$ state, L-BFGS with exact gradients walks the field grid, starting each field from the previous one and keeping the lower energy.
Up to $N=34$, with three to five seeds, the best vector lies within $2\times10^{-3}$ per spin of the reference at every field.
At $N=48$, two-hour runs allow only about 110 gradient evaluations for 1{,}442 parameters and barely improve on the zero-field angles.
The first hardware series of Sec.~\ref{subsec:app-knit} uses the best trained vectors.
 
\subsection{Compilation and error suppression}
\label{app:device}
 
All device data are taken on \texttt{ibm\_phoenix}.
Every branch circuit is compiled with Qiskit~\cite{JavadiAbhari2024Qiskit} against the current backend target, using the SWAP-based heuristic search SABRE~\cite{li2019sabre} at optimization level~3.
The circuits of each size are transpiled once with all rotation angles left as parameters, so the angles enter only as virtual $Z$ rotations and the physical circuit is the same at every field.
Idle qubits are protected by dynamical decoupling with pairs of $X$ pulses of opposite phase~\cite{Viola1999DD}, and every two-qubit gate is Pauli twirled~\cite{Wallman2016RC}.
 
Every shot is post-selected with the inversion test of Ref.~\cite[App.~H]{DragoiEgger2026Transpilation}, implemented with \texttt{qiskit-addon-utils}~\cite{Rossmannek2024QiskitAddonUtils}.
After the final measurement, each checked qubit is flipped by an $X$ pulse and measured again, and a shot is kept only if the two outcomes disagree on every checked qubit, which removes shots with readout or relaxation errors at the end of the circuit.
The checked qubits are those of the subcircuit and the idle device qubits next to them.
Increased error rates due to larger sizes and more checks make the median kept fraction of shots fall from $0.83$ at $N=8$ to $0.04$ at $N=142$.
Fewer kept outcomes therefore enter each estimate at larger sizes, leading to increased statistical errors.
 
Each point of a sweep uses $M=24{,}576$ shots for the 24 branch settings and $2{,}048$ shots for one prefix circuit per patch.
The sixteen series shown, about $3\times10^6$ shots in total, used 317 seconds of charged processor time, about twice the execution time because of twirling and per-circuit overheads.

\subsection{Signal loss and its depth dependence}
\label{app:depth}

The dominant effect of the device noise is a damping of expectation values toward zero.
At $h=0$, every submitted angle vector prepares the exact zero-field state, whose energy $E_{\mathrm{exact}}(0)$ is known (App.~\ref{app:endpoints}), so the fraction of the signal that a series $s$ retains is
\begin{equation}
  \lambda_s=\frac{E_s(0)}{E_{\mathrm{exact}}(0)} .
\end{equation}
A series is characterized by one size, one setting of the error suppression, and one choice of angles.
Figure~\ref{fig:signal} shows the retained signal of every series.
It falls from $0.82$ at $N=8$ to $0.53$ at $N=34$ for the trained angles (blue), and from $0.44$ at $N=48$ to $0.039$ at $N=142$ for the zero-field angles (black), with the sharpest drop between the $3\times4$ and $4\times4$ patches.
Two additional series (purple) with transferred and refined angles at $N=48$ and 62 retain $0.38$ and $0.35$, like the zero-field circuits of the same sizes.

\figSignalFraction

The retained signal follows the routed two-qubit depth $D$ of the larger subcircuit (Table~\ref{tab:app-grid}) consistent with decoherence of idle qubits during the sequential sweeps.
A weighted fit of $-\ln\lambda_s$ through the origin, over the fifteen series with $\lambda_s>0.02$, gives
\begin{equation}
  \lambda_s\simeq e^{-cD_s},
  \qquad
  c=(0.61\pm0.03)\%\ \text{per two-qubit layer},
\end{equation}
so every layer of two-qubit gates removes about $0.6\%$ of the signal.

\figRescaling
\subsection{Correction calibrated at zero field}
\label{app:rescaling}

Since the $h=0$ circuit has a known ideal energy, it can serve as a calibration circuit for the other fields of the same series.
Dividing the energies of a series $s$ by its retained signal $\lambda_s$ measured at $h=0$ (App.~\ref{app:depth}) gives corrected energies $\widehat E_s(h)=E_s(h)/\lambda_s$, with the error of $\lambda_s$ propagated, where we omit series with $\lambda_s<0.005$.
This rescaling intuition rests on two facts.
The physical circuit is the same at every field (App.~\ref{app:device}), and twirling makes the gate noise a Pauli channel on average.
If that channel acted only after the last gate, each energy term would be damped by a factor independent of the angles, and the factor measured at $h=0$ would correct the whole series.
In practice, errors occur throughout the circuit and remain Pauli errors under propagation only in Clifford circuits, as at $h=0$.
The calibration also uses only star and plaquette terms, since the field terms vanish at $h=0$, whereas single-spin terms are usually damped less.
The rescaling is therefore a one-parameter approximation, not an error-mitigation guarantee.
It cannot remove preparation error, and it multiplies statistical errors by $1/\lambda_s$.

For the trained angles at $N=8$ to 34, the ideal energy is close to the ground energy, so the correction can be tested directly (Fig.~\ref{fig:rescaling}(a,b)).
It removes about 86\% of the raw gap on average and brings the energies within $0.13$ per spin of the reference for $0<h<1$.
At $h=1$, where single-spin terms dominate, it overcompensates by about $0.15$ per spin, below the reference.
For the zero-field angles, the corrected energies should reproduce Eq.~\eqref{eq:ep-zf} at every field (Fig.~\ref{fig:rescaling}(c)).
They do so within $2.3$ standard errors up to $N=120$, deviate by up to four at $N=142$.

\end{document}